\newif\ifcomments  %
\newif\ifsupp  %
\newif\ifhidefigure  %
\commentstrue
\supptrue
\ifhidefigure
\PassOptionsToPackage{draft}{graphicx}
\fi

\documentclass{article}

\usepackage{amsmath,amssymb,pifont}
\usepackage{amstext}
\usepackage{amsthm}
\usepackage{multirow}
\usepackage{booktabs}
\usepackage[skip=0pt]{subcaption}
\usepackage{times}
\usepackage{lipsum}
\usepackage[shortlabels]{enumitem}
\usepackage{comment}
\usepackage{cancel}
\usepackage{wrapfig}
\usepackage{dblfloatfix}
\usepackage{array}
\usepackage{siunitx}
\usepackage{csvsimple}
\usepackage[multidot]{grffile}
\usepackage{makecell}
\usepackage{dsfont}
\usepackage{mathtools}
\usepackage{xcolor}
\usepackage{comment}
\usepackage[multiple]{footmisc}
\usepackage{mathrsfs}
\usepackage{todonotes}
\usepackage{tikz}
\usepackage{xspace}

\usepackage{booktabs}
\usepackage{makecell}
\usepackage{nicefrac}
\usepackage{hyperref}
\usepackage[capitalise,noabbrev]{cleveref}
\usepackage{bbold}
\newcommand{\Restart}{DP-FTRL-TreeRestart\xspace}
\newcommand{\NoRestart}{DP-FTRL-NoTreeRestart\xspace}
\newcommand{\SomeRestart}{DP-FTRL-SometimesRestart\xspace}

\newcommand{\eps}{\ensuremath{\varepsilon}}

\newcommand{\bolda}{\ensuremath{\boldsymbol{a}}}
\newcommand{\boldb}{\ensuremath{\boldsymbol{b}}}

\newcommand{\boldr}{\ensuremath{\boldsymbol{r}}}
\newcommand{\bolds}{\ensuremath{\boldsymbol{s}}}

\newcommand{\boldv}{\ensuremath{\boldsymbol{v}}}

\newcommand{\boldx}{\bfx}

\newcommand{\boldR}{\mathbb{R}}

\newcommand{\bfS}{\ensuremath{\mathbf{S}}}

\newcommand{\bfb}{\ensuremath{\mathbf{b}}}

\newcommand{\bfx}{\ensuremath{\mathbf{x}}}

\newcommand{\calA}{\ensuremath{\mathcal{A}}}
\newcommand{\calB}{\ensuremath{\mathcal{B}}}
\newcommand{\calC}{\ensuremath{\mathcal{C}}}
\newcommand{\calD}{\ensuremath{\mathcal{D}}}

\newcommand{\calL}{\ensuremath{\mathcal{L}}}

\newcommand{\calN}{\ensuremath{\mathcal{N}}}

\newcommand{\calS}{\ensuremath{\mathcal{S}}}
\newcommand{\calT}{\ensuremath{\mathcal{T}}}

\renewcommand{\Pr}{\mathop{\mathbf{Pr}}}

\newcommand{\E}{\mathop{\mathbf{E}}}

\newcommand{\R}{\mathbb{R}}

\newtheorem{lem}{Lemma}[section]

\newtheorem{thm}[lem]{Theorem}
\newtheorem{cor}[lem]{Corollary}

\newtheorem{defn}[lem]{Definition}

\DeclareMathOperator*{\argmin}{arg\,min}

\makeatletter
\newcommand{\vast}{\bBigg@{4}}
\newcommand{\Vast}{\bBigg@{5}}
\makeatother

\newcommand{\power}[2]{\left(#1\right)^{#2}}
\newcommand{\ip}[2]{\langle #1, #2\rangle}

\newcommand{\privT}{\theta^\texttt{priv}\,}

\newcommand{\ltwo}[1]{\left\|#1\right\|_2}

\newcommand{\norm}[1]{\| #1 \|}

\DeclarePairedDelimiterX{\infdivx}[2]{(}{)}{%
  #1\;\delimsize\|\;#2%
}

\newcommand*\samethanks[1][\value{footnote}]{\footnotemark[#1]}

\newcommand{\mypar}[1]{\smallskip
	\noindent{\textbf{{#1}:}}}
	
\renewcommand{\epsilon}{\varepsilon}

\renewcommand{\tilde}{\widetilde}
\newcommand{\aftrl}{\calA_{\sf FTRL}}
\newcommand{\aftrlls}{\calA_\textsf{FTRL-LS}}

\newcommand{\poprisk}[1]{{\sf PopRisk}(#1)\,}
\newcommand{\tree}{\calT}
\newcommand{\treeb}{\calT_{\sf bias}}
\newcommand{\treeq}{\calT_{\sf cov}}
\newcommand{\init}{{\texttt{InitializeTree}}\,}
\newcommand{\addt}{{\texttt{AddToTree}}\,}
\newcommand{\gett}{{\texttt{GetSum}}\,}
\newcommand{\gettrv}{{\texttt{GetSumReducedVariance}}\,}
\newcommand{\initb}{{\texttt{InitializeTreeBias}}\,}
\newcommand{\addtb}{{\texttt{AddToTreeBias}}\,}
\newcommand{\gettb}{{\texttt{GetSumBias}}\,}
\newcommand{\outputt}{{\bf Output}\,}
\newcommand{\initq}{{\texttt{InitializeTreeCov}}\,}
\newcommand{\addtq}{{\texttt{AddToTreeCov}}\,}
\newcommand{\gettq}{{\texttt{GetSumCov}}\,}
\newcommand{\node}{\texttt{node}}

\newcommand{\npJ}{J_t^{\sf np}}

\newcommand{\clip}[2]{{\sf clip}\left(#1,#2\right)}
\newcommand{\pJ}{J_t^{\sf priv}}

\newcommand{\thetash}{\widetilde{\theta}}

\newcommand{\pl}{{\sf polylog}}
\newcommand{\plnd}{\pl\left(\frac{1}{\delta},n\right)}
\newcommand{\plnh}{\pl\left(\frac{1}{\delta},n,\frac{1}{\beta}\right)}

\newcommand{\npsgd}{\textsf{Noisy-SGD}}
\newcommand{\dpftrl}{\textsf{DP-FTRL}}

\newcommand{\remove}[1]{}
\newcommand{\bigO}[1]{O\left(#1\right)}
\renewcommand{\bfb}{\boldb}

\newcommand{\leftn}{\texttt{left}}
\newcommand{\rightn}{\texttt{right}}

\newcommand{\epochs}{E}

\newcommand{\tr}{\calT}
\newcommand{\leaf}{\calN}

\newcommand{\level}{\gamma}
\newcommand{\participation}{\xi}

\newcommand{\contrib}{\texttt{contrib}}
\newcommand{\start}{\texttt{start}}
\newcommand{\en}{\texttt{end}}
\newcommand{\size}{\texttt{size}}
\newcommand{\nul}{\bot}

\usepackage{fullpage}
\usepackage[numbers, sort, comma, square]{natbib}
\usepackage{algorithm,algorithmic}

\begin{document}

\title{Practical and Private (Deep) Learning Without Sampling or Shuffling}
\author{Peter Kairouz\thanks{Google. \texttt{\{kairouz, mcmahan, shuangsong, omthkkr, athakurta, xuzheng\}@google.com}} \and Brendan McMahan\samethanks[1] \and Shuang Song\samethanks[1] \and Om Thakkar\samethanks[1] \and Abhradeep Thakurta\samethanks[1] \and Zheng Xu\samethanks[1]}
\maketitle

\begin{abstract}
We consider training models with differential privacy (DP) using mini-batch gradients.
The existing state-of-the-art, Differentially Private Stochastic Gradient Descent (DP-SGD), requires  \emph{privacy amplification by sampling or shuffling} to obtain the best privacy/accuracy/computation trade-offs. Unfortunately, the precise requirements on exact sampling and shuffling can be hard to obtain in important practical scenarios, particularly federated learning (FL).
We design and analyze a DP variant of  Follow-The-Regularized-Leader (DP-FTRL) that compares favorably (both theoretically and empirically) to amplified DP-SGD, while allowing for much more flexible data access patterns. DP-FTRL does not use any form of privacy amplification.
\end{abstract}

\section{Introduction}
\label{sec:intro}

Differentially private stochastic gradient descent (DP-SGD)~\cite{song2013stochastic,BST14,DP-DL} has become state-of-the-art in training private (deep) learning models~\cite{DP-DL,mcmahan2018general,esa2,papernot2020tempered,opacus,TB21}. It operates by running stochastic gradient descent~\cite{robbins1951stochastic} on noisy mini-batch gradients\footnote{Gradient computed on a subset of the training examples, also called a mini-batch.}, with the noise calibrated such that it ensures differential privacy. The privacy analysis heavily uses tools like \emph{privacy amplification by sampling/shuffling}~\cite{KLNRS,BST14,DP-DL,wang2019subsampled,zhu2019poission,soda-shuffling,feldman2020hiding} to obtain the best privacy/utility trade-offs. Such amplification tools require that each mini-batch is a perfectly (uniformly) random subset of the training data. This assumption can make practical deployment prohibitively hard, especially in the context of distributed settings like federated learning (FL) where one has little control on which subset of the training data one sees at any time~\cite{kairouz2019advances,balle2020privacy}. 

We propose a new online learning~\cite{hazan2019introduction,shalev2011online} based DP algorithm, \emph{differentially private follow-the-regularized-leader} (DP-FTRL), that has privacy/utility/computation trade-offs that are competitive with DP-SGD, and does not rely on privacy amplification. DP-FTRL \emph{significantly outperforms} un-amplified DP-SGD at all privacy levels. In the higher-accuracy / lower-privacy regime, DP-FTRL outperforms even \emph{amplified} DP-SGD. We emphasize that in the context of ML applications, using a DP mechanism even with a large $\epsilon$ is practically much better for privacy than using a non-DP mechanism~\cite{song2019auditing,jagielski2020auditing,thakkar2020understanding,nasr2021adversary}. 

\mypar{Privacy amplification and its perils}  At a high-level, DP-SGD can be thought of as an iterative noisy state update procedure for $T$ steps operating over mini-batches of the training data. 
For a time step $t \in [T]$ and an arbitrary mini-batch of size $k$ from a data set $D$ of size $n$, let $\sigma_t$ be the standard deviation of the noise needed in the $t^{th}$ update to satisfy $\epsilon_t$-differential privacy. 
If the mini-batch is chosen \emph{u.a.r. and i.i.d.} from $D$ at each time step\footnote{One can also create a mini-batch with Poisson sampling~\cite{DP-DL,mcmahan2017learning,zhu2019poission}, except the batch size is now a random variable. For brevity, we focus on the fixed batch setting.}  $t$, then privacy amplification by sampling~\cite{KLNRS,BST14,DP-DL,wang2019subsampled} allows one to scale down the noise to $\sigma_t\cdot(k/n)$, while still ensuring $\epsilon_t$-differential privacy.\footnote{A similar argument holds for amplification by shuffling~\cite{soda-shuffling,feldman2020hiding}, when the data are uniformly shuffled at the beginning of every epoch.We do not consider privacy amplification by iteration~\cite{FMTT18} in this paper, as it only applies to smooth convex functions.}
Such amplification is crucial for DP-SGD to obtain state-of-the-art models in practice~\cite{DP-DL,papernot2020tempered,TB21} when $k\ll n$.

There are two major bottlenecks for such deployments: 
i) For large data sets, achieving uniform sampling/shuffling of the mini-batches in every round (or epoch) can be prohibitively expensive in terms of computation and/or engineering complexity, 
ii) In distributed settings like federated learning (FL)~\cite{FL1}, uniform sampling/shuffling may be infeasible to achieve because of widely varying available population at each time step. 
Our work answers the following question in affirmative: \emph{Can we design an algorithm that does not rely on privacy amplification, and hence allows data to be accessed in an arbitrary order, while providing privacy/utility/computation trade-offs competitive with DP-SGD?}

\mypar{DP-FTRL and amplification-free model training} DP-FTRL can be viewed as a differentially private variant of the follow-the-regularized-leader (FTRL) algorithm~\cite{xiao2010dual,mcmahan2011follow,duchi2011adaptive}. The main idea in DP-FTRL is to use the \emph{tree aggregation trick}~\cite{Dwork-continual,CSS11-continual} to add noise to the sum of mini-batch gradients, in order to ensure privacy. Crucially, it deviates from DP-SGD by adding correlated noise across time steps, as opposed to independent noise. This particular aspect of DP-FTRL allows it to get strong privacy/utility trade-off without relying on privacy amplification.

\mypar{Federated Learning (FL) and DP-FTRL} There has been prior work \cite{balle2020privacy, ramaswamy2020training} detailing challenges for obtaining strong privacy guarantees that incorporate limited availability of participating clients in real-world applications of federated learning. Although there exist techniques like the Random Check-Ins \cite{balle2020privacy} that obtain privacy amplification for FL settings, implementing such techniques may still require clients to keep track of the number of training rounds being completed at the server during their period(s) of availability to be able to uniformly randomize their participation.
On the other hand, since the privacy guarantees of DP-FTRL (Algorithm~\ref{Alg:PFTRL}) do not depend on any type of privacy amplification, it does not require any local/central randomness apart from noise addition to the model updates.

Appendices~\ref{sec:related} and Section ~\ref{sec:background} describe additional related work and background, respectively.

\subsection{Problem Formulation}
\label{sec:problemDesc}

Suppose we have a stream of data samples $D=[d_1,\ldots,d_n]\in\calD^n$, where $\calD$ is the domain of data samples, and a loss function $\ell:\calC\times\calD\to\boldR$, where $\calC\in\boldR^p$ is the space of all models. We consider the following two problem settings.

\mypar{Regret Minimization} At every time step $t\in[n]$, while observing samples $[d_1,\ldots,d_{t-1}]$, the algorithm $\calA$ outputs a model $\theta_t\in\calC$ which is used to predict on example $d_t$. The performance of $\calA$ is measured in terms of regret against an arbitrary post-hoc comparator $\theta^* \in \calC$:
\begin{equation}
R_D(\calA;\theta^*)=\frac{1}{n}\sum\limits_{t=1}^n\ell(\theta_t;d_t)-\frac{1}{n}\sum\limits_{t=1}^n\ell(\theta^*;d_t).
\label{eq:reg}
\end{equation}
We consider the algorithm $\calA$ low-regret if $R_D(\calA;\theta^*)=o(1)$. To ensure a low-regret algorithm, we will  assume $\ltwo{\nabla\ell(\theta;d)}\leq L$ for any data sample $d$, and any models $\theta\in\calC$. We consider both \emph{adversarial regret}, where the data sample $d_t$ are drawn adversarially based on the past output $\{\theta_1,\ldots,\theta_t\}$~\cite{hazan2019introduction}, and \emph{stochastic regret}~\cite{hazan2014beyond}, where the data samples in $D$ are drawn i.i.d. from some fixed distribution $\tau$. 

\mypar{Excess Risk Minimization} In this setting, we look at the problem of minimizing the excess population risk. Assuming the data set $D$ is sampled i.i.d. from a distribution $\tau$, and the algorithm $\calA$ outputs $\widehat\theta\in\calC$, we want to minimize
\begin{equation}
\poprisk{\calA}=\mathbb{E}_{d\sim\tau}\ell(\widehat\theta;d)-\min\limits_{\theta\in\calC}\mathbb{E}_{d\sim\tau}\ell(\theta;d).
\label{eq:poprisk}
\end{equation}
All the algorithms in this paper guarantee differential privacy~\cite{DMNS,ODO} and R\'enyi differential privacy~\cite{mironov2017renyi} (See Section~\ref{sec:background} for details). The definition of a single data record can be one training example (a.k.a., \emph{example level} privacy), or a group of training examples from one individual (a.k.a., \emph{user level} privacy). Except for the empirical evaluations in the FL setting, we focus on example level privacy. The specific definition of differential privacy (DP) we use is in Definition~\ref{def:diiffP}, which is semantically similar to the traditional add/remove notion of DP~\cite{vadhan2017complexity}, where two data sets are neighbors if their symmetric difference is one. In particular, it is a special instantiation of~\cite[Definition II.3]{esa++}. 
An advantage of Definition~\ref{def:diiffP} is that it allows capturing the notion of neighborhood defined  by addition/removal of single data record in the data set, without having the necessity to change the number of records of the data set ($n$).
Since the algorithms in this paper are motivated from natural streaming/online algorithms, Definition~\ref{def:diiffP} is convenient to operate with.
  Similar to the traditional add/remove notion~\cite{vadhan2017complexity}, Definition~\ref{def:diiffP} can capture the regular replacement version of DP (originally defined in~\cite{DMNS}, where the notion of neighborhood is defined by replacing any data record with its worst-case alternative), by incurring up to a factor of two in the privacy parameters $\epsilon$ and $\delta$.

\begin{defn}[Differential privacy] 

Let $\calD$ be the domain of data records,  $\nul\not\in\calD$ be a special element, and let $\widehat{\calD}=\calD\cup\{\nul\}$ be the extended domain. 
A randomized algorithm $\calA:\widehat{\calD}^{n}\to\calS$ is $(\eps,\delta)$-differentially private if for any data set $D\in\widehat{\calD}^n$ and any neighbor  $D'\in\widehat{\calD}^n$ (formed from $D$ by replacing one record with $\nul$),
and for any event $S\in\calS$, we have 
\begin{align*}
    \Pr[\calA(D)\in S] &\leq e^{\eps} \cdot \Pr[\calA(D')\in S] +\delta,\ \ \text{and}\\
    \Pr[\calA(D')\in S] &\leq e^{\eps} \cdot \Pr[\calA(D)\in S] +\delta,
\end{align*}
where the probability is over the randomness of $\calA$.
\label{def:diiffP}
\end{defn}
In Algorithm~\ref{Alg:PFTRL} we treat $\nul$ specially, namely assuming it always produces a zero gradient.

\subsection{Our Contributions}
\label{sec:contrib}

\begin{table*}[ht]
\caption{\label{tab:tbl1} Best known regret guarantees for dataset size $n$ and model dimension $p$. Here, high probability means w.p. at least $1-\beta$ over the the randomness of the algorithm. The expected regret is an expectation over the random choice of the data set and the randomness of the algorithm.
} %

\renewcommand{\arraystretch}{1.6}
\begin{center}
\begin{tabular}{|*{5}{c|}}
\hline
Class &\multicolumn{2}{c|}{Adversarial Regret} &\multicolumn{2}{c|}{Stochastic Regret} \\
\hline
& Expected & High probability & Expected & High probability \\
\hline\hline
\makecell{Least-squares\\(and linear)} 
&\makecell{$O\left(\left(\frac{1}{\sqrt n}+\frac{\sqrt{p}}{\epsilon n}\right)\cdot\right.$\\$\left.\plnd\right)$\\\cite{agarwal2017price}}
&\makecell{Same as\\general\\convex}
&\makecell{$O\left(\left(\frac{1}{\sqrt n}+\frac{\sqrt{p}}{\epsilon n}\right)\cdot\right.$ \\$\left.\plnd\right)$\\\cite{agarwal2017price}}
&\makecell{$O\left(\left(\frac{1}{\sqrt n}+\frac{\sqrt{p}}{\epsilon n}\right)\cdot\right.$ \\$\left.\plnh\right)$\\\color{blue} [Theorem~\ref{thm:stocReg-ls-exist}]}\\
\hline
General convex & \multicolumn{4}{c|}{Constrained and unconstrained: $ O\left(\left(\frac{1}{\sqrt n}+\frac{p^{1/4}}{\sqrt{\epsilon n}}\right)\cdot\plnh\right)$ {\color{blue} [Theorem~\ref{thm:regFTRL}]}}\\
\hline
\end{tabular}
\end{center}
\end{table*}

Our primary contribution in this paper is a private online learning algorithm: differentially private follow-the-regularized leader (DP-FTRL) (Algorithm~\ref{Alg:PFTRL}). We provide tighter privacy/utility trade-offs based on DP-FTRL (see Table~\ref{tab:tbl1} for a summary), and show how it can be easily adapted to train (federated) deep learning models, with comparable, and sometimes even better privacy/utility/computation trade-offs as DP-SGD. We summarize these contributions below.

\mypar{DP-FTRL algorithm} We provide DP-FTRL, a differentially private variant of the Follow-the-regularized-leader (FTRL) algorithm~\cite{mcmahan10boundopt,mcmahan2011follow,shalev2011online,hazan2019introduction} for online convex optimization (OCO). We also provide a variant called the momentum DP-FTRL that has superior performance in practice. \cite{agarwal2017price} provided a instantiation of DP-FTRL specific to linear losses. \cite{thakurta2013nearly} provided an algorithm similar to DP-FTRL, where instead of just linearizing the loss, a quadratic approximation to the regularized loss was used.

\mypar{Regret guarantees} In the adversarial OCO setting (Section~\ref{sec:advReg}), compared to prior work~\cite{JKT-online,thakurta2013nearly,agarwal2017price}, DP-FTRL has the following major advantages. First, it improves the best known regret guarantee in~\cite{thakurta2013nearly} by a factor of $\sqrt{\epsilon}$ (from $\widetilde{O}\left(\sqrt{\frac{\sqrt{p}}{\epsilon^2 n}}\right)$ to $\widetilde{O}\left(\sqrt{\frac{\sqrt{p}}{\epsilon n}}\right)$, when $\epsilon\leq 1$). This improvement is significant because it \emph{distinguishes centrally private OCO from locally private}~\cite{Warner,evfimievski2003limiting,KLNRS} OCO\footnote{Although not stated formally in the literature, a simple argument shows that locally private SGD~\cite{DJW13} can achieve the same regret as in~\cite{thakurta2013nearly}.}. Second, unlike~\cite{thakurta2013nearly}, DP-FTRL (and its analysis) extends to the unconstrained setting $\calC=\mathbb{R}^p$. Also, in the case of composite losses~\cite{duchi2010composite,xiao2010dual,mcmahan2011follow,mcmahan17survey}, i.e., where the loss functions are of the form $\ell(\theta;d_t)+r_t(\theta)$ with $r:\calC\to\mathbb{R}^+$ (e.g., $\norm{\cdot}_1$) being a convex regularizer, DP-FTRL has a regret guarantee for the losses $\ell(\theta;d_t)$'s of form: (regret bound without the $r_t$'s) $+\frac{1}{n}\sum\limits_{t=1}^n r_t(\theta^*)$. 

In the stochastic OCO setting (Section~\ref{sec:lssq}), we show that for least-square losses (where $\ell(\theta;d_t)=(y_t-\ip{\boldx_t}{\theta})^2$ with $d_t=(\boldx_t,y_t)$) and linear losses (when $\ell(\theta;d_t)=\ip{d_t}{\theta}$),
a variant of DP-FTRL achieves regret of the form $O\left(\left(\frac{1}{\sqrt n}+\frac{\sqrt p}{\epsilon n}\right)\cdot\plnh\right)$ with probability $1-\beta$ over the randomness of algorithm. Our guarantees are strictly high-probability guarantees, i.e., the regret only depends on ${\sf polylog}(1/\beta)$. 

\mypar{Population risk guarantees} In Section~\ref{sec:popRisk},
using the standard online-to-batch conversion~\cite{cesa2002generalization,SSSS09}, we obtain a population risk guarantee for DP-FTRL. For general Lipschitz convex losses, the population risk for DP-FTRL in Theorem~\ref{thm:poprisk} is same as that in \citep[Appendix F]{BST14} (up to logarithmic factors), but the advantage of DP-FTRL is that it is a single pass algorithm (over the data set $D$), as opposed to requiring $n$ passes over the data. 
\emph{Thus, we provide the best known population risk guarantee for a single pass algorithm that does not rely on convexity for privacy.}
While the results in~\cite{bassily2019private,bassily2020stability,feldman2019private} have a tighter (and optimal) excess population risk of $\widetilde \Theta(1/\sqrt{n}+\sqrt{p}/(\epsilon n))$, they either require convexity to ensure privacy for a single pass algorithm, or need to make $n$-passes over the data. 
For restricted classes like linear and least-squared losses, DP-FTRL can achieve the optimal population risk via the tighter stochastic regret guarantee. Whether DP-FTRL can achieve the optimal excess population risk in the general convex setting 
is left as an open problem.

\mypar{Empirical contributions}
In Sections~\ref{sec:practical_variants} and~\ref{sec:empEval} we study some trade-offs between privacy/utility/computation for DP-FTRL and DP-SGD. 
We conduct our experiments on four benchmark data sets: MNIST, CIFAR-10, EMNIST, and StackOverflow.
We start by fixing the computation available to the techniques, and observing privacy/utility trade-offs. 
We find that DP-FTRL achieves better utility compared to DP-SGD for moderate to large $\epsilon$.
In scenarios where amplification cannot be ensured (e.g., due to practical/implementation constraints), DP-FTRL provides substantially better performance as compared to unamplified DP-SGD.
Moreover, we show that with a modest increase in the computation cost, DP-FTRL, without any need for amplification, can match the performance of amplified DP-SGD.
Next, we focus on privacy/computation trade-offs for both the techniques when a utility target is desired.
We show that DP-FTRL can provide better trade-offs compared to DP-SGD for various accuracy targets, which can result in significant savings in privacy/computation cost as the size of data sets becomes limited.

To shed light on the empirical efficacy of DP-FTRL (in comparison) to DP-SGD, in Section~\ref{sec:equivDP-GD}, we show that a variant of DP-SGD (with correlated noise) can be viewed as an equivalent formulation of DP-FTRL in the unconstrained setting (
$\calC=\mathbb{R}^p$). In the case of traditional DP-SGD~\cite{BST14}, the scale of the noise added per-step $t\in[n]$ is asymptotically same as that of DP-FTRL once $t=\omega(n)$.

{\color{blue}
\section{Errata and Fixes for the ICML 2021 version}
\label{sec:errata}

In this section we provide an errata for the ICML-2021 proceedings version~\cite{kairouz21b} of this paper. For the no-tree-restart case of DP-FTRL, the privacy accounting of Theorem D.2 in~\cite{kairouz21b} was erroneous, as it incorrectly computed the sensitivity of the complete binary tree. Specifically, the analysis did not take into account that if DP-FTRL is executed across multiple epochs of the training data, then a single user can both contribute to multiple leaf nodes in the tree, and can contribute more than once to a single non-leaf node. In Section~\ref{app:multipass} we provide corrected versions of the theorem, and also provide corrected empirical evaluation for MNIST, CIFAR-10, and EMNIST based on this (these were the only empirical results affected). These experiments are detailed in Section~\ref{sec:empEval}, and the corresponding appendices.  Qualitatively, when DP-FTRL and DP-SGD (the amplified version) are compared for large number of epochs of training, the crossover point (w.r.t. $\epsilon$) where DP-FTRL outperforms DP-SGD shifts to a larger value.
However, for small number of epochs of training, the crossover point remains unchanged.}
\section{Background}
\label{sec:background}

\mypar{Differential Privacy} Throughout the paper, we use the notion of approximate differential privacy~\cite{DMNS,ODO} and R\'enyi differential privacy (RDP)~\cite{DP-DL, mironov2017renyi}. For meaningful privacy guarantees, $\epsilon$ is assumed to be a small constant,  and $\delta \ll 1/|D|$.

\begin{defn}[R\'enyi differential privacy]
Analogous to the definitiion of $(\epsilon,\delta)$-differential privacy in Definition~\ref{def:diiffP}, a randomized algorithm $\calA$ is $(\alpha, \eps)$-RDP if the condition on $\calA(D)$ and $\calA(D')$ in Definition~\ref{def:diiffP} are replaced with the following:
\begin{align*}
    \frac{1}{\alpha-1} \log \E_{s\sim \calA(D)}{\power{\frac{\Pr\left[\calA(D) = s\right]}{\Pr\left[\calA(D') = s\right]}}{\alpha}} &\leq \epsilon,\ \ \text{and}\\
     \frac{1}{\alpha-1} \log \E_{s\sim \calA(D')}{\power{\frac{\Pr\left[\calA(D') = s\right]}{\Pr\left[\calA(D) = s\right]}}{\alpha}} &\leq \epsilon.
\end{align*}
\end{defn}

\citet{DP-DL} and~\citet{mironov2017renyi} have shown that an $(\alpha,\epsilon)$-RDP algorithm guarantees $\left(\epsilon + \frac{\log (1/\delta)}{\alpha-1}, \delta\right)$-differential privacy. Follow-up works~\cite{asoodeh2020better,canonne2020discrete} provide tighter conversions. We used the conversion in~\cite{canonne2020discrete} in our experiments.

To answer a query $f(D)$ with $\ell_2$ sensitivity $L$, i.e., $\max_{\text{neighboring }D,D'}\|f(D) - f(D')\|_2 \leq L$, the Gaussian mechanism~\cite{DMNS} returns $f(D) + \calN(0, L^2 \sigma^2)$, which guarantees $\left(\sqrt{1.25 \log(2/\delta)}/\sigma, \delta\right)$-differential privacy~\cite{DMNS,dwork2014algorithmic} and $(\alpha, \alpha/2\sigma^2)$-RDP~\cite{mironov2017renyi}.

\mypar{DP-SGD and Privacy Amplification} Differentially-private stochastic gradient descent (DP-SGD) is a common algorithm to solve private optimization problems. The basic idea is to enforce a bounded $\ell_2$ norm of individual gradient, and add Gaussian noise to the gradients used in SGD updates. 
Specifically, consider a dataset $D = \{d_1,\dots,d_n\}$ and an objective function of the form $\sum_{i=1}^n\ell(\theta; d_i)$ for some loss function $\ell$. DP-SGD uses an update rule
\begin{align*}
\theta_{t+1} \leftarrow \theta_t - \frac{\eta}{|\calB|} \left(\sum_{i\in\calB}\clip{\nabla_{\theta} \ell(\theta_t;d_i)}{L} + \calN(0, L^2 \sigma^2)\right)
\end{align*}
where $\clip{v}{L}$ projects $v$ to the $\ell_2$-ball of radius $L$, and $\calB \subseteq [n]$ represents a mini-batch of data.

Using the analysis of the Gaussian mechanism, we know that such an update step guarantees $(\alpha, \alpha/2\sigma^2)$-RDP with respect to the mini-batch $\calB$. By parallel composition, running one epoch with disjoint mini-batches guarantees $(\alpha, \alpha/2\sigma^2)$-RDP.
On the other hand, previous works~\cite{BST14,DP-DL,wang2019subsampled} has shown that if $\calB$ is chosen uniformally at random from $[n]$, or if we use poisson sampling to collect a batch of samples $\calB$, then one step would guarantee $\left(\alpha, \bigO{\alpha/2\sigma^2 \cdot (|\calB|/n)^2} \right)$-RDP.

\mypar{Tree-based Aggregation} 
Consider the problem of privately releasing prefix sum of a data stream, i.e., given a stream $D = (d_1, d_2, \dots, d_T)$ such that each $d_i \in \mathbb{R}^p$ has $\ell_2$ norm bounded by $L$, we aim to release $s_t=\sum_{i=1}^t d_i$ for all $t\in [1,T]$ under differential privacy. 
\citet{Dwork-continual} and \citet{CSS11-continual} proposed a tree-based aggregation algorithm to solve this problem. Consider a complete binary tree $\calT$ with leaf nodes as $d_1$ to $d_T$, and internal nodes as the sum of all leaf nodes in its subtree. To release the exact prefix sum $s_t$, we only need to sum up $\bigO{\log(t)}$ nodes.
To guarantee differential privacy for releasing the tree $\calT$, since any $d_i$ appears in $\log(T)$ nodes in $\calT$, using composition, we can add Gaussian noise of standard deviation of the order $L\sqrt{\log(T) \log(1/\delta)}/\epsilon$ to guarantee $(\epsilon, \delta)$-differential privacy.

\citet{thakurta2013nearly} used this aggregation algorithm to build a nearly optimal algorithms for private online learning. Importantly, this work showed the privacy guarantee holds even for \emph{adaptively chosen sequences} $\{d_t\}_{t=1}^T$, which is crucial for model training tasks.

\section{Private Follow-The-Regularized-Leader}
\label{sec:privateFTRL}

In this section, we provide the formal description of the DP-FTRL algorithm (Algorithm~\ref{Alg:PFTRL}) and its privacy analysis. We then show that a variant of differentially private stochastic gradient descent (DP-SGD)~\cite{song2013stochastic,BST14} can be viewed of as an instantiation of DP-FTRL under appropriate choice of learning rate. 

Critically, \emph{our privacy guarantees for DP-FTRL hold when the data $D$ are processed in an arbitrary (even adversarially chosen) order,} and do not depend on the convexity of the loss functions.  The utility guarantees, i.e., the regret and the excess risk guarantees require convex losses (i.e., $\ell(\cdot;\cdot)$ is convex in the first parameter). In the presentation below, we assume differentiable losses for brevity. The arguments extend to non-differentiable convex losses via standard use of sub-differentials~\cite{shalev2011online,hazan2019introduction}.

\subsection{Algorithm Description}
\label{sec:algoDesc}

The main idea of DP-FTRL is based on three observations: i) For online convex optimization, to bound the regret, for a given loss function $\ell(\theta;d_t)$ (i.e., the loss at time step $t$), it suffices for the algorithm to operate on a linearization of the loss at $\theta_t$ (the model output at time step $t$): $\tilde \ell(\theta;d_t)=\ip{\nabla_\theta\ell(\theta_t;d_t)}{\theta-\theta_t}$, ii) Under appropriate choice of $\lambda$, optimizing for $\theta_{t+1}=\arg\min\limits_{\theta\in\calC}\sum\limits_{i=1}^t\tilde \ell(\theta;d_t)+\frac{\lambda}{2}\ltwo{\theta}^2$ over $\theta\in\calC$ gives a good model at step $t+1$, and iii) For all $t\in[n]$, one can privately keep track of $\sum\limits_{i=1}^t\tilde \ell(\theta;d_t)$ using the now standard \emph{tree aggregation protocol}~\cite{Dwork-continual,CSS11-continual}. While a variant of this idea was used in~\cite{thakurta2013nearly} under the name of \emph{follow-the-approximate-leader}, one key difference is that they used a quadratic approximation of the regularized loss, i.e., $\ell(\theta;d_t)+\frac{\lambda}{t}\ltwo{\theta}^2$. This formulation results in a more complicated algorithm,  sub-optimal regret analysis, and failure to maintain structural properties (like sparsity) introduced by composite losses~\cite{duchi2010composite,xiao2010dual,mcmahan2011follow,mcmahan17survey}. %

\begin{algorithm}[ht]
\caption{$\aftrl$: Differentially Private Follow-The-Regularized-Leader (DP-FTRL)}
\begin{algorithmic}[1]
\REQUIRE Data set: $D=[d_1,\cdots,d_n]$ arriving in a stream, in an arbitrary order; constraint set: $\calC$, noise scale: $\sigma$, regularization parameter: $\lambda$,  clipping norm: $L$.
\STATE $\theta_1\leftarrow\argmin\limits_{\theta\in\calC}\frac{\lambda}{2}\ltwo{\theta}^2$. \outputt $\theta_1$.
\STATE  $\tree\leftarrow \init(n,\sigma^2,L)$.
\FOR{$t\in[n]$}
    \STATE Let  $\nabla_t\leftarrow \clip{\nabla_\theta \ell(\theta_t;d_t)}{L}$, where $\clip{\boldv}{L}=\boldv\cdot\min\left\{\frac{L}{\ltwo{\boldv}},1\right\}$, taking $\nabla_\theta \ell(\theta;\nul)=\mathbf{0}$.
    \STATE $\tree\leftarrow\addt(\tree,t,\nabla_t)$.
    \STATE $\bolds_t\leftarrow\gett(\tree,t)$, i.e., estimate $\sum\limits_{i=1}^t\nabla_i$ via tree-aggregation protocol.
    \STATE $\theta_{t+1}\leftarrow\arg\min\limits_{\theta\in\calC}\ip{\bolds_t}{\theta}+\frac{\lambda}{2}\|\theta\|^2_2$. \outputt $\theta_{t+1}$. %
\ENDFOR
\end{algorithmic}
\label{Alg:PFTRL}
\end{algorithm}

Later in the paper, we provide two variants of DP-FTRL (momentum DP-FTRL, and DP-FTRL for least square losses) which will have superior privacy/utility trade-offs for certain problem settings.

DP-FTRL is formally described in Algorithm~\ref{Alg:PFTRL}. There are three functions, $\init$, $\addt$, $\gett$, that correspond to the tree-aggregation algorithm. 
At a high-level, $\init$ initializes the tree data structure $\tree$, $\addt$ allows adding a new gradient $\nabla_t$ to $\tree$, and $\gett$ returns the prefix sum $\sum\limits_{i=1}^t\nabla_i$ privately. 
In our experiments (Section~\ref{sec:empEval}), we use the iterative estimator from \cite{honaker2015efficient} to obtain the optimal estimate of the prefix sums in $\gett$.
Please refer to Appendix~\ref{app:tree} for the formal algorithm descriptions.

It can be shown that the error introduced in DP-FTRL due to privacy is dominated by the error in estimating $\sum\limits_{i=1}^t\nabla_t$ at each $t\in[n]$. It follows from~\cite{thakurta2013nearly} that for a sequence of (adaptively chosen) vectors $\{\nabla_t\}_{t=1}^n$, if we perform $\addt(\tree, t, \nabla_t)$ for each $t\in[n]$, then we can write $\gett(\tree,t) = \sum_{i=1}^t \nabla_i + \boldb_t$ where $\boldb_t$ is normally distributed with mean zero, and $\forall t\in[n] , \ltwo{\bfb_t}\leq L\sigma\sqrt{p\lceil\lg (n)\rceil\ln(n/\beta)}$ w.p. at least $1-\beta$. 

\mypar{Momentum Variant} 
We find that using a momentum term $\gamma\in[0,1]$ with 
Line 7 in Algorithm~\ref{Alg:PFTRL} replaced by
$$
\boldv_{t}\leftarrow\gamma\cdot\boldv_{t-1}+\bolds_t, \text{ } \theta_{t+1}\leftarrow\arg\min\limits_{\theta\in\calC}\ip{\boldv_t}{\theta}+\frac{\lambda}{2}\|\theta -\theta_0 \|^2_2
$$
gives superior empirical privacy/utility trade-off compared to the original algorithm when training non-convex models. 
Throughout the paper, we refer to this variant as momentum DP-FTRL, or DP-FTRLM. Although we do not provide formal regret guarantee for this variant, we conjecture that the superior empirical performance is due to the following reason. The noise added by the tree aggregation algorithm is always bounded by $O(\sqrt{p\ln(1/\delta)}\cdot \ln(n)/\epsilon)$. However, the noise at time step $t$ and $t+1$ can differ by a factor of $O(\sqrt{\ln n})$. This creates sudden jumps in between the output models comparing to DP-SGD. The momentum can smooth out these jumps.

\mypar{Privacy analysis} In Theorem~\ref{thm:privFTRL}, we provide the privacy guarantee for Algorithm~\ref{Alg:PFTRL} and its momentum variant (with proof in Appendix~\ref{app:privFTRL}). In Appendix~\ref{app:multipass}, we extend it to multiple passes over the data set $D$, and batch sizes $>1$.

\begin{thm}[Privacy guarantee]
Algorithm~\ref{Alg:PFTRL} (and its momentum variant) guarantees $\left(\alpha,\frac{\alpha \lceil\lg(n+1)\rceil}{2\sigma^2}\right)$-R\'enyi differential privacy, where $n$ is the number of samples in $D$.
Setting $\sigma=\frac{\sqrt{2\lceil\lg(n+1)\rceil\ln(1/\delta)}}{\epsilon}$, one can guarantee $(\epsilon,\delta)$-differential privacy, for $\epsilon\leq 2\ln(1/\delta)$. 
\label{thm:privFTRL}
\end{thm}

\mypar{DP-FTRL's memory footprint as compared to DP-SGD} 
At any given iteration, the cost of computing the mini-batch gradients is exactly the same for both DP-FTRL and DP-SGD.
The only difference between the memory usage of DP-FTRL as compared to DP-SGD is that DP-FTRL needs to keep track of worst-case $(\log_2(t)+2)$ past gradient/noise information vectors (in $\R^p$) for iteration $t$. 
Note that these are precomputed objects that can be stored in memory.

\subsection{Comparing the Noise Added by DP-SGD with privacy amplification, and DP-FTRL}
\label{sec:equivDP-GD}

In this section, we 
use the equivalence of non-private SGD and FTRL~\cite{mcmahan17survey} to establish equivalence between DP-SGD with privacy amplification (a variant of noisy-SGD) and DP-FTRL in the unconstrained case, i.e., $\calC=\mathbb{R}^p$. We further compare them based on the noise variance added at the same privacy level of $(\epsilon,\delta)$-DP. For the brevity of presentation, we only make the comparison in the setting with $\epsilon=O(1)$.

Let $D = \{d_1,\dots,d_n\}$ be the data set of size $n$.
Consider a general noisy-SGD algorithm with update rule
\begin{equation}\label{eq:unconstrained-sgd}
\theta_{t+1}\leftarrow \theta_t-\eta\cdot\left(\nabla_\theta\ell\left(\theta_t;d_t\right)+\bolda_t\right),
\end{equation}
where $\eta$ is the learning rate and $\bolda_t$ is some random noise.
DP-SGD~\emph{with privacy amplification}, that achieves $(\epsilon, \delta)$-DP, can be viewed as a special case, where $d_t$ is sampled u.a.r. from $D$, and $\bolda_t$ is drawn i.i.d. from $\mathcal{N}\left(0,O\left(\frac{L^2\ln(1/\delta) }{n\epsilon^2}\right)\right)$~\cite{BST14}. If we expand the recursive relation, we can see that the total amount of noise added to the estimation of $\theta_{t+1}$ is 
$\eta\sum\limits_{i=1}^t\bolda_i = \mathcal{N}\left(0, O\left(\frac{\eta^2 L^2 t\cdot\ln(1/\delta)}{n\epsilon^2}\right)\right)$. 

For DP-FTRL, define $\boldb_0=0$, and let $\boldb_t$ be the noise added by the tree-aggregation algorithm at time step $t$ of Algorithm $\aftrl$.
We can show that DP-FTRL, that achieves $(\epsilon,\delta)$-DP, is equivalent to \eqref{eq:unconstrained-sgd}, where
i) the noise $\bolda_t=\boldb_t-\boldb_{t-1}$, 
ii) the data samples $d_t$'s are drawn in sequence from $D$, 
and iii) the learning rate $\eta$ is set to be $\frac{1}{\lambda}$, where $\lambda$ is the regularization parameter in Algorithm $\aftrl$. 
In this variant of noisy SGD, the total noise added to the model is 
$\eta\sum\limits_{i=1}^t\bolda_i = \eta\boldb_t = 
\mathcal{N}\left(0, O\left(\frac{\eta^2 L^2\cdot\ln(1/\delta)\cdot\ln^2(n)}{\epsilon^2}\right)\right)$. The variance of the noise in $\boldb_t$ follows from the following two facts: i) Theorem~\ref{thm:privFTRL} provides the explicit noise variance ($L^2\sigma^2$) to be added to ensure $(\epsilon,\delta)$-differential privacy to the tree-aggregation scheme in Algorithm~\ref{Alg:PFTRL} (Algorithm $\aftrl$), and ii) The $\gett(\tree,t)$ operation in Algorithm $\aftrl$, only requires $O(\ln(n))$ nodes of the binary tree in the tree-aggregation scheme.

Under the same form of the update rule, we can roughly 
(as the noise is not independent in the DP-FTRL case) compare the two algorithms. 
When $t=\Omega(n)$, the noise of DP-SGD \emph{with privacy amplification} matches that of DP-FTRL up to factor of ${\sf polylog}\left(n\right)$. As a result, we expect (and as corroborated by the population risk guarantees) sampled DP-SGD and DP-FTRL to perform similarly. (In Appendix~\ref{app:equivDP-GD} we provide a formal equivalence.) 

It is worth noting that the above calculation overestimates the variance of $\boldb_t$ used for DP-FTRL in the above analysis. If we look carefully at the tree-aggregation scheme, it should be evident that $\boldb_t\sim\mathcal{N}\left(0, O\left(\frac{\eta^2 L^2\cdot\ln(1/\delta)\cdot\ln(n)\cdot \nu}{\epsilon^2}\right)\right)$, where $\nu\in\left[\lceil\lg(n+1)\rceil\right]$ is the number of ones in the binary representation of $t\in[n]$. Furthermore, the variance is reduced by a factor of $\approx\sqrt{2}$ by using techniques from~\cite{honaker2015efficient}. 
Because of these, and due to the fact that privacy amplification by sampling~\cite{KLNRS} is most effective at smaller values of $\epsilon$, in our experiments we see that DP-FTRL is competitive to DP-SGD with amplification, even when there is a ${\sf polylog}(n)$ gap in our analytical noise variance computation.

\section{Regret and Population Risk Guarantees}
\label{sec:highProbabRegAndPopRisk}

In this section we consider the setting when loss function $\ell$ is convex in its first parameter, and provide for DP-FTRL: i) Adversarial regret guarantees for general convex losses, ii) Tighter stochastic regret guarantees for least-squares and linear losses, and iii) Population risk guarantees via online-to-batch conversion. All our guarantees are high-probability over the randomness of the algorithm, i.e., w.p. at least $1-\beta$, the error only depends on ${\sf polylog}(1/\beta)$.

\subsection{Adversarial Regret for (Composite) Losses}
\label{sec:advReg}

The theorem here gives a regret guarantee for Algorithm~\ref{Alg:PFTRL} against a \emph{fully adaptive}~\cite{shalev2011online} adversary who chooses the loss function $\ell(\theta;d_t)$ based on $[\theta_1,\ldots,\theta_{t}]$, but without knowing the internal randomness of the algorithm. See Appendix~\ref{app:regFTRL} for a more general version of Theorem~\ref{thm:regFTRL}, and its proof.

\begin{thm}[Regret guarantee]
Let $\theta$ be any model in $\calC$,
$[\theta_1,\ldots,\theta_n]$ be the outputs of Algorithm $\aftrl$ (Algorithm~\ref{Alg:PFTRL}),
and let $L$ be a bound on the $\ell_2$-Lipschitz constant of the loss functions. 
Setting $\lambda$ optimally and plugging in the noise scale $\sigma$ from Theorem~\ref{thm:privFTRL} to ensure $(\epsilon,\delta)$-differential privacy, we have that for any $\theta^*\in\calC$, w.p. at least $1-\beta$ over the randomness of $\aftrl$, the regret 
{\small
\begin{align*}
& R_D(\aftrl;\theta^*)= \\
& \qquad O\left(L\ltwo{\theta^*}\cdot\left(\frac{1}{\sqrt n}+\sqrt\frac{p^{1/2}\ln^{2}(1/\delta)\ln(1/\beta)}{\epsilon n}\right)\right).
\end{align*}}
\label{thm:regFTRL}
\end{thm}

\mypar{Extension to composite losses} Composite losses~\cite{duchi2010composite,mcmahan2011follow,mcmahan17survey} refer to the setting where in each round, the algorithm is provided with a function $f_t(\theta)=\ell(\theta;d_t)+r_t(\theta)$ with $r_t:\calC\to\mathbb{R}^+$ being a convex regularizer that does not depend on the data sample $d_t$. The $\ell_1$-regularizer, $r_t(\theta) = \norm{\theta}_1$, is perhaps the most important practical example, playing a critical role in high-dimensional statistics (e.g., in the LASSO method)~\cite{bhlmann11high}, as well as for applications like click-through-rate (CTR) prediction where very sparse models are needed for efficiency~\cite{mcmahan2013ad}. In order to operate on composite losses, we simply replace Line~7 of Algorithm $\aftrl$ with \[
\theta_{t+1}\leftarrow\arg\min\limits_{\theta\in\calC}\ip{\bolds_t}{\theta}+{\sum\limits_{i=1}^t r_i(\theta)}+\frac{\lambda}{2}\|\theta\|^2_2,\]
which can be solved in closed form in many important cases such as $\ell_1$ regularization. We obtain Corollary~\ref{cor:regFTRLComp}, analogous to \citep[Theorem 1]{mcmahan17survey} in the non-private case. We do not require any assumption (e.g., Lipschitzness) on the regularizers beyond convexity since we \emph{only linearize the losses} in Algorithm $\aftrl$. It is worth mentioning that~\cite{thakurta2013nearly} is fundamentally incompatible with this type of guarantee.

\begin{cor}
Let $\theta$ be any model in $\calC$,
$[\theta_1,\ldots,\theta_n]$ be the outputs of Algorithm $\aftrl$ (Algorithm~\ref{Alg:PFTRL}),
and $L$ be a bound on the $\ell_2$-Lipschitz constant of the loss functions. W.p. at least $1-\beta$ over the randomness of the algorithm, for any $\theta^*\in\calC$, assuming $\mathbf{0}\in\calC $, we have: 
$$R_D(\aftrl;\theta^*)  \leq\frac{L\sigma\sqrt{p\lceil\lg n\rceil\ln(n/\beta)} +L^2}{\lambda}
+\frac{\lambda}{2n} \ltwo{\theta^*}^2 + \frac{1}{n}\sum\limits_{t=1}^n r_t(\theta^*). $$

\label{cor:regFTRLComp}
\end{cor}

\subsection{Stochastic Regret for Least-squared Losses}
\label{sec:lssq}

In this setting, for each data sample $d_i=(\boldx_i,y_i)$ (with $\boldx_i\in\mathbb{R}^p$ and $y_i\in\mathbb{R}$) in the data set $D=\{d_1,\ldots,d_n\}$, the corresponding loss takes the least-squares form\footnote{A similar argument as in Theorem~\ref{thm:stocReg-ls-exist} can be used in the setting where the loss functions are linear,  $\ell(\theta;d)=\ip{\theta}{d}$ with $d\in\mathbb{R}^p$ and $\ltwo{d}\leq L$.}:
$\ell(\theta;d_i)=(y_i-\ip{\boldx_i}{\theta})^2$.  
We also assume that each data sample $d_i$ is drawn i.i.d. from some fixed distribution $\tau$. 

A straightforward modification of DP-FTRL, $\aftrlls$ (Algorithm \ref{Alg:PFTRL-ls} in Appendix~\ref{app:lssq}), achieves the following.

\begin{thm}[Stochastic regret for least-squared losses]
Let $D=\{(\boldx_1,y_1),\ldots,(\boldx_n,y_n)\}\in\calD^n$ be a data set drawn i.i.d. from $\tau$, let $L=\max\limits_{\boldx\in\calD}\ltwo{\boldx}$, and let $\max\limits_{y\sim\calD}|y|\leq 1$. Let $\theta^* \in \calC$, $\mu=\max\limits_{\theta\in\calC}\ltwo{\theta}$, and $\rho=\max\{\mu,\mu^2\}$.
Then $\aftrlls$ provides $(\epsilon,\delta)$-differentially privacy while outputting $[\theta_1,\ldots,\theta_n]$ s.t. w.p. at least $1-\beta$ for any $\theta^*\in\calC$, $\mathbb{E}_D\left[R_D(\aftrlls;\theta^*)\right]=$

$O \left(L^2\rho^2\left(\sqrt\frac{\ln(n)}{n}+\frac{\sqrt{p\ln^5(n/\beta)\cdot\ln(1/\delta)}}{\epsilon n}\right)\right).$

\label{thm:stocReg-ls-exist}
\end{thm}
The arguments of~\cite{agarwal2017price} can be extended to show a similar regret guarantee \emph{in expectation only}, whereas ours is a high-probability guarantee.

\subsection{Excess Risk via Online-to-Batch Conversion}
\label{sec:popRisk}

Using the online-to-batch conversion~\cite{cesa2002generalization,SSSS09}, from Theorem~\ref{thm:regFTRL}, we can obtain a population risk guarantee\\ $O\left(\left(\sqrt\frac{{\ln(1/\beta)}}{n}+\sqrt{\frac{p^{1/2}\ln^{2}(1/\delta)\ln(1/\beta)}{\epsilon n}}\right)\right)$, where $\beta$ is the failure probability. (See Appendix~\ref{app:pop} for a formal statement.) 
For least squares and linear losses, using the regret guarantee in Theorem~\ref{thm:stocReg-ls-exist} and online-to-batch conversion, one can actually achieve the optimal population risk (up to logarithic factors) $O\left(\sqrt\frac{\ln(n)\ln(1/\beta)}{n}+\frac{\sqrt{p\ln^5(n/\beta)\cdot\ln(1/\delta)}}{\epsilon n}\right)$.

\section{Practical Extensions}
\label{sec:practical_variants}

In this section we consider two practical extensions to Algorithm~\ref{Alg:PFTRL} that are important for real-world use, and are considered in our empirical evaluations.  %

\mypar{Minibatch DP-FTRL} So far, for simplicity we have focused on DP-FTRL with model updates corresponding to new gradient from a single sample. However, in practice, instead of computing the gradient on a single data sample $d_t$ at time step $t$, we will estimate gradient over a batch
$M_t=\left\{d^{(1)}_t,\ldots,d^{(q)}_t\right\}$ as
$\nabla_t=\frac{1}{q}\sum\limits_{i=1}^q\clip{\nabla_\theta\ell\left(\theta_t;d^{(1)}_t\right)}{L}$. This immediately implies the number of steps per epoch to be $\lceil n/q\rceil$. Furthermore, since the $\ell_2$-sensitivity in each batch gets scaled down to $\frac{L}{q}$ instead of $L$ (as in Algorithm $\aftrl$). We take the above two observations into consideration in our privacy accounting.

\mypar{Multiple participations} While Algorithm~\ref{Alg:PFTRL} is stated for a single epoch of training, i.e., where each sample in the data set is used once for obtaining a gradient update, there can be situations where $E > 1$ epochs of training are preferred. We consider three algorithm variants that support this:

\begin{itemize}

    \item \Restart: The simplest approach, discussed in detail \cref{sec:DP-FTRL-TR}, is to simply use separate trees for each epoch, and compose the privacy costs using strong composition.
    
    \item \NoRestart: This approach, considered in detail in \cref{sec:DP-FTRL-NTR}, allows a single aggregation tree to span multiple epochs (possibly even processing data in an arbitrary order as long as each $d_i$ occurs in at most $E$ steps). This requires a more nuanced privacy analysis, as the same training example occurs in multiple leaf nodes, and further can influence interior nodes multiple times, increasing the sensitivity.
    
    \item \SomeRestart: One can combine the above ideas, which can yield improved privacy/utility tradeoffs. For example (as in the experiments of  \cref{sec:interleaving}), one can perform 100 epochs of training, resetting the tree every 20 epochs, using the analysis for \NoRestart within each group of 20 epochs with a shared aggregation tree, and then combining these 5 blocks via strong composition as in \cref{sec:DP-FTRL-TR}. This approach is discussed in depth in \cref{sec:DP-FTRL-SometimesRestart}.

\end{itemize}

\section{Empirical Evaluation}
\label{sec:empEval}
We provide an empirical evaluation of DP-FTRL on four benchmark data sets, and compare its performance with the state-of-the-art  DP-SGD on three axes: (1) \textbf{Privacy}, measured as an $(\epsilon, \delta)$-DP guarantee on the mechanism, 
(2) \textbf{Utility}, measured as (expected) test set accuracy for the trained model under the DP guarantee,
and (3) \textbf{Computation cost}, which we measure in terms of mini-batch size
and number of training iterations. The code is open sourced\footnote{\url{https://github.com/google-research/federated/tree/master/dp_ftrl} for FL experiments, and \url{https://github.com/google-research/DP-FTRL} for centralized learning.}.
First, we evaluate the privacy/utility trade-offs provided by each technique at fixed computation costs. 
Second, we evaluate the privacy/computation trade-offs each technique can provide at fixed utility targets.
A natural application for this is distributed frameworks such as FL, where the privacy budget and a desired utility threshold can be fixed, and the goal is to satisfy both constraints with the least computation. 
Computational cost is of critical importance in FL, as it can get challenging to find available clients with increasing mini-batch size and/or number of training rounds. 
We show the following results: 
(1) DP-FTRL provides superior privacy/utility trade-offs than unamplified DP-SGD, 
(2) For a modest increase in computation cost, DP-FTRL (that does not use any privacy amplification) can match the privacy/utility trade-offs of amplified DP-SGD for all privacy regimes, and further
(3) For regimes with large privacy budgets, DP-FTRL achieves higher accuracy than amplified DP-SGD even at the same computation cost, 
(4) For realistic data set sizes, DP-FTRL can provide superior privacy/computation trade-offs compared to DP-SGD.

\subsection{Experimental Setup}
\label{sec:exp_setup}
\mypar{Datasets} We conduct our evaluation on three image classification tasks, MNIST \cite{lecun1998gradient}, CIFAR-10~\cite{cifar10}, EMNIST (ByMerge split) \cite{cohen_afshar_tapson_schaik_2017}; and a next word prediction task on StackOverflow data set~\cite{so_data}.
Since StackOverflow is naturally keyed by users, we assume training in a federated learning setting, i.e., using the Federated Averaging optimizer for training over users in StackOverflow. The privacy guarantee is thus user-level, in contrast to the example-level privacy for the other three datasets (see Definition \ref{def:diiffP}).

For all experiments with DP, we set the privacy parameter $\delta$ to $10^{-5}$ on MNIST and CIFAR-10, and $10^{-6}$ on EMNIST and StackOverflow, s.t. $\delta < n^{-1}$, where $n$ is the number of users in StackOverflow (or examples in the other data sets).

\mypar{Model Architectures} For all the image classification tasks, we use small convolutional neural networks as in prior work \cite{papernot2020tempered}. 
For StackOverflow, we use the one-layer LSTM network described in \citep{reddi2020adaptive}.
See Appendix~\ref{app:mod_archs} for more details.

\mypar{Optimizers}
We consider DP-FTRL with mini-batch model updates, and multiple epochs. 
In the centralized training experiments, we use \Restart in this section with a small number of epochs. In Appendix~\ref{sec:interleaving}, we provides more results for \SomeRestart for a larger number of epochs. For StackOverflow, we always use \Restart and there are less than five restarts even for 1000 clients per round due to the large population. We provide a privacy analysis for both approaches in Appendix~\ref{app:multipass}. We also consider the momentum variant DP-FTRLM, and find that DP-FTRLM with momentum $0.9$ always outperforms DP-FTRL. 
Similarly, for DP-SGD~\cite{TFpriv}, we consider its momentum variant (DP-SGDM), and report the best-performing variant in each task.
See Appendix~\ref{app:mom_comp} for a comparison of the two optimizers for both techniques.

\begin{figure*}[ht]
\centering
\begin{subfigure}[b]{0.9\textwidth}
\centering
\includegraphics[width=\textwidth]{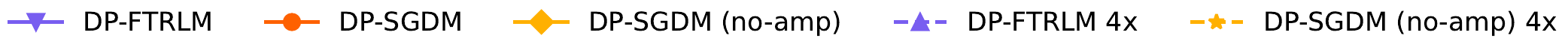}
\end{subfigure}

\begin{subfigure}[b]{0.3\textwidth}
\centering
\includegraphics[width=\textwidth]{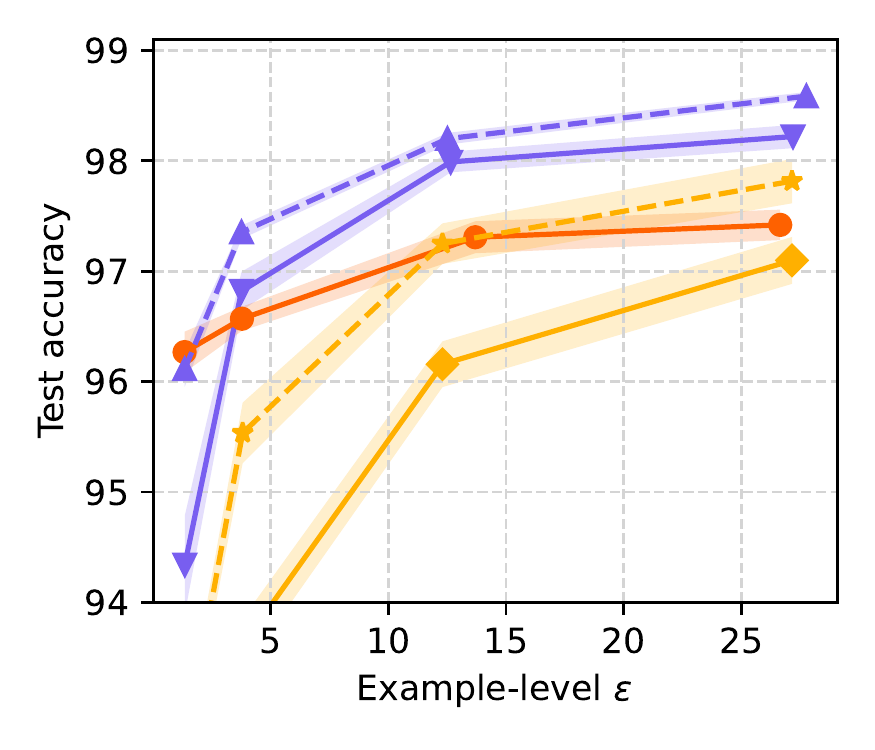}
\caption{MNIST}
\label{fig:acc_privacy_mnist}
\end{subfigure}
\begin{subfigure}[b]{0.3\textwidth}
\centering
\includegraphics[width=\textwidth]{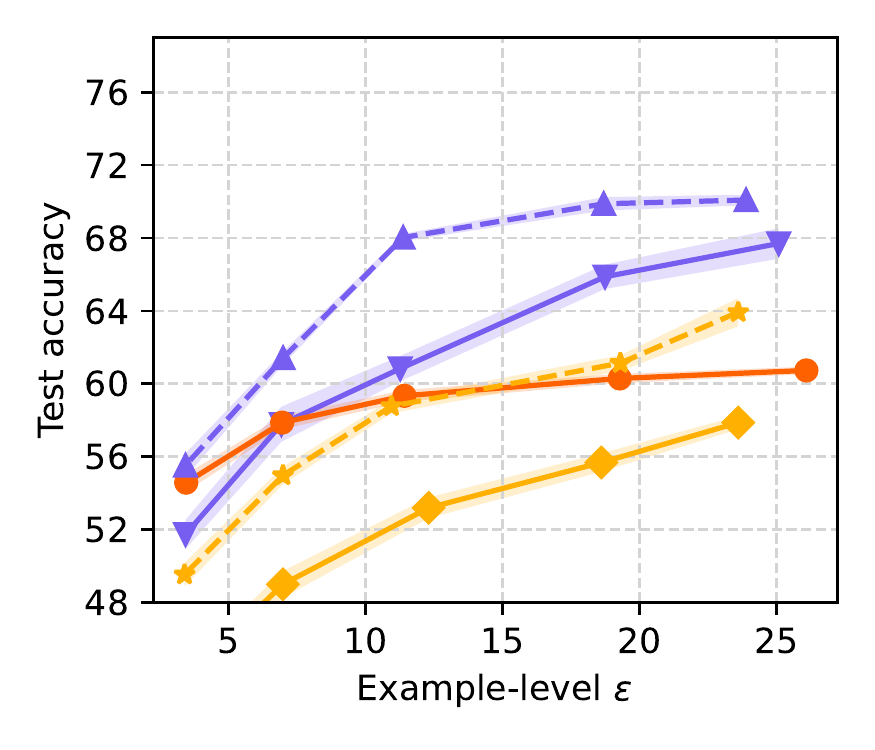}
\caption{CIFAR-10}
\label{fig:acc_privacy_cifar}
\end{subfigure}
\begin{subfigure}[b]{0.3\textwidth}
\centering
\includegraphics[width=\textwidth]{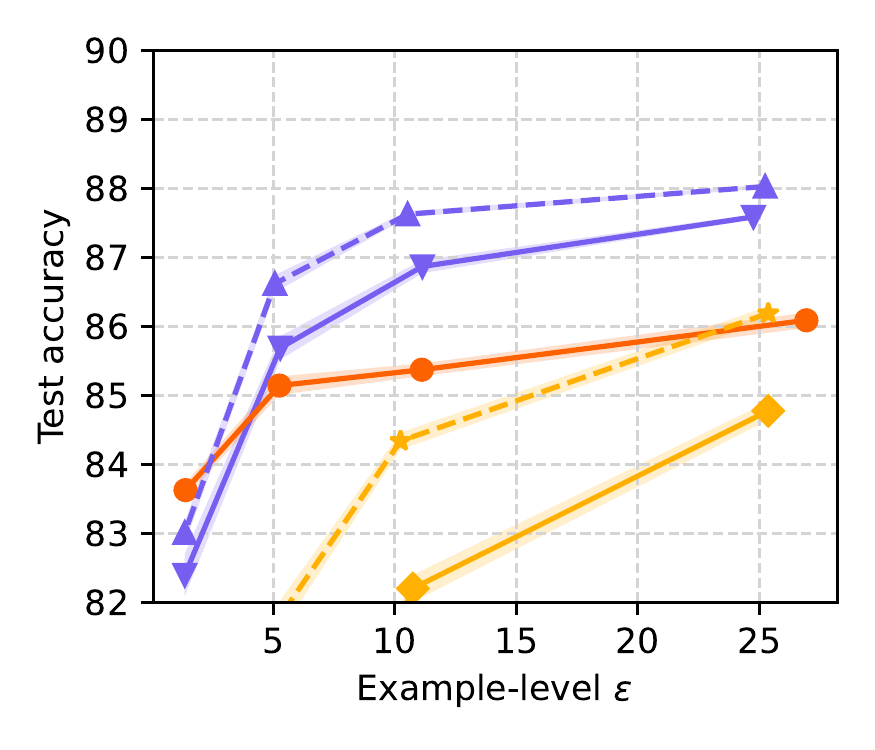}
\caption{EMNIST (ByMerge)}
\label{fig:acc_privacy_emnist}
\end{subfigure}
\caption{Privacy/accuracy trade-offs for DP-SGD (private baseline), DP-SGD without amplification (label ``DP-SGD (no-amp)"), and DP-FTRLM on MNIST (mini-batch size 250), CIFAR-10 (mini-batch size 500), and EMNIST (mini-batch size 500). ``4x" in the label denotes four times computation cost (by increasing batch size four times). All algorithms use 5 epochs of training for the smaller batch size and 20 epochs for the larger batch size, handled by \Restart for FTRL.
}
\label{fig:acc_privacy}
\end{figure*}

\subsection{Privacy/Utility Trade-offs with Fixed Computation}
\label{sec:privTarget}

In Figure~\ref{fig:acc_privacy}, we show accuracy / privacy tradeoffs (by varying the noise multiplier) at fixed computation costs. 
Since both DP-FTRL and DP-SGD require clipping gradients from each sample and adding noise to the aggregated update in each iteration, we consider the number of iterations and the minibatch size as a proxy for computation cost. For each experiment, we run five independent trials, and plot the mean and standard deviation of the final test accuracy at different privacy levels.
We provide details of hyperparameter tuning for all the techniques in Appendix~\ref{app:hyp_tun_privTarget}.

DP-SGD is the state-of-the-art technique used for private deep learning, and amplification by subsampling (or shuffling) forms a crucial component in its privacy analysis.
Thus, we take amplified DP-SGD (or its momentum variant when performance is better) at a fixed computation cost as our baseline. 
We fix the (samples in mini-batch, training iterations) to (250, 1200) for MNIST, (500, 500) for CIFAR-10, and (500, 6975) for EMNIST. These number of steps correspond to $5$ epochs for the smaller batch size and $20$ epochs for the larger batch size.
Our goal is to achieve equal or better tradeoffs \emph{without relying on any privacy amplification}.
As has been mentioned before, we use the \Restart variant of DP-FTRL in this section. Additionally, we make use of a trick where we add additional nodes to the aggregation tree to ensure we can use the root as a low-variance estimate of the total gradient sum for each epoch; details are given in Appendix~\ref{sec:tree completion}.
The privacy computation follows from Appendix~\ref{sec:privacy given order} and~\ref{sec:tree completion}.

DP-SGD without any privacy amplification (``DP-SGD (no-amp)") cannot achieve this:
For all the data sets, the accuracy with DP-SGD (no-amp) at the highest $\eps$ in Figure~\ref{fig:acc_privacy} is worse than the accuracy of the DP-SGD baseline even at its lowest $\eps$.
Further, if we increase the computation by four times (increasing the mini-batch size by four times), the privacy/utility trade-offs of ``DP-SGD (no-amp) 4x" are still substantially worse than the private baseline.

For DP-FTRLM at the same computation cost as our DP-SGD baseline, as the privacy parameter $\eps$ increases, the relative performance of DP-FTRLM improves for each data set, even outperforming the baseline for larger values of $\eps$. Further, if we increase the batch size by four times for DP-FTRLM, its privacy-utility trade-off almost always matches or outperforms the amplified DP-SGD baseline, affirmatively answering this paper's primary question.
In particular, for CIFAR-10 (Figure~\ref{fig:acc_privacy_cifar}), ``DP-FTRLM 4x'' provides superior performance than the DP-SGD even for the lowest $\eps$. 

The number of epochs used here is relatively small. We chose to consider this setting as the advantage of DP-FTRL is more significant in such regime. In Appendix~\ref{sec:interleaving}, we consider running $100$ epochs on CIFAR-10 and $50$ epochs on EMNIST using the \SomeRestart variant. The results demonstrate similar trends, except that the ``cross-over'' point of $\eps$ after which DP-FTRL outperforms DP-SGD shifts right (but is still $< 15$).

We observe similar results for StackOverflow with user-level DP in Figure~\ref{fig:acc_privacy_stackoverflow}.  We fix the computation cost to 100 clients per round (also referred to as the report goal), and $1600$ training rounds. 
DP-SGDM (or more precisely in this case, DP-FedAvg with server momentum) is our baseline.
For DP-SGDM without privacy amplification (DP-SGDM no-amp), the privacy/accuracy trade-off never matches that of the DP-SGDM baseline, and gets significantly worse for lower $\eps$.
With a 4x increase in report goal, DP-SGDM no-amp nearly matches the privacy/utility trade-off of the DP-SGD baseline, outperforming it for larger $\eps$.

For DP-FTRLM, with the same computation cost as the DP-SGDM baseline, it outperforms the baseline for the larger $\eps$, whereas for the four-times increased report goal, it provides a strictly better privacy/utility trade-off.
We conclude DP-FTRL provides superior privacy/utility trade-offs than unamplified DP-SGD, and for a modest increase in computation cost, it can match the performance of DP-SGD, without the need for privacy amplification.

\begin{figure*}[htb]
\centering
\begin{subfigure}[b]{0.33\textwidth}
\centering
\includegraphics[width=\textwidth]{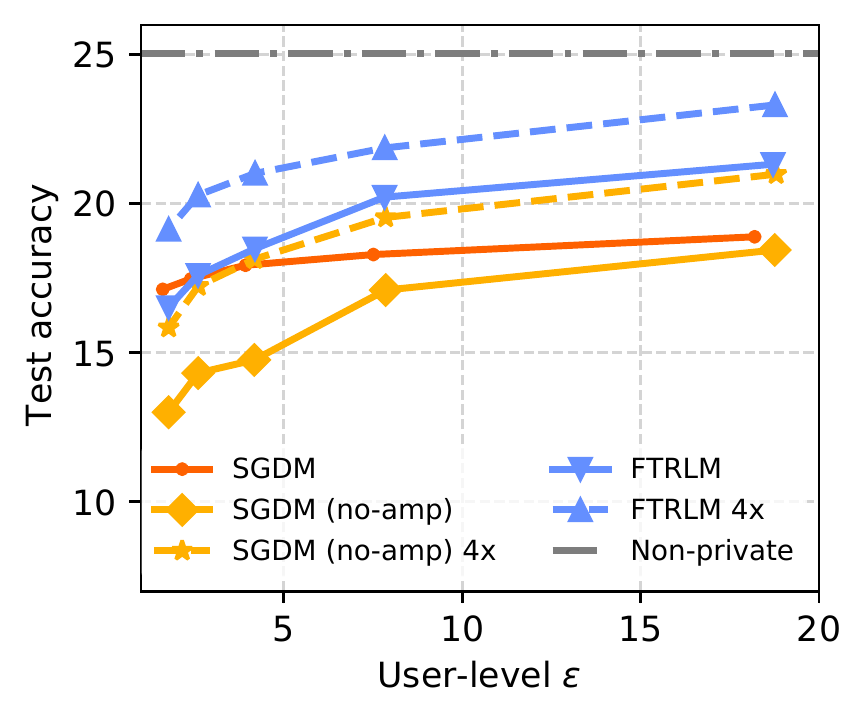}
\caption{}
\label{fig:acc_privacy_stackoverflow}
\end{subfigure}
\begin{subfigure}[b]{0.33\textwidth}
\centering
\includegraphics[width=\textwidth]{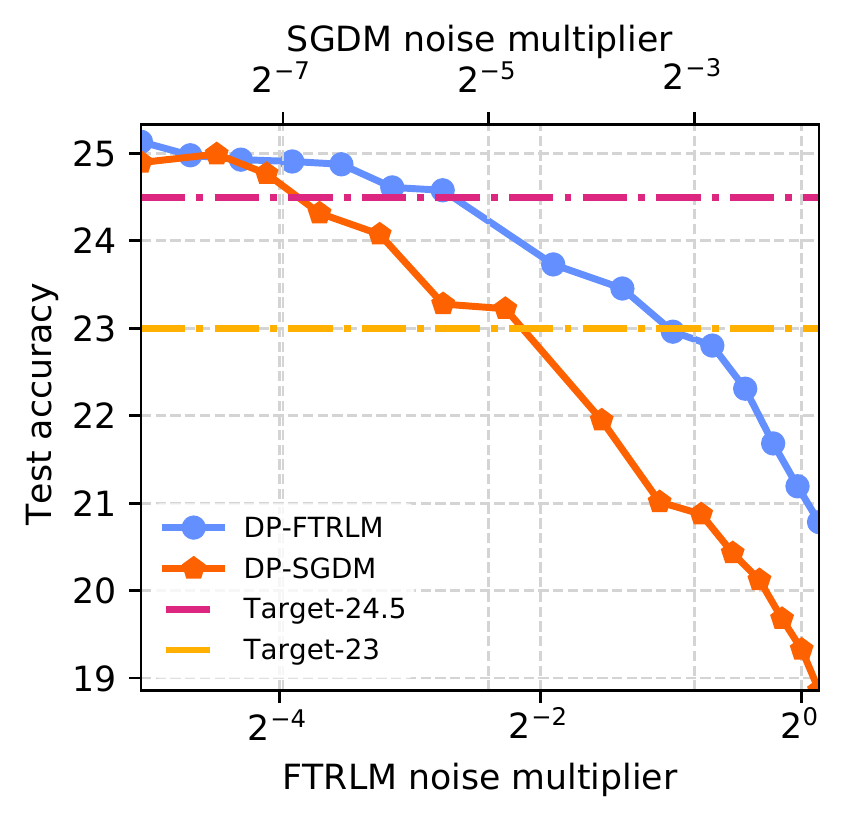}
\caption{}
\label{fig:fl-utility-noise}
\end{subfigure}
\begin{subfigure}[b]{0.33\textwidth}
\centering
\includegraphics[width=\textwidth]{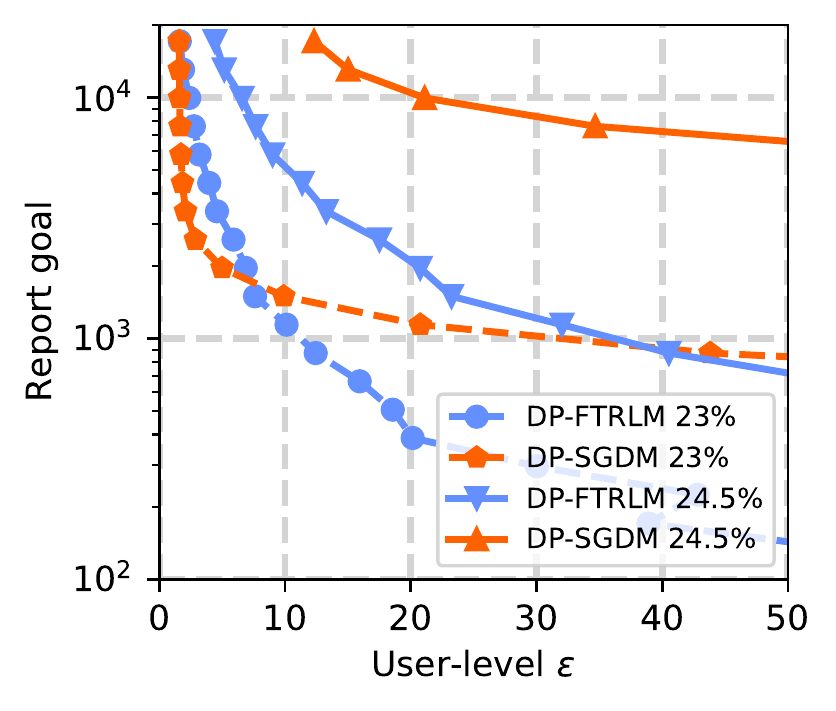}
\caption{}
\label{fig:fl-utility-population-real}
\end{subfigure}
\caption{(a) Accuracy on StackOverflow under different privacy epsilon by varying noise multiplier and batch sizes.
(b) Test accuracy of DP-SGDM and DP-FTRLM with various noise multipliers for StackOverflow.
(c) Relationship between user-level privacy $\eps$ (when $\delta\approx\nicefrac{1}{\text{population}}$) and computation cost (report goal) for two fixed accuracy targets (see legend) on the StackOverflow data set.
}
\label{fig:fl_results}
\end{figure*}

\subsection{Privacy/Computation Trade-offs with Fixed Utility}
\label{sec:utilTarget}
For a sufficiently large data set / population, better privacy vs. accuracy trade-offs can essentially always be achieved at the cost of increased computation. Thus, in this section we slice the privacy/utility/computation space by fixing utility (accuracy) targets, and evaluating how much computation (report goal) is necessary to achieve different $\eps$ for StackOverflow.
Our non-private baseline achieves an accuracy of 25.15\%, and we fix 24.5\% (2.6\% relative loss) and 23\% (8.6\% relative loss) as our accuracy targets. Note that from the accuracy-privacy trade-offs presented in Figure~\ref{fig:acc_privacy_stackoverflow}, achieving even 23\% for either DP-SGD or DP-FTRL will result in a large $\eps$ for the considered report goals.

For each target, we tune hyperparameters (see Appendix~\ref{app:hyp_tun_utilTarget} for details) for both DP-SGDM and DP-FTRLM at a fixed computation cost to obtain the maximum noise scale for each technique while ensuring the trained models meet the accuracy target. Specifically, we fix a report goal of 100 clients per round for 1600 training rounds, and tune DP-SGD and DP-FTRL for 15 noise multipliers, ranging from $(0, 0.3)$ for DP-SGD, and $(0, 1.13)$ for DP-FTRL. At this report goal, for noise multiplier $0.3$, DP-SGD provides $\sim19\%$ accuracy at $\eps\sim18.2$, whereas for noise multiplier $1.13$ DP-FTRL provides $\sim21\%$ accuracy at $\eps\sim18.7$. 
We provide the results in Figure~\ref{fig:fl-utility-noise}. 

For each target accuracy, we choose the largest noise multiplier for each technique that results in the trained model achieving the accuracy target. For accuracies (23\%, 24.5\%), we select noise multipliers (0.035, 0.007) for DP-SGDM, and (0.387, 0.149) for DP-FTRLM, respectively.  
This data allows us to evaluate the privacy/computation trade-offs for both techniques, assuming the accuracy stays constant as we scale up the noise and report goal together (maintaining a constant signal-to-noise ratio while improving $\eps$). This assumption was introduced and validated by \cite{mcmahan2017learning}, which showed that keeping the clipping norm bound, training rounds, and the scale of the noise added to the model update constant, increasing the report goal does not change the final model accuracy.  %

We plot the results in Figure~\ref{fig:fl-utility-population-real}.
We see that for utility target 24.5\% and $\delta=10^{-6}$, DP-FTRLM achieves any privacy $\eps \in (0, 50)$ at a lower computational cost than DP-SGDM. 
For utility target 23\%, we observe the same behavior for $\eps > 8.8$.

\section{Conclusion}
\label{sec:conclusion}

In this paper we introduce the DP-FTRL algorithm, which we show to have the tightest known regret guarantees under DP, and have the best known excess population risk guarantees for a single pass algorithm on non-smooth convex losses. 
For linear and least-squared losses, we show DP-FTRL actually achieves the optimal population risk.
Furthermore, we show on benchmark data sets that DP-FTRL, which does not rely on any privacy amplification, can outperform amplified DP-SGD at large values of $\eps$, and be competitive to it for all ranges of $\epsilon$ for a modest increase in computation cost (batch size).
This work leaves two main open questions: i) Can DP-FTRL achieve the optimal excess population risk for all convex losses in a single pass?, and ii) Can one tighten the empirical gap between DP-SGD and DP-FTRL at smaller values of $\epsilon$, possibly via a better estimator of the gradient sums from the tree data structure?
\section*{Acknowledgements}

We would specially thank Thomas Steinke for providing us with dynamic programming based privacy accounting scheme (and its associated proof) for \NoRestart. We would also like to thank Adam Smith for suggesting the use of~\cite{honaker2015efficient} for variance reduction, Vinith Suriyakumar for noticing an error in a reported empirical result, and Yin-Tat Lee for independently finding the privacy accounting bug in multi-pass  \NoRestart. We would additionally like to thank Borja Balle and Satyen Kale for the helpful discussions through the course of this project. 

\bibliography{reference}

\begin{thebibliography}{76}
\providecommand{\natexlab}[1]{#1}
\providecommand{\url}[1]{\texttt{#1}}
\expandafter\ifx\csname urlstyle\endcsname\relax
  \providecommand{\doi}[1]{doi: #1}\else
  \providecommand{\doi}{doi: \begingroup \urlstyle{rm}\Url}\fi

\bibitem[Abadi et~al.(2016)Abadi, Chu, Goodfellow, McMahan, Mironov, Talwar,
  and Zhang]{DP-DL}
Mart{\'{\i}}n Abadi, Andy Chu, Ian~J. Goodfellow, H.~Brendan McMahan, Ilya
  Mironov, Kunal Talwar, and Li~Zhang.
\newblock Deep learning with differential privacy.
\newblock In \emph{Proc. of the 2016 {ACM} {SIGSAC} Conf. on Computer and
  Communications Security ({CCS}'16)}, pages 308--318, 2016.

\bibitem[Abernethy et~al.(2019)Abernethy, Jung, Lee, McMillan, and
  Tewari]{abernethy2017online}
Jacob Abernethy, Young~Hun Jung, Chansoo Lee, Audra McMillan, and Ambuj Tewari.
\newblock Online learning via the differential privacy lens.
\newblock In \emph{NeurIPS}, 2019.

\bibitem[Agarwal and Singh(2017)]{agarwal2017price}
Naman Agarwal and Karan Singh.
\newblock The price of differential privacy for online learning.
\newblock In \emph{Proceedings of the 34th International Conference on Machine
  Learning-Volume 70}, pages 32--40, 2017.

\bibitem[Asoodeh et~al.(2020)Asoodeh, Liao, Calmon, Kosut, and
  Sankar]{asoodeh2020better}
Shahab Asoodeh, Jiachun Liao, Flavio~P Calmon, Oliver Kosut, and Lalitha
  Sankar.
\newblock A better bound gives a hundred rounds: Enhanced privacy guarantees
  via f-divergences.
\newblock In \emph{2020 IEEE International Symposium on Information Theory
  (ISIT)}, pages 920--925. IEEE, 2020.

\bibitem[Balle et~al.(2020)Balle, Kairouz, McMahan, Thakkar, and
  Thakurta]{balle2020privacy}
Borja Balle, Peter Kairouz, Brendan McMahan, Om~Dipakbhai Thakkar, and
  Abhradeep Thakurta.
\newblock Privacy amplification via random check-ins.
\newblock \emph{Advances in Neural Information Processing Systems}, 33, 2020.

\bibitem[Bassily et~al.(2014)Bassily, Smith, and Thakurta]{BST14}
Raef Bassily, Adam Smith, and Abhradeep Thakurta.
\newblock Private empirical risk minimization: Efficient algorithms and tight
  error bounds.
\newblock In \emph{Proc. of the 2014 IEEE 55th Annual Symp. on Foundations of
  Computer Science (FOCS)}, pages 464--473, 2014.

\bibitem[Bassily et~al.(2019)Bassily, Feldman, Talwar, and
  Thakurta]{bassily2019private}
Raef Bassily, Vitaly Feldman, Kunal Talwar, and Abhradeep Thakurta.
\newblock Private stochastic convex optimization with optimal rates.
\newblock In \emph{Advances in Neural Information Processing Systems}, pages
  11279--11288, 2019.

\bibitem[Bassily et~al.(2020)Bassily, Feldman, Guzm{\'a}n, and
  Talwar]{bassily2020stability}
Raef Bassily, Vitaly Feldman, Crist{\'o}bal Guzm{\'a}n, and Kunal Talwar.
\newblock Stability of stochastic gradient descent on nonsmooth convex losses.
\newblock \emph{arXiv preprint arXiv:2006.06914}, 2020.

\bibitem[Bhlmann and van~de Geer(2011)]{bhlmann11high}
Peter Bhlmann and Sara van~de Geer.
\newblock \emph{Statistics for High-Dimensional Data: Methods, Theory and
  Applications}.
\newblock Springer Publishing Company, Incorporated, 2011.
\newblock ISBN 3642201911.

\bibitem[Bonawitz et~al.(2019)Bonawitz, Eichner, Grieskamp, Huba, Ingerman,
  Ivanov, Kiddon, Konecny, Mazzocchi, McMahan, et~al.]{bonawitz2019towards}
Keith Bonawitz, Hubert Eichner, Wolfgang Grieskamp, Dzmitry Huba, Alex
  Ingerman, Vladimir Ivanov, Chloe Kiddon, Jakub Konecny, Stefano Mazzocchi,
  H~Brendan McMahan, et~al.
\newblock Towards federated learning at scale: System design.
\newblock \emph{arXiv preprint arXiv:1902.01046}, 2019.

\bibitem[Canonne et~al.(2020)Canonne, Kamath, and Steinke]{canonne2020discrete}
Cl{\'e}ment Canonne, Gautam Kamath, and Thomas Steinke.
\newblock The discrete gaussian for differential privacy.
\newblock \emph{arXiv preprint arXiv:2004.00010}, 2020.

\bibitem[Cesa-Bianchi et~al.(2002)Cesa-Bianchi, Conconi, and
  Gentile]{cesa2002generalization}
Nicol{\'o} Cesa-Bianchi, Alex Conconi, and Claudio Gentile.
\newblock On the genomeralization ability of on-line learning algorithms.
\newblock In \emph{Advances in neural information processing systems}, pages
  359--366, 2002.

\bibitem[Chan et~al.(2011)Chan, Shi, and Song]{CSS11-continual}
T.-H.~Hubert Chan, Elaine Shi, and Dawn Song.
\newblock Private and continual release of statistics.
\newblock \emph{{ACM} Trans. on Information Systems Security}, 14\penalty0
  (3):\penalty0 26:1--26:24, November 2011.

\bibitem[Chaudhuri et~al.(2011)Chaudhuri, Monteleoni, and
  Sarwate]{chaudhuri2011differentially}
Kamalika Chaudhuri, Claire Monteleoni, and Anand~D Sarwate.
\newblock Differentially private empirical risk minimization.
\newblock \emph{Journal of Machine Learning Research}, 12\penalty0
  (Mar):\penalty0 1069--1109, 2011.

\bibitem[Cohen et~al.(2017)Cohen, Afshar, Tapson, and
  Schaik]{cohen_afshar_tapson_schaik_2017}
Gregory Cohen, Saeed Afshar, Jonathan Tapson, and Andre~Van Schaik.
\newblock Emnist: Extending mnist to handwritten letters.
\newblock \emph{2017 International Joint Conference on Neural Networks
  (IJCNN)}, 2017.
\newblock \doi{10.1109/ijcnn.2017.7966217}.

\bibitem[Duchi et~al.(2011)Duchi, Hazan, and Singer]{duchi2011adaptive}
John Duchi, Elad Hazan, and Yoram Singer.
\newblock Adaptive subgradient methods for online learning and stochastic
  optimization.
\newblock \emph{Journal of machine learning research}, 12\penalty0 (7), 2011.

\bibitem[Duchi et~al.(2010)Duchi, Shalev-Shwartz, Singer, and
  Tewari]{duchi2010composite}
John~C Duchi, Shai Shalev-Shwartz, Yoram Singer, and Ambuj Tewari.
\newblock Composite objective mirror descent.
\newblock In \emph{COLT}, pages 14--26. Citeseer, 2010.

\bibitem[Duchi et~al.(2013)Duchi, Jordan, and Wainwright]{DJW13}
John~C. Duchi, Michael~I. Jordan, and Martin~J. Wainwright.
\newblock Local privacy and statistical minimax rates.
\newblock In \emph{54th Annual {IEEE} Symposium on Foundations of Computer
  Science, {FOCS} 2013, 26-29 October, 2013, Berkeley, CA, {USA}}, pages
  429--438. {IEEE} Computer Society, 2013.
\newblock \doi{10.1109/FOCS.2013.53}.
\newblock URL \url{https://doi.org/10.1109/FOCS.2013.53}.

\bibitem[Dwork and Roth(2014)]{dwork2014algorithmic}
Cynthia Dwork and Aaron Roth.
\newblock The algorithmic foundations of differential privacy.
\newblock \emph{Foundations and Trends in Theoretical Computer Science},
  9\penalty0 (3--4):\penalty0 211--407, 2014.

\bibitem[Dwork et~al.(2006{\natexlab{a}})Dwork, Kenthapadi, McSherry, Mironov,
  and Naor]{ODO}
Cynthia Dwork, Krishnaram Kenthapadi, Frank McSherry, Ilya Mironov, and Moni
  Naor.
\newblock Our data, ourselves: Privacy via distributed noise generation.
\newblock In \emph{Advances in Cryptology---EUROCRYPT}, pages 486--503,
  2006{\natexlab{a}}.

\bibitem[Dwork et~al.(2006{\natexlab{b}})Dwork, McSherry, Nissim, and
  Smith]{DMNS}
Cynthia Dwork, Frank McSherry, Kobbi Nissim, and Adam Smith.
\newblock Calibrating noise to sensitivity in private data analysis.
\newblock In \emph{Proc. of the Third Conf. on Theory of Cryptography (TCC)},
  pages 265--284, 2006{\natexlab{b}}.
\newblock URL \url{http://dx.doi.org/10.1007/11681878\_14}.

\bibitem[Dwork et~al.(2010)Dwork, Naor, Pitassi, and Rothblum]{Dwork-continual}
Cynthia Dwork, Moni Naor, Toniann Pitassi, and Guy~N. Rothblum.
\newblock Differential privacy under continual observation.
\newblock In \emph{Proc. of the Forty-Second {ACM} Symp. on Theory of Computing
  ({STOC}'10)}, pages 715--724, 2010.

\bibitem[Erlingsson et~al.(2019)Erlingsson, Feldman, Mironov, Raghunathan,
  Talwar, and Thakurta]{soda-shuffling}
{\'U}lfar Erlingsson, Vitaly Feldman, Ilya Mironov, Ananth Raghunathan, Kunal
  Talwar, and Abhradeep Thakurta.
\newblock Amplification by shuffling: From local to central differential
  privacy via anonymity.
\newblock In \emph{Proceedings of the Thirtieth Annual ACM-SIAM Symposium on
  Discrete Algorithms}, pages 2468--2479. SIAM, 2019.

\bibitem[Erlingsson et~al.(2020{\natexlab{a}})Erlingsson, Feldman, Mironov,
  Raghunathan, Song, Talwar, and Thakurta]{esa++}
{\'{U}}lfar Erlingsson, Vitaly Feldman, Ilya Mironov, Ananth Raghunathan,
  Shuang Song, Kunal Talwar, and Abhradeep Thakurta.
\newblock Encode, shuffle, analyze privacy revisited: Formalizations and
  empirical evaluation.
\newblock \emph{CoRR}, abs/2001.03618, 2020{\natexlab{a}}.
\newblock URL \url{https://arxiv.org/abs/2001.03618}.

\bibitem[Erlingsson et~al.(2020{\natexlab{b}})Erlingsson, Feldman, Mironov,
  Raghunathan, Song, Talwar, and Thakurta]{esa2}
{\'{U}}lfar Erlingsson, Vitaly Feldman, Ilya Mironov, Ananth Raghunathan,
  Shuang Song, Kunal Talwar, and Abhradeep Thakurta.
\newblock Encode, shuffle, analyze privacy revisited: Formalizations and
  empirical evaluation.
\newblock \emph{CoRR}, abs/2001.03618, 2020{\natexlab{b}}.

\bibitem[Evfimievski et~al.(2003)Evfimievski, Gehrke, and
  Srikant]{evfimievski2003limiting}
Alexandre Evfimievski, Johannes Gehrke, and Ramakrishnan Srikant.
\newblock Limiting privacy breaches in privacy preserving data mining.
\newblock In \emph{Proceedings of the twenty-second ACM SIGMOD-SIGACT-SIGART
  symposium on Principles of database systems}, pages 211--222, 2003.

\bibitem[Facebook(2020)]{opacus}
Facebook.
\newblock Introducing opacus: A high-speed library for training pytorch models
  with differential privacy, 2020.

\bibitem[Feldman et~al.(2018)Feldman, Mironov, Talwar, and Thakurta]{FMTT18}
Vitaly Feldman, Ilya Mironov, Kunal Talwar, and Abhradeep Thakurta.
\newblock Privacy amplification by iteration.
\newblock In \emph{59th Annual IEEE Symp. on Foundations of Computer Science
  (FOCS)}, pages 521--532, 2018.

\bibitem[Feldman et~al.(2020{\natexlab{a}})Feldman, Koren, and
  Talwar]{feldman2019private}
Vitaly Feldman, Tomer Koren, and Kunal Talwar.
\newblock Private stochastic convex optimization: Optimal rates in linear time.
\newblock In \emph{Proc. of the Fifty-Second {ACM} Symp. on Theory of Computing
  ({STOC}'20)}, 2020{\natexlab{a}}.

\bibitem[Feldman et~al.(2020{\natexlab{b}})Feldman, McMillan, and
  Talwar]{feldman2020hiding}
Vitaly Feldman, Audra McMillan, and Kunal Talwar.
\newblock Hiding among the clones: A simple and nearly optimal analysis of
  privacy amplification by shuffling.
\newblock \emph{arXiv preprint arXiv:2012.12803}, 2020{\natexlab{b}}.

\bibitem[Google(2019)]{TFpriv}
Google.
\newblock Tensorflow-privacy.
\newblock \url{https://github.com/tensorflow/privacy}, 2019.

\bibitem[Hazan(2019)]{hazan2019introduction}
Elad Hazan.
\newblock Introduction to online convex optimization.
\newblock \emph{arXiv preprint arXiv:1909.05207}, 2019.

\bibitem[Hazan and Kale(2014)]{hazan2014beyond}
Elad Hazan and Satyen Kale.
\newblock Beyond the regret minimization barrier: optimal algorithms for
  stochastic strongly-convex optimization.
\newblock \emph{The Journal of Machine Learning Research}, 15\penalty0
  (1):\penalty0 2489--2512, 2014.

\bibitem[Honaker(2015)]{honaker2015efficient}
James Honaker.
\newblock {Efficient Use of Differentially Private Binary Trees}.
\newblock In \emph{Theory and Practice of Differential Privacy (TPDP 2015),
  London, UK}, 2015.

\bibitem[Iyengar et~al.(2019)Iyengar, Near, Song, Thakkar, Thakurta, and
  Wang]{iyengar2019towards}
Roger Iyengar, Joseph~P Near, Dawn Song, Om~Thakkar, Abhradeep Thakurta, and
  Lun Wang.
\newblock Towards practical differentially private convex optimization.
\newblock In \emph{2019 IEEE Symposium on Security and Privacy (SP)}, 2019.

\bibitem[Jagielski et~al.(2020)Jagielski, Ullman, and
  Oprea]{jagielski2020auditing}
Matthew Jagielski, Jonathan Ullman, and Alina Oprea.
\newblock Auditing differentially private machine learning: How private is
  private sgd?
\newblock \emph{arXiv preprint arXiv:2006.07709}, 2020.

\bibitem[Jain and Thakurta(2014)]{jain2014near}
Prateek Jain and Abhradeep~Guha Thakurta.
\newblock (near) dimension independent risk bounds for differentially private
  learning.
\newblock In \emph{International Conference on Machine Learning}, pages
  476--484, 2014.

\bibitem[Jain et~al.(2012)Jain, Kothari, and Thakurta]{JKT-online}
Prateek Jain, Pravesh Kothari, and Abhradeep Thakurta.
\newblock Differentially private online learning.
\newblock In \emph{Proc. of the 25th Annual Conf. on Learning Theory (COLT)},
  volume~23, pages 24.1--24.34, June 2012.

\bibitem[Kairouz et~al.(2019)Kairouz, McMahan, Avent, Bellet, Bennis, Bhagoji,
  Bonawitz, Charles, Cormode, Cummings, et~al.]{kairouz2019advances}
Peter Kairouz, H~Brendan McMahan, Brendan Avent, Aur{\'e}lien Bellet, Mehdi
  Bennis, Arjun~Nitin Bhagoji, Keith Bonawitz, Zachary Charles, Graham Cormode,
  Rachel Cummings, et~al.
\newblock Advances and open problems in federated learning.
\newblock \emph{arXiv preprint arXiv:1912.04977}, 2019.

\bibitem[Kairouz et~al.(2021)Kairouz, Mcmahan, Song, Thakkar, Thakurta, and
  Xu]{kairouz21b}
Peter Kairouz, Brendan Mcmahan, Shuang Song, Om~Thakkar, Abhradeep Thakurta,
  and Zheng Xu.
\newblock Practical and private (deep) learning without sampling or shuffling.
\newblock In \emph{Proceedings of the 38th International Conference on Machine
  Learning}, pages 5213--5225, 2021.

\bibitem[Kalai and Vempala(2005)]{kalai2005efficient}
Adam Kalai and Santosh Vempala.
\newblock Efficient algorithms for online decision problems.
\newblock \emph{Journal of Computer and System Sciences}, 71\penalty0
  (3):\penalty0 291--307, 2005.

\bibitem[Kasiviswanathan et~al.(2008)Kasiviswanathan, Lee, Nissim,
  Raskhodnikova, and Smith]{KLNRS}
Shiva~Prasad Kasiviswanathan, Homin~K. Lee, Kobbi Nissim, Sofya Raskhodnikova,
  and Adam~D. Smith.
\newblock What can we learn privately?
\newblock In \emph{49th Annual {IEEE} Symp. on Foundations of Computer Science
  (FOCS)}, pages 531--540, 2008.

\bibitem[Kifer et~al.(2012)Kifer, Smith, and Thakurta]{kifer2012private}
Daniel Kifer, Adam Smith, and Abhradeep Thakurta.
\newblock Private convex empirical risk minimization and high-dimensional
  regression.
\newblock In \emph{Conference on Learning Theory}, pages 25--1, 2012.

\bibitem[Krizhevsky(2009)]{cifar10}
Alex Krizhevsky.
\newblock Learning multiple layers of features from tiny images, 2009.

\bibitem[LeCun et~al.(1998)LeCun, Bottou, Bengio, and
  Haffner]{lecun1998gradient}
Yann LeCun, L{\'e}on Bottou, Yoshua Bengio, and Patrick Haffner.
\newblock Gradient-based learning applied to document recognition.
\newblock \emph{Proceedings of the IEEE}, 86\penalty0 (11):\penalty0
  2278--2324, 1998.

\bibitem[McMahan(2011)]{mcmahan2011follow}
Brendan McMahan.
\newblock Follow-the-regularized-leader and mirror descent: Equivalence
  theorems and l1 regularization.
\newblock In \emph{Proceedings of the Fourteenth International Conference on
  Artificial Intelligence and Statistics}, pages 525--533, 2011.

\bibitem[McMahan et~al.(2017{\natexlab{a}})McMahan, Moore, Ramage, Hampson, and
  y~Arcas]{FL1}
Brendan McMahan, Eider Moore, Daniel Ramage, Seth Hampson, and
  Blaise~Ag{\"{u}}era y~Arcas.
\newblock Communication-efficient learning of deep networks from decentralized
  data.
\newblock In \emph{Proceedings of the 20th International Conference on
  Artificial Intelligence and Statistics, {AISTATS} 2017, 20-22 April 2017,
  Fort Lauderdale, FL, {USA}}, pages 1273--1282, 2017{\natexlab{a}}.
\newblock URL \url{http://proceedings.mlr.press/v54/mcmahan17a.html}.

\bibitem[McMahan(2017)]{mcmahan17survey}
H.~Brendan McMahan.
\newblock A survey of algorithms and analysis for adaptive online learning.
\newblock \emph{Journal of Machine Learning Research}, 18\penalty0
  (90):\penalty0 1--50, 2017.
\newblock URL \url{http://jmlr.org/papers/v18/14-428.html}.

\bibitem[McMahan and Streeter(2010)]{mcmahan10boundopt}
H.~Brendan McMahan and Matthew Streeter.
\newblock Adaptive bound optimization for online convex optimization.
\newblock In \emph{Proceedings of the 23rd Annual Conference on Learning Theory
  (COLT)}, 2010.

\bibitem[McMahan et~al.(2013)McMahan, Holt, Sculley, Young, Ebner, Grady, Nie,
  Phillips, Davydov, Golovin, et~al.]{mcmahan2013ad}
H~Brendan McMahan, Gary Holt, David Sculley, Michael Young, Dietmar Ebner,
  Julian Grady, Lan Nie, Todd Phillips, Eugene Davydov, Daniel Golovin, et~al.
\newblock Ad click prediction: a view from the trenches.
\newblock In \emph{Proceedings of the 19th ACM SIGKDD international conference
  on Knowledge discovery and data mining}, pages 1222--1230, 2013.

\bibitem[McMahan et~al.(2017{\natexlab{b}})McMahan, Ramage, Talwar, and
  Zhang]{mcmahan2017learning}
H~Brendan McMahan, Daniel Ramage, Kunal Talwar, and Li~Zhang.
\newblock Learning differentially private recurrent language models.
\newblock \emph{arXiv preprint arXiv:1710.06963}, 2017{\natexlab{b}}.

\bibitem[McMahan et~al.(2018)McMahan, Andrew, Erlingsson, Chien, Mironov,
  Papernot, and Kairouz]{mcmahan2018general}
H~Brendan McMahan, Galen Andrew, Ulfar Erlingsson, Steve Chien, Ilya Mironov,
  Nicolas Papernot, and Peter Kairouz.
\newblock A general approach to adding differential privacy to iterative
  training procedures.
\newblock \emph{arXiv preprint arXiv:1812.06210}, 2018.

\bibitem[Mironov(2017)]{mironov2017renyi}
Ilya Mironov.
\newblock R{\'e}nyi differential privacy.
\newblock In \emph{2017 IEEE 30th Computer Security Foundations Symposium
  (CSF)}, pages 263--275. IEEE, 2017.

\bibitem[Nasr et~al.(2021)Nasr, Song, Thakurta, Papernot, and
  Carlini]{nasr2021adversary}
Milad Nasr, Shuang Song, Abhradeep Thakurta, Nicolas Papernot, and Nicholas
  Carlini.
\newblock Adversary instantiation: Lower bounds for differentially private
  machine learning.
\newblock In \emph{IEEE S and P (Oakland)}, 2021.

\bibitem[Overflow(2018)]{so_data}
Stack Overflow.
\newblock {The Stack Overflow Data}, 2018.
\newblock \url{https://www.kaggle.com/stackoverflow/stackoverflow}.

\bibitem[Papernot et~al.(2020{\natexlab{a}})Papernot, Chien, Song, Thakurta,
  and Erlingsson]{papernot2020making}
Nicolas Papernot, Steve Chien, Shuang Song, Abhradeep Thakurta, and Ulfar
  Erlingsson.
\newblock Making the shoe fit: Architectures, initializations, and tuning for
  learning with privacy, 2020{\natexlab{a}}.
\newblock URL \url{https://openreview.net/forum?id=rJg851rYwH}.

\bibitem[Papernot et~al.(2020{\natexlab{b}})Papernot, Thakurta, Song, Chien,
  and Erlingsson]{papernot2020tempered}
Nicolas Papernot, Abhradeep Thakurta, Shuang Song, Steve Chien, and {\'U}lfar
  Erlingsson.
\newblock Tempered sigmoid activations for deep learning with differential
  privacy.
\newblock \emph{arXiv preprint arXiv:2007.14191}, 2020{\natexlab{b}}.

\bibitem[Pichapati et~al.(2019)Pichapati, Suresh, Yu, Reddi, and
  Kumar]{pichapati2019adaclip}
Venkatadheeraj Pichapati, Ananda~Theertha Suresh, Felix~X Yu, Sashank~J Reddi,
  and Sanjiv Kumar.
\newblock Adaclip: Adaptive clipping for private sgd.
\newblock \emph{arXiv preprint arXiv:1908.07643}, 2019.

\bibitem[Ramaswamy et~al.(2020)Ramaswamy, Thakkar, Mathews, Andrew, McMahan,
  and Beaufays]{ramaswamy2020training}
Swaroop Ramaswamy, Om~Thakkar, Rajiv Mathews, Galen Andrew, H.~Brendan McMahan,
  and Françoise Beaufays.
\newblock Training production language models without memorizing user data,
  2020.

\bibitem[Reddi et~al.(2020)Reddi, Charles, Zaheer, Garrett, Rush,
  Kone{\v{c}}n{\`y}, Kumar, and McMahan]{reddi2020adaptive}
Sashank Reddi, Zachary Charles, Manzil Zaheer, Zachary Garrett, Keith Rush,
  Jakub Kone{\v{c}}n{\`y}, Sanjiv Kumar, and H~Brendan McMahan.
\newblock Adaptive federated optimization.
\newblock \emph{arXiv preprint arXiv:2003.00295}, 2020.

\bibitem[Robbins and Monro(1951)]{robbins1951stochastic}
Herbert Robbins and Sutton Monro.
\newblock A stochastic approximation method.
\newblock \emph{The annals of mathematical statistics}, pages 400--407, 1951.

\bibitem[Shalev-Shwartz(2012)]{shalev12}
Shai Shalev-Shwartz.
\newblock Online learning and online convex optimization.
\newblock \emph{Foundations and Trends in Machine Learning}, 4\penalty0
  (2):\penalty0 107--194, 2012.

\bibitem[Shalev{-}Shwartz et~al.(2009)Shalev{-}Shwartz, Shamir, Srebro, and
  Sridharan]{SSSS09}
Shai Shalev{-}Shwartz, Ohad Shamir, Nathan Srebro, and Karthik Sridharan.
\newblock Stochastic convex optimization.
\newblock In \emph{{COLT} 2009 - The 22nd Conference on Learning Theory,
  Montreal, Quebec, Canada, June 18-21, 2009}, 2009.
\newblock URL
  \url{http://www.cs.mcgill.ca/\%7Ecolt2009/papers/018.pdf\#page=1}.

\bibitem[Shalev-Shwartz et~al.(2011)]{shalev2011online}
Shai Shalev-Shwartz et~al.
\newblock Online learning and online convex optimization.
\newblock \emph{Foundations and trends in Machine Learning}, 4\penalty0
  (2):\penalty0 107--194, 2011.

\bibitem[Smith and Thakurta(2013)]{thakurta2013nearly}
Adam Smith and Abhradeep Thakurta.
\newblock (nearly) optimal algorithms for private online learning in
  full-information and bandit settings.
\newblock In \emph{Advances in Neural Information Processing Systems}, pages
  2733--2741, 2013.

\bibitem[Song and Shmatikov(2019)]{song2019auditing}
Congzheng Song and Vitaly Shmatikov.
\newblock Auditing data provenance in text-generation models.
\newblock In \emph{Proceedings of the 25th ACM SIGKDD International Conference
  on Knowledge Discovery \& Data Mining}, pages 196--206, 2019.

\bibitem[Song et~al.(2013)Song, Chaudhuri, and Sarwate]{song2013stochastic}
Shuang Song, Kamalika Chaudhuri, and Anand~D Sarwate.
\newblock Stochastic gradient descent with differentially private updates.
\newblock In \emph{2013 IEEE Global Conference on Signal and Information
  Processing}, pages 245--248. IEEE, 2013.

\bibitem[Thakkar et~al.(2019)Thakkar, Andrew, and McMahan]{TAB19}
Om~Thakkar, Galen Andrew, and H.~Brendan McMahan.
\newblock Differentially private learning with adaptive clipping.
\newblock \emph{CoRR}, abs/1905.03871, 2019.
\newblock URL \url{http://arxiv.org/abs/1905.03871}.

\bibitem[Thakkar et~al.(2020)Thakkar, Ramaswamy, Mathews, and
  Beaufays]{thakkar2020understanding}
Om~Thakkar, Swaroop Ramaswamy, Rajiv Mathews, and Fran{\c{c}}oise Beaufays.
\newblock Understanding unintended memorization in federated learning.
\newblock \emph{arXiv preprint arXiv:2006.07490}, 2020.

\bibitem[Tram{\`e}r and Boneh(2021)]{TB21}
Florian Tram{\`e}r and Dan Boneh.
\newblock Differentially private learning needs better features (or much more
  data).
\newblock In \emph{International Conference on Learning Representations
  (ICLR)}, 2021.

\bibitem[Vadhan(2017)]{vadhan2017complexity}
Salil Vadhan.
\newblock The complexity of differential privacy.
\newblock In \emph{Tutorials on the Foundations of Cryptography}, pages
  347--450. Springer, 2017.

\bibitem[Wang et~al.(2019)Wang, Balle, and Kasiviswanathan]{wang2019subsampled}
Yu-Xiang Wang, Borja Balle, and Shiva~Prasad Kasiviswanathan.
\newblock Subsampled r{\'e}nyi differential privacy and analytical moments
  accountant.
\newblock In \emph{The 22nd International Conference on Artificial Intelligence
  and Statistics}, pages 1226--1235. PMLR, 2019.

\bibitem[Warner(1965)]{Warner}
Stanley~L. Warner.
\newblock Randomized response: A survey technique for eliminating evasive
  answer bias.
\newblock \emph{J. of the American Statistical Association}, 60\penalty0
  (309):\penalty0 63--69, 1965.

\bibitem[Wu et~al.(2017)Wu, Li, Kumar, Chaudhuri, Jha, and Naughton]{WLKCJN17}
Xi~Wu, Fengan Li, Arun Kumar, Kamalika Chaudhuri, Somesh Jha, and Jeffrey~F.
  Naughton.
\newblock Bolt-on differential privacy for scalable stochastic gradient
  descent-based analytics.
\newblock In Semih Salihoglu, Wenchao Zhou, Rada Chirkova, Jun Yang, and Dan
  Suciu, editors, \emph{Proceedings of the 2017 {ACM} International Conference
  on Management of Data, {SIGMOD}}, 2017.

\bibitem[Xiao(2010)]{xiao2010dual}
Lin Xiao.
\newblock Dual averaging methods for regularized stochastic learning and online
  optimization.
\newblock \emph{The Journal of Machine Learning Research}, 11:\penalty0
  2543--2596, 2010.

\bibitem[Zhu and Wang(2019)]{zhu2019poission}
Yuqing Zhu and Yu-Xiang Wang.
\newblock Poission subsampled r{\'e}nyi differential privacy.
\newblock In \emph{International Conference on Machine Learning}, pages
  7634--7642. PMLR, 2019.

\end{thebibliography}
\bibliographystyle{plainnat}

\ifsupp
\appendix
\section{Other Related Work}
\label{sec:related}

Differentially private empirical risk minimization (ERM) and private online learning are  well-studied areas in the privacy literature~\cite{chaudhuri2011differentially,kifer2012private,JKT-online,thakurta2013nearly,song2013stochastic,BST14,jain2014near,DP-DL,mcmahan2017learning, WLKCJN17,agarwal2017price,abernethy2017online,bassily2019private,iyengar2019towards, pichapati2019adaclip,TAB19,feldman2019private,papernot2020tempered}\footnote{This is only a small representative subset of the literature.}. The connection between private ERM and private online learning was first explored in~\cite{JKT-online}, and the idea of using stability induced by differential privacy for designing low-regret algorithms was explored in~\cite{kalai2005efficient,agarwal2017price,abernethy2017online}. To the best of our knowledge, this paper for the first time explores the idea using a purely online learning algorithm for training deep learning models, without relying on any stochasticity in the data for privacy.
\section{Missing Details from Section~\ref{sec:privateFTRL}}
\label{sec:privFTRL}

\subsection{Details of the Tree Aggregation Scheme}
\label{app:tree}

In this section we provide the formal details of the tree aggregation scheme used in Algorithm~\ref{Alg:PFTRL} (Algorithm $\aftrl$).

\begin{enumerate}
\item $\init(n,\sigma^2,L)$: Initialize a complete binary tree $\tree$ with $2^{\lceil\lg(n)\rceil}$ leaf nodes, with each node being sampled i.i.d. from $\calN(0,L^2\sigma^2\cdot\mathbb{I}_{p\times p})$.
\item $\addt(\tree,t,\boldv)$: Add $\boldv$ to all the nodes along the path to the root of $\tree$, starting from $t$-th leaf node.
\item $\gett(\tree,t)$: Let $[\node_1,\ldots,\node_{h}]$ be the list of nodes from the root of $\tree$ to the $t$-th leaf node, with $\node_1$ being the root node and $\node_h$ being the leaf node.  
\begin{enumerate}
    \item Initialize $\bolds\leftarrow {\bf 0}^p$ and convert $t$ to binary in $h$ bit representation $[b_1,\ldots,b_h]$, with $b_1$ being the most significant bit.
    \item For each $j\in[h]$, if $b_j=1$, then add the value in left sibling of $\node_j$ to $\bolds$. Here if $\node_j$ is the left child, then it is treated as its own left sibling.
    \item Return $\bolds$.
\end{enumerate}
\end{enumerate}

\mypar{Incorporating the iterative estimator from~\cite{honaker2015efficient}} Here, we state a variant of the $\gett(\tree,t)$ function (called $\gettrv(\tree,t)$) based on the variance reduction technique used in~\cite{honaker2015efficient}. The main idea is as follows: In the estimator for $\gett(\tree,t)$ above, each $\node_j$ refers to a noisy/private estimate of all the nodes in the sub-tree of $\tree$ rooted at $\node_j$. Notice that one can obtain independent estimates of the same, with one for each level of the sub-tree rooted at $\node_j$, by summing up the nodes at the corresponding level. Of course, the variance of each of these estimates will be different.~\cite{honaker2015efficient} provided an estimator to combine these independent estimates in order to lower the overall variance in the final estimate. In the following, we provide the formal description of $\gettrv(\tree,t)$. The text {\color{blue}colored in blue} is the only difference from $\gett(\tree,t)$. The recurrent updating rule in \cref{eq:tree_recur} only use the nodes ``below'' the current node to reduce the variance \cite{honaker2015efficient} as we can not access the future gradients for a streaming algorithm. In practice, the value of the left node of a sub-tree $\boldr'_{[x:y]}$ is stored in the worst-case $\log_2(t)+1$ memory and only the right node will be recursively calculated on the fly. 

\begin{enumerate}
    \setcounter{enumi}{2}
    \item   $\gettrv(\tree,t)$: Let $[\node_1,\ldots,\node_{h}]$ be the list of nodes from the root of $\tree$ to the $t$-th leaf node, with $\node_1$ being the root node and $\node_h$ being the leaf node. 
    \begin{enumerate}
    \item Initialize $\bolds\leftarrow {\bf 0}^p$ and convert $t$ to binary in $h$ bit representation $[b_1,\ldots,b_h]$, with $b_1$ being the most significant bit.
        {\color{blue}\item For each $j\in[h]$, if $b_j=1$, then do the following.
        \begin{enumerate}
            \item Indexing the leaf nodes $1, 2, \dots$, for any two leaf node indices $\leftn\leq\rightn$, let 
            $r_{\leftn:\rightn}\leftarrow $ value in $\tree$ corresponding to the least common ancestor of $\leftn$ and $\rightn$.
            \item For the sub-tree rooted at the left sibling of $\node_j$ (or $\node_j$ itself if it is the left child), let $[a:b]$ be the indices of the leaf nodes in this subtree of $\tree$.
            \item Estimate $\bolds_{[x:z]}$ representing the sum of the values in leaf nodes $x$ through $z$ recursively as follows:
            \begin{equation}\label{eq:tree_recur}\bolds_{[x:z]}\leftarrow\frac{\boldr'_{[x:z]}}{2-(x-z+1)^{-1}},\ \ \text{where}\ \  \boldr'_{[x:z]}\leftarrow\boldr_{[x:z]}+\frac{\boldr'_{[x:y]}+\boldr'_{[y+1:z]}}{2} \text{ and $y=\lfloor(x+z)/2\rfloor$},\end{equation}
            with base case $\boldr'_{[x:x]} = \boldr_{[x:x]}$, which is simply the value at leaf $x$.
            \item Add $\bolds_{[a:b]}$ to $\bolds$.
        \end{enumerate}}
        \item Return $\bolds$.
    \end{enumerate}
\end{enumerate}

\subsection{Proof of Theorem~\ref{thm:privFTRL}}
\label{app:privFTRL}

\begin{proof}
Notice that in Algorithm~\ref{Alg:PFTRL}, all accesses to private information is only through the tree data structure $\calT$. Hence, to prove the privacy guarantee, it is sufficient to show that for any data set $V=\{\boldv_1,\ldots,\boldv_n\}$ (with each $\ltwo{\boldv_i}\leq L$), the operations on the tree data structure (i.e., the $\init$, $\addt$, $\gett$)  provide the privacy guarantees in the Theorem statement. First, notice that each $\boldv_i$ affects at most $\lceil\lg(n+1)\rceil$ nodes in the tree $\calT$. Additionally, notice that the computation in each node of the tree $\calT$ is essentially a summation query. With these two observations, one can use standard properties of Gaussian mechanism~\cite{ODO},\citep[Corollary 3]{mironov2017renyi}, and adaptive RDP composition~\citep[Proposition1]{mironov2017renyi} to complete the proof. 

\iffalse
%
Notice that in Algorithm~\ref{Alg:PFTRL}, all accesses to private information is only through the tree data structure $\calT$. Hence, to prove the privacy guarantee, it is sufficient to show that for any data set $V=\{\boldv_1,\ldots,\boldv_n\}$ (with each $\ltwo{\boldv_i}\leq L$), the operations on the tree data structure (i.e., the $\init$, $\addt$, $\gett$)  provide the privacy guarantees in the Theorem statement. First, notice that each $\boldv_i$ affects at most $\lceil\lg(n)\rceil$ nodes in the tree $\calT$. Additionally, notice that the computation in each node of the tree $\calT$ is essentially a summation query. With these two observations, one can use standard properties of Gaussian mechanism~\cite{ODO},\citep[Corollary 3]{mironov2017renyi}, and adaptive RDP composition~\citep[Proposition1]{mironov2017renyi} to complete the proof. 
\fi

While the original work on tree aggregation~\cite{Dwork-continual,CSS11-continual} did not use either Gaussian mechanism or RDP composition, it is not hard to observe that the translation to the current setting is immediate.
\end{proof}

\subsection{Missing details from Section~\ref{sec:equivDP-GD} (Comparing Noise in DP-SGD (with amplification) and DP-FTRL)}
\label{app:equivDP-GD}

\begin{thm}
Consider data set $D=\{d_1,\ldots,d_n\}$, model space $\calC=\mathbb{R}^p$ and initial model $\theta_0={\mathbf{0}}^p$.
For $t\in[n]$, 
let the update of Noisy-SGD be $\theta^{\npsgd}_{t+1}\leftarrow \theta_t-\eta\cdot\left(\nabla_\theta\ell\left(\theta_t^{\npsgd};d_t\right)+\bolda_t\right)$, where $\bolda_t$'s are noise random variables. 
Let the DP-FTRL (Algorithm~\ref{Alg:PFTRL}) updates be $\theta^{\dpftrl}_{t+1}\leftarrow \argmin\limits_{\theta\in\mathbb{R}^p}\sum\limits_{i=1}^t \nabla_\theta\ell\left(\theta^{\dpftrl}_i;d_i\right)+\ip{\boldb_t}{\theta}+\frac{1}{2\eta}\ltwo{\theta}^2$, where $\boldb_t$'s are the noises added by the tree-aggregation mechanism. 

If we instantiate $\bolda_t=\boldb_t-\boldb_{t-1}$, and $\eta=\frac{1}{\lambda}$, then for all $t\in[n]$, $\theta^{\npsgd}_t=\theta^{\dpftrl}_t$.
\label{thm:eqvDPSGD}
\end{thm}

\begin{proof}
Consider the non-private SGD and FTRL. Recall that the SGD update is $\theta^{\sf SGD}_{t+1}\leftarrow \theta^{\sf SGD}_t-\eta\nabla_\theta\ell(\theta^{\sf SGD}_t;d_t)$, where $\eta$ is the learning rate. Opening up the recurrence, we have $\theta^{\sf SGD}_{t+1}\leftarrow\theta_0-\eta\sum\limits_{i=1}^t \nabla_\theta\ell(\theta^{\sf SGD}_i;d_i)$. If $\theta^{\sf SGD}_0=\mathbf{0}^p$, then equivalently $\theta^{\sf SGD}_{t+1}\leftarrow \argmin\limits_{\theta\in\mathbb{R}^p}\ip{\sum\limits_{i=1}^t \nabla_\theta\ell(\theta^{\sf SGD}_i;d_i)}{\theta} + \frac{1}{2\eta}\ltwo{\theta}^2$. 
This is identical to the update rule of the non-private FTRL (i.e., with $\sigma$ set to $0$ in DP-FTRL) with regularization parameter $\lambda$ set to $\frac{1}{\eta}$.

Now we consider the Noisy-SGD and DP-FTRL.
Recall that Noisy-SGD has update rule $\theta^{\npsgd}_{t+1}\leftarrow \theta^{\npsgd}_t-\eta\left(\nabla_\theta\ell(\theta^{\npsgd}_t;d_t)+\bolda_t\right)$, where $\bolda_t$ is the Gaussian noise added at time step $t$. 
Similar as before, this rule can be written as 
\begin{align}\label{eqn:equivalent pf sgd}
\theta^{\npsgd}_{t+1}\leftarrow
\argmin\limits_{\theta\in\mathbb{R}^p}
\ip{\sum\limits_{i=1}^t \nabla_\theta\ell(\theta^{\npsgd}_i;d_i)}{\theta} + \ip{\sum_{i=1}^t \bolda_i}{\theta} + \frac{1}{2\eta}\ltwo{\theta}^2.    
\end{align}
The update rule of DP-FTRL can be written as
\begin{align}\label{eqn:equivalent pf ftrl}
\theta^{\dpftrl}_{t+1}\leftarrow \argmin\limits_{\theta\in\mathbb{R}^p}\ip{\sum\limits_{i=1}^t \nabla_\theta\ell\left(\theta^{\dpftrl}_i;d_i\right)}{\theta}+\ip{\boldb_t}{\theta}+\frac{\lambda}{2}\ltwo{\theta}^2,
\end{align}
where $\boldb_t$ is the noise that gets added by the tree-aggregation mechanism at time step $t+1$. 
If we 
1) set $\lambda = \frac{1}{\eta}$,
2) draw data samples sequentially from $D$ in Noisy-SGD, and
3) set $\bolda_t=\boldb_t-\boldb_{t-1}$ so that $\sum\limits_{i=1}^t \bolda_t=\boldb_t$,
we can establish the equivalence between \eqref{eqn:equivalent pf sgd} and \eqref{eqn:equivalent pf ftrl}.
This completes the proof.
\end{proof}

\section{Missing Details from Section~\ref{sec:highProbabRegAndPopRisk}}
\label{app:highProbabRegAndPopRisk}

\subsection{Proof of Theorem~\ref{thm:regFTRL}}
\label{app:regFTRL}
We first present a more detailed version of Theorem~\ref{thm:regFTRL} and then present its proof.
\begin{thm}[Regret guarantee (Theorem~\ref{thm:regFTRL} in detail)]
Let $[\theta_1,\ldots,\theta_n]$ be the outputs of Algorithm $\aftrl$ (Algorithm~\ref{Alg:PFTRL}), and $L$ be a bound on the $\ell_2$-Lipschitz constant of the loss functions. W.p. at least $1-\beta$ over the randomness of $\aftrl$, the following is true for any $\theta^*\in\calC$.
\begin{align*}
& \frac{1}{n}\sum\limits_{t=1}^n\ell(\theta_t;d_t)-\frac{1}{n}\sum\limits_{t=1}^n\ell(\theta^*;d_t) 
\leq \frac{L\sigma\sqrt{p\lceil\lg n\rceil\ln(n/\beta)} +L^2}{\lambda}+\frac{\lambda}{2n}\left(\ltwo{\theta^*}^2-\ltwo{\theta_1}^2\right)
\end{align*}
Setting $\lambda$ optimally and plugging in the noise scale $\sigma$ from Theorem~\ref{thm:privFTRL} to ensure $(\epsilon,\delta)$-differential privacy, we have
\begin{align*}
& R_D(\aftrl;\theta^*)= O\left(L\ltwo{\theta^*}\cdot\left(\frac{1}{\sqrt n}+\sqrt\frac{p^{1/2}\ln^{2}(1/\delta)\ln(1/\beta)}{\epsilon n}\right)\right).
\end{align*}
\label{thm:regFTRLDet}
\end{thm}

\begin{proof}
Recall that by Algorithm $\aftrl$, $\theta_{t+1}\leftarrow\argmin\limits_{\theta\in\calC}\underbrace{\sum\limits_{i=1}^t\ip{\nabla_i}{\theta}+\frac{\lambda}{2}\ltwo{\theta}^2+\ip{\bfb_t}{\theta}}_{\pJ(\theta)}$, where the Gaussian noise $\bfb_t=\bolds_t-\sum\limits_{i=1}^t\nabla_i$ for $\bolds_t$ being the output of $\gett(\tree,t)$.
By standard concentration of spherical Gaussians, w.p. at least $1-\beta$, $\forall t\in[n]$, $\ltwo{\bfb_t}\leq L\sigma\sqrt{p\lceil\lg (n)\rceil\ln(n/\beta)}$. We will use this bound to control the error introduced due to privacy.
Now, consider the optimizer of the non-private objective: 
\begin{align*}
\thetash_{t+1} \leftarrow 
\argmin\limits_{\theta\in\calC}\underbrace{
\sum\limits_{i=1}^t\ip{\nabla_i}{\theta}+\frac{\lambda}{2}\ltwo{\theta}^2}_{\npJ(\theta)},
\qquad \text{where } \nabla_t = \nabla \ell(\theta_t;d_t).
\end{align*}
That is, post-hoc we consider the hypothetical application of non-private FTRL to the same sequence of \emph{linearized} loss functions $f_t(\thetash) = \ip{\nabla_t}{\thetash} = \ip{\nabla \ell(\theta_t;d_t)}{\thetash}$ seen in the private training run.
In the following, we will first bound how much the models output by $\aftrl$ deviate from models output by the hypothetical non-private FTRL discussed above. 
Then, we invoke standard regret bound for FTRL, while accounting for the deviation of the models output by $\aftrl$.

To bound $\ltwo{\thetash_{t+1}-\theta_{t+1}}$, we apply Lemma~\ref{lem:boundDev}. 
We set $\phi_1(\theta)=\npJ(\theta)/\lambda$, $\phi_2(\theta)=\pJ(\theta) / \lambda$, and both $\|\cdot\|$ and its dual as the $\ell_2$ norm. 
We thus have $\Psi(\theta) = \ip{\boldb_t   }{\theta}/\lambda$, with $\boldb_t/\lambda$ being its subgradient.
Therefore,
\begin{equation}
    \ltwo{\thetash_{t+1}-\theta_{t+1}}\leq \frac{\ltwo{\bfb_t}}{\lambda}.\label{eq:2}
\end{equation}

\begin{lem}[Lemma 7 from~\cite{mcmahan17survey} restated]
Let $\phi_1:\calC\to\mathbb{R}$ be a convex function (defined over $\calC\subseteq\mathbb{R}^p$) s.t. $\theta_1\in\argmin\limits_{\theta\in\calC }\phi_1(\theta)$ exists. Let $\Psi(\theta)$ be a convex function s.t. $\phi_2(\theta)=\phi_1(\theta)+\Psi(\theta)$ is 1-strongly convex w.r.t. $\|\cdot\|$-norm. Let $\theta_2\in\argmin\limits_{\theta\in\calC}\phi_2(\theta)$. Then for any $\bfb$ in the subgradient of $\Psi$ at $\theta_1$, the following is true:
$\|\theta_1-\theta_2\|_*\leq \|\bfb\|_*$. Here $\|\cdot\|_*$ is the dual-norm of $\|\cdot\|$.
\label{lem:boundDev}
\end{lem}

We can now easily bound the regret. By standard linear approximation ``trick'' from the online learning literature~\cite{shalev12,hazan2019introduction}, we have the following. For $\nabla_t=\nabla_\theta\ell(\theta_t;d_t)$,
\begin{align}
  \frac{1}{n}\sum\limits_{t=1}^n\ell(\theta_t;d_t)-\frac{1}{n}\sum\limits_{t=1}^n\ell(\theta^*;d_t)
& \leq \frac{1}{n}\sum\limits_{t=1}^n\ip{\nabla_t}{\theta_t-\theta^*}\nonumber\\
& =\frac{1}{n}\sum\limits_{t=1}^n\ip{\nabla_t}{\theta_t-\thetash_t+\thetash_t-\theta^*}\nonumber\\
& =\underbrace{\frac{1}{n}\sum\limits_{t=1}^n\ip{\nabla_t}{\thetash_t-\theta^*}}_{A}+\underbrace{\frac{1}{n}\sum\limits_{t=1}^n\ip{\nabla_t}{\theta_t-\thetash_t}}_{B}.
\label{eq:4}
\end{align}
One can bound the term $A$ in~\eqref{eq:4} by~\citep[Theorem 5.2]{hazan2019introduction} and get $A\leq \left(\frac{L^2}{\lambda}+\frac{\lambda}{2n}\left(\ltwo{\theta^*}^2-\ltwo{\theta_1}^2\right)\right)$.
As for term $B$, using~\eqref{eq:2} and the concentration on $\boldb_t$ mentioned earlier, we have, w.p. at least $1-\beta$,
\begin{align}
B&\leq \frac{1}{n}\sum\limits_{t=1}^n\ltwo{\nabla_t}\cdot\ltwo{\thetash_t-\theta_t}
\leq\frac{1}{n}\sum\limits_{t=1}^n L\cdot\ltwo{\thetash_t-\theta_t}\leq \frac{L\sigma\sqrt{p\lceil\lg n\rceil\ln(n/\beta)}}{\lambda}.\label{eq:5}
\end{align}
Combining~\eqref{eq:4} and~\eqref{eq:5}, we immediately have the first part of of Theorem~\ref{thm:regFTRL}. 
To prove the second part of the theorem, we just optimize for the regularization parameter $\lambda$ and plug in the noise scale $\sigma$ from Theorem~\ref{thm:privFTRL}.
\end{proof}

\subsection{Additional Details for Section~\ref{sec:lssq}}
\label{app:lssq}

In Algorithm~\ref{Alg:PFTRL-ls}, we present a version of DP-FTRL for least square loss. In this modified algorithm, the functions $\initb$, $\addtb$, and $\gettb$ are identical to $\init$, $\addt$, and $\gett$ respectively in Algorithm~\ref{Alg:PFTRL}. 
The functions $\addtq$, $\addtq$, and $\gettq$ are similar to $\init$, $\addt$, and $\gett$, except that the $p$-dimensional vector versions are replaced by $p\times p$-dimensional matrix version, and the noise in $\initq$ is initialized by symmetric $p\times p$ Gaussian matrices with each entry drawn i.i.d. from $\calN\left(0,L^4\sigma^2\right)$.

\begin{algorithm}[h]
\caption{$\aftrlls$: Differentially Private Follow-The-Regularized-Leader (DP-FTRL) for least-squared losses}
\begin{algorithmic}[1]
\REQUIRE Data set: $D=\{(\boldx_1,y_1),\cdots,(\boldx_n,y_n)\}$ arriving in a stream, constraint set: $\calC$, noise scale: $\sigma$, regularization parameter: $\lambda$, upper bound on $\left\{\ltwo{\boldx_t}\right\}_{t=1}^n: L$.
\STATE $\theta_1\leftarrow\argmin\limits_{\theta\in\calC}\frac{\lambda}{2}\ltwo{\theta}^2$. \outputt $\theta_1$. 
\STATE  $\treeb\leftarrow \initb(n,\sigma^2,L)$, $\treeq\leftarrow \initq(n,\sigma^2,L^2)$.
\FOR{$t\in[n]$}
    \STATE Let $\boldv_t\leftarrow y_t\cdot \boldx_t$, and $M_t\leftarrow \boldx_t\boldx_t^\top$. \STATE $\treeb\leftarrow\addtb(\treeb,t,\boldv_t)$ and $\treeq\leftarrow\addtq(\treeq,t,M_t)$.
    \STATE $\bolds_t\leftarrow\gettb(\treeb,t)$, and $W_t\leftarrow\gettq(\treeq,t)$.
    {\STATE $\theta_{t+1}\leftarrow\arg\min\limits_{\theta\in\calC}\left(\theta^\top\cdot W_t\cdot\theta-2\ip{\bolds_t}{\theta}\right)+\frac{\lambda}{2}\|\theta\|^2_2$. \outputt $\theta_{t+1}$.}
\ENDFOR
\end{algorithmic}
\label{Alg:PFTRL-ls}
\end{algorithm}

We first present the privacy guarantee of Algorithm~\ref{Alg:PFTRL-ls} in Theorem~\ref{thm:privFTRL-ls}. Its proof is almost identical to that of Theorem~\ref{thm:privFTRL}, except that we need to measure the sensitivity of the covariance matrix in the Frobenius norm. 
\begin{thm}[Privacy guarantee] If $\ltwo{\boldx}\leq L$ and $|y|\leq 1$ for all $(\boldx,y)\in\calD$ and $\theta\in\calC$, then 
Algorithm~\ref{Alg:PFTRL} (Algorithm $\aftrl$) satisfies $\left(\alpha,\frac{\alpha\lceil\lg(n)\rceil}{\sigma^2}\right)$-RDP. Correspondingly, by setting $\sigma=\frac{2\sqrt{\lceil\lg(n)\rceil\ln(1/\delta)}}{\epsilon}$  one can satisfy $(\epsilon,\delta)$-differential privacy guarantee, as long as $\epsilon\leq 2\ln(1/\delta)$.
\label{thm:privFTRL-ls}
\end{thm}

In Theorem~\ref{thm:stocReg-ls}, we present the regret guarantee for Algorithm~\ref{Alg:PFTRL-ls}.
\begin{thm}[Stochastic regret for least-squared losses]
Let $D=\{(\boldx_1,y_1),\ldots,(\boldx_n,y_n)\}\in\calD^n$ be a data set drawn i.i.d. from $\tau$, 
with $L=\max\limits_{\boldx\in\calD}\ltwo{\boldx}$
and $\max\limits_{y\sim\calD}|y|\leq 1$. 
Let $\calC$ be the model space and $\mu=\max\limits_{\theta\in\calC}\ltwo{\theta}$. 
Let $\theta^*$ be any model in $\calC$, and $[\theta_1,\ldots,\theta_n]$ be the outputs of Algorithm $\aftrlls$ (Algorithm~\ref{Alg:PFTRL-ls}). Then w.p. at least $1-\beta$ (over the randomness of the algorithm), we have
\begin{align*}
&\mathbb{E}_{D}\left[R_D(\aftrlls;\theta^*)\right]=\mathbb{E}_D\left[\frac{1}{n}\sum\limits_{t=1}^n\left(y_t-\ip{\boldx_t}{\theta_t}\right)^2-\left(y_t-\ip{\boldx_t}{\theta^*}\right)^2\right]
\\=&
O\left(\frac{p\ln^2(n)\ln(n/\beta)\sigma^2\cdot\left(L^2+L^4\mu^2+L^3\mu\right)}{\lambda n}
+\frac{L^4\mu^2}{\lambda}+\frac{\lambda\ln(n)}{n}\cdot\ltwo{\theta^*}^2\right).
\end{align*}
Setting $\lambda$ optimally and plugging in the noise scale $\sigma$ from Theorem~\ref{thm:privFTRL-ls} to ensure $(\epsilon,\delta)$-differential privacy, we have, 
\begin{align*}
&\mathbb{E}_{D}\left[R_D(\aftrlls;\theta^*)\right]=L^2\cdot\ltwo{\theta^*}\cdot O\left(\left(\sqrt\frac{\mu^2\ln(n)}{n}+\frac{\sqrt{p\ln^5(n/\beta)\cdot\ln(1/\delta)\cdot\max\{\mu,\mu^2\}}}{\epsilon n}\right)\right).
\end{align*}
\label{thm:stocReg-ls}
\end{thm}

\begin{proof}
Consider the following regret function: $R_D(\aftrlls;\theta^*)=\frac{1}{n}\sum\limits_{t=1}^n\left((y_t-\ip{\theta_t}{\boldx_t})^2-(y_t-\ip{\theta}{\boldx_t})^2\right)$. We will bound $\mathbb{E}_D\left[R_D(\aftrlls;\theta^*)\right]$. Following the notation in the proof of Theorem~\ref{thm:regFTRL}, recall the following two functions.
\begin{itemize}
\item $\theta_{t+1}\leftarrow\argmin\limits_{\theta\in\calC}\underbrace{\sum\limits_{i=1}^t\left(\theta^\top \boldx_i\boldx_i^\top\theta-2y_i\ip{\boldx_i}{\theta}\right)+\frac{\lambda}{2}\ltwo{\theta}^2+\ip{\bfb_t}{\theta}+\theta^\top B_t\theta}_{\pJ(\theta)}$, 
where the noise $\bfb_t=\sum\limits_{i=1}^t y_i\boldx_i-\bolds_t$ with $\bolds_t$ being the output of $\gettb(\treeb,t)$, 
and the noise $B_t=W_t-\sum\limits_{i=1}^t\boldx_i\boldx_i^\top$ with $W_t$ being the output of $\gettq(\treeq,t)$. 
By standard bound on Gaussian random variables, w.p. at least $1-\beta$, $\forall t\in[n]$, $\ltwo{\bfb_t}=O\left( L\sigma\sqrt{p\ln(n)\ln(n/\beta)}\right)$ and $\ltwo{B_t}=O\left( L^2\sigma\sqrt{p\ln(n)\ln(n/\beta)}\right)$. We will use this bound to control the error introduced due to privacy.
\item $\thetash_{t+1} \leftarrow 
\argmin\limits_{\theta\in\calC}\underbrace{\sum\limits_{i=1}^t\left(\theta^\top \boldx_i\boldx_i^\top\theta-2y_i\ip{\boldx_i}{\theta}\right)+\frac{\lambda}{2}\ltwo{\theta}^2}_{\npJ(\theta)}$.
\end{itemize}
By an analogous argument to~\eqref{eq:2} in the proof of Theorem~\ref{thm:regFTRL}, we have
\begin{align}
\ltwo{\thetash_{t+1}-\theta_{t+1}}=O\left( \frac{L\sigma+L^2\sigma\cdot\mu}{\lambda}\cdot \sqrt{p\ln(n)\ln(n/\beta)}\right).\label{eq:4a}
\end{align}
Therefore,
\begin{align}
&\npJ(\thetash_{t+1})+\ip{\boldb_t}{\thetash_{t+1}}+\thetash_{t+1}^\top B_t\thetash_{t+1}\geq\underbrace{\npJ(\theta_{t+1})+\ip{\boldb_t}{\theta_{t+1}}+\theta_{t+1}^\top B_t\theta_{t+1}}_{\pJ(\theta_{t+1})}+\frac{\lambda}{2}\ltwo{\thetash_{t+1}-\theta_{t+1}}^2\label{eq:51}\\
\Rightarrow\quad& \npJ(\theta_{t+1})-\npJ(\thetash_{t+1})=O\left( \ltwo{\boldb_t}\cdot\ltwo{\thetash_{t+1}-\theta_{t+1}}+\ltwo{B_t}\cdot\ltwo{\thetash_{t+1}-\theta_{t+1}}\cdot\mu\right)\label{eq:52}\\
\Rightarrow\quad& \npJ(\theta_{t+1})-\npJ(\thetash_{t+1})=O\left((p\ln(n)\ln(n/
\beta)\sigma^2)\cdot\frac{L^2+L^4\mu^2+L^3\mu}{\lambda}\right)\label{eq:53}.
\end{align}
\eqref{eq:51} follows from the strong convexity of $\pJ$ and the fact that that $\theta_{t+1}$ is the minimizer of $\pJ$. \eqref{eq:52} follows from the bounds on $\ltwo{\boldb_t}$, $\ltwo{B_t}$,  and~\eqref{eq:4a}. We now use Theorem 2 from~\cite{SSSS09} to bound $\mathbb{E}_{(\boldx,y)\sim\calD}\left[(y-\ip{\boldx}{\theta_{t+1}})^2+\frac{\lambda}{2}\ltwo{\theta_{t+1}}^2\right]-\mathbb{E}_{(\boldx,y)\sim\calD}\left[(y-\ip{\boldx}{\theta^*})^2+\frac{\lambda}{2}\ltwo{\theta^*}^2\right]$ for any $\theta^*\in\calC$.

Using Theorem 2 from~\cite{SSSS09} and~\eqref{eq:53}, we have that w.p. at least $1-\beta$ over the randomness of the algorithm,
\begin{align}
&\mathbb{E}_{(\boldx,y)\sim\tau}\left[(y-\ip{\boldx}{\theta_{t+1}})^2+\frac{\lambda}{2}\ltwo{\theta_{t+1}}^2\right]-\mathbb{E}_{(\boldx,y)\sim\tau}\left[(y-\ip{\boldx}{\theta^*})^2+\frac{\lambda}{2}\ltwo{\theta^*}^2\right]\nonumber\\&\leq \frac{2}{t}\cdot\mathbb{E}\left[\npJ\left(\theta_{t+1}\right)-\npJ\left(\thetash_{t+1}\right)\right]+O\left(\frac{L^4\mu^2}{\lambda }\right)\nonumber\\
&=O\left((p\ln(n)\ln(n/\beta))\sigma^2)\cdot\frac{L^2+L^4\mu^2+L^3\mu}{\lambda t}+\frac{L^4\mu^2}{\lambda}\right).\label{eq:abc63}
\end{align}
\eqref{eq:abc63} immediately implies the following:
\begin{align}
\mathbb{E}_D\left[R_D(\aftrlls;\theta^*)\right]&=O\left(\left(p\ln^2(n)\ln(n/\beta)\sigma^2\right)\frac{L^2+L^4\mu^2+L^3\mu}{\lambda n}+\frac{L^4\mu^2}{\lambda}+\frac{\lambda\ln(n)}{n}\cdot\ltwo{\theta^*}^2\right).\label{eq:abc23}
\end{align}
We get the regret guarantee in Theorem~\ref{thm:stocReg-ls} by optimizing for $\lambda$.
\end{proof}

\subsection{Formal Statement of Online-to-batch Conversion for Excess Population Risk}
\label{app:pop}

\begin{thm}[Corollary to Theorem~\ref{thm:regFTRL} and~\cite{SSSS09}]
Recall the setting of parameters from Theorem~\ref{thm:regFTRL}, and let $\privT=\frac{1}{n}\sum\limits_{t=1}^n\theta_t$ (where $[\theta_1,\ldots,\theta_n]$ are outputs of Algorithm $\aftrl$ (Algorithm~\ref{Alg:PFTRL}). If the data set $D$ is drawn i.i.d. from the distribution $\tau$, then we have that w.p. at least $1-\beta$ (over the randomness of the algorithm $\aftrl$), 

$$\mathbb{E}_D\left[\poprisk{\privT}\right]=L\mu\cdot O\left(\sqrt\frac{{\ln(1/\beta)}}{n}+\sqrt{\frac{p^{1/2}\ln^{2}(1/\delta)\ln(1/\beta)}{\epsilon n}}\right).$$

Here, $\mu=\max\limits_{\theta\in\calC}\ltwo{\theta}$ is an upper bound on the norm of any model in $\calC$.
\label{thm:poprisk}
\end{thm}

\section{Multi-pass DP-FTRL: Handling multiple participations}
\label{app:multipass}
In this section, we provide details of the \Restart, \NoRestart, and \SomeRestart algorithms introduced in \cref{sec:practical_variants} for handling data where each example (or user in the case of user-level DP) can be considered in multiple training steps. The privacy accounting code is open sourced \footnote{\url{https://github.com/tensorflow/privacy/blob/master/tensorflow_privacy/privacy/analysis/tree_aggregation_accountant.py} for \Restart and \NoRestart dynamic programming. \url{https://github.com/google-research/DP-FTRL/privacy.py} for \NoRestart with given data order and tree completion trick. }.

\subsection{DP-FTRL with Tree-Restarts (\Restart)}
\label{sec:DP-FTRL-TR}

In this approach, we restart tree aggregation at every epoch of training, so each $d_i$ contributes at most once to each tree. Since this amounts to adaptive composition of Algorithm~\ref{Alg:PFTRL-ls} for $E$ times,  the privacy guarantee for this method can be obtained from Theorem~\ref{thm:privFTRL} and the adaptive sequential composition property of RDP \cite{mironov2017renyi}.
    
    \begin{thm}[Privacy guarantee for \Restart] If $\ltwo{\nabla_\theta \ell(\theta;d)}\leq L$ for all $d\in\calD$ and $\theta\in\calC$, then 
DP-FTRL (Algorithm~\ref{Alg:PFTRL}) with Tree Restart (\Restart) for $E$ epochs satisfies $\left(\alpha,\frac{\alpha E L^2\lceil\lg(n+1)\rceil}{2\sigma^2}\right)$-RDP. Correspondingly, by setting $\sigma=\frac{\sqrt{E L^2\lceil\lg(n+1)\rceil\ln(1/\delta)}}{\epsilon}$ one can satisfy $(\epsilon,\delta)$-differential privacy guarantee, as long as $\epsilon\leq 2\ln(1/\delta)$.
\label{thm:privFTRL_restart}
\end{thm}

When the tree-restart strategy is used, we can use a tree-completion trick to improve performance. The general idea is to add ``virtual steps'' to complete the binary tree, such that the noise added to the step before restarting is smaller. The details can be found at Appendix~\ref{sec:tree completion}.

\subsection{DP-FTRL without Tree Restarts (\NoRestart)}
\label{sec:DP-FTRL-NTR}

In this case, we build a single binary tree over all the iterations of training. In this section we provide the privacy accounting for this \NoRestart approach, which perhaps surprisingly becomes much more involved compared to the tree restart version. We make the following three assumptions: i) A singe training example  contributes only~\emph{once} to a single gradient computation, ii) an example can contribute at most $\epochs$ number of times during the training process, and iii) any two successive appearance of a single training example is ensured to have a \emph{minimum separation} of $\participation$ iterations of $\aftrl$ (Algorithm~\ref{Alg:PFTRL}). For the clarity of notation, here we will denote the number of iterations of Algorithm $\aftrl$ with $T$ (instead of $n$), and each $\nabla_t$ for $t\in[T]$ corresponds to gradient computation at time step $t$. Additionally, we will refer to be binary tree used in Algorithm~\ref{Alg:PFTRL} (Algorithm $\aftrl$), with the leaf nodes being the $\nabla_t$'s, as $\calT$.

In Algorithm~\ref{Alg:accnr1}, we provide a privacy accounting scheme for the above instantiation of Algorithm $\aftrl$. Later, we provide a tighter privacy accounting via a dynamic programming approach, albeit at a higher computation cost. In all the privacy analysis in this section, we essentially use the standard machinery of Gaussian mechanism with RDP~\cite[Proposition 7]{mironov2017renyi}, except we use a specific sensitivity analysis for the tree-aggregation algorithm used in DP-FTRL.

\begin{algorithm}[ht]
\caption{Privacy accounting for \NoRestart}
\begin{algorithmic}[1]
\REQUIRE Binary tree for the DP-FTRL algorithm: $\tree$, minimum separation: $\participation$, maximum contributions: $E$.
\STATE $\rho\leftarrow 0$
\FOR{each level $\level\in\tree$}

    \STATE $k_\gamma\leftarrow$ \# of nodes in level $\gamma$,  $\mu_\gamma\leftarrow\left\lceil \frac{2^{{\sf height}(\gamma)}}{\xi+1}\right\rceil$, and $k^*_\gamma\leftarrow\min\left\{k,\left\lfloor\frac{E}{\mu_\gamma}\right\rfloor\right\}$.
    \STATE $\zeta_\gamma\leftarrow k^*_\gamma\cdot\mu_\gamma^2+\mathbb{1}_{k^*_\gamma < k}\left(E-k^*_\gamma\cdot\mu_\gamma\right)^2$
    \STATE $\rho\leftarrow\rho+ \zeta_\gamma$.

\ENDFOR
\STATE {\bf return} $\rho$.
\end{algorithmic}
\label{Alg:accnr1}
\end{algorithm}

\begin{thm}
For the \NoRestart instantiation of Algorithm~\ref{Alg:PFTRL} (Algorithm $\aftrl$), the privacy accounting scheme in Algorithm~\ref{Alg:accnr1} ensures $\left(\alpha,\frac{\alpha}{2\sigma^2}\cdot\rho\right)$-Renyi differential privacy (RDP), where $\sigma$ is the noise scale in the tree aggregation scheme of Section~\ref{app:tree}.
\label{thm:acc}
\end{thm}

\begin{proof}

The proof proceeds by bounding the RDP cost for the release of all the noised node values at a given level $\level \in \tree$. By slight abuse of notation, we treat $\level$ as the set of nodes at that level, $\level = \left\{z_1,\ldots,z_k\right\}$. For any node $z \in \level$, let $\bfS_{z}\leftarrow \sum\limits_{i\in\leaf_z}\nabla_i$, where $\leaf_z$ is the set of leaf nodes in the subtree rooted at $z$ and $\nabla_i$ is the gradient of a single training example. We will analyze the RDP privacy cost of the query $\bfS_\level = \langle\bfS_{z_1}, \dots, \bfS_{z_k}\rangle$ with the Gaussian mechanism. The key is to bound the $\ell_2$ sensitivity of $\bfS_\level$ and compare that to the noise added.

Fix an example $x$ in the training data set and let $ c = \langle c_1,\ldots,c_k\rangle$ be the number of contributions of $x$ to each of the nodes $z_1, \dots, z_k$. We observe immediately that the $\ell_2$ sensitivity of $\bfS_\level$ is simply $\norm{L c}_2$, where $L$ is the $\ell_2$-Lipschitz constant of the individual loss functions, so it is sufficient to bound $\norm{c}^2_2$. We note $c$ is bound by two constraints:
\begin{itemize}
\item  By definition of $E$, $\sum\limits_{i=1}^k c(z_i) \le E$.
\item For any node $z \in \tree$, the number of leaves in the subtree rooted at $z$ is at most $2^{\sf height(z)}$. Since, each example is allowed to participate every $(\participation+1)$ steps, the number of contributions $c_i$ to any $z_i \in \gamma$ is upper bounded by  $\mu_\gamma = \left\lceil\frac{2^{\sf height(\gamma)}}{\participation+1}\right\rceil$.
\end{itemize}
Applying these two constraints, the following quadratic program can be used to bound  $\norm{c}^2_2$:
\begin{align}
    &\max\limits_{c_1,\ldots,c_k}\sum_{i=1}^k c_i^2,\nonumber\\
    &\text{s.t.} \sum\limits_{i=1}^k c_i \leq \epochs,\nonumber\\
    &\forall i\in[k], 0\leq c_i\leq \mu_\gamma
    \label{eq:ld1}
\end{align}

By KKT conditions, the above is maximized as follows: Letting $k^*=\min\left\{k,\left\lfloor\frac{E}{\mu_\gamma}\right\rfloor\right\}$, set $c_{1}=\mu_\gamma$, $c_{2}=\mu_\gamma,\ldots,c_{k^*}=\mu_\gamma$, set $c_{k^*+1}=E-k^*\mu_\gamma$, and set $c_i=$ for all $i>(k^*+1)$. 

Thus, $\norm{c}^2_2 \le \zeta$, and the $\ell_2$ sensitivity of $\bfS_\gamma$ is bounded by $L \sqrt{\zeta}$. Recall from \cref{app:tree} we add noise $\calN(0, L^2 \sigma^2)$ to each node, and hence to $\bfS_\gamma$. From these to facts and the RDP properties of Gaussian mechanism~\cite[Proposition 7]{mironov2017renyi}, we have that for a given order $\alpha$,
\[\frac{\alpha L^2 \zeta}{2 L^2 \sigma^2} = \frac{\alpha\zeta}{2\sigma^2}\] 
is the RDP cost of the complete level $\gamma$.

To complete the proof, it suffices to apply adaptive RDP composition across all the levels of the tree $\tree$.
\end{proof}

\mypar{Dynamic programming based improvement to Algorithm~\ref{Alg:accnr1}} Consider the same binary tree $\tr$ described above. For any data set $D$, for any fixed individual $x$, and a node $z\in\tr$, let $c(z)$ be the total number of participation of $x$ in the sub-tree rooted at $z$. Extending the level-wise argument of the proof of \cref{thm:acc} to the whole tree (as a single application of the Gaussian mechanism), it is not hard to see that DP-FTRL satisfies $\left(\alpha,\frac{\alpha}{2\sigma^2} \cdot \sum_{z\in\tr} c(z)^2\right)$-RDP guarantee. We can upper-bound this by computing $\zeta = \max_c \sum\limits_{z \in \tree} c(z)^2$ subject to $c$ being a realizable assignment under the participation constraints mentioned above. We provide the following program to compute an upper bound on $\zeta$.

\begin{align}
    \zeta\left(\contrib,\start,\en,\size\right)&=\contrib^2\cdot\mathbb{I}\left[\size\ \text{is a power of}\ 2\right]+\nonumber\\
    &\max\limits_{\substack{i\in\{0,\ldots,\contrib\}\\ j\in\{0,\ldots,\xi\}}}\left\{\zeta\left(\contrib-i,\start,j,k\right)+\zeta\left(i,j,\en,\size-k\right)\right\}\nonumber\\
    &\text{where}\ k\ \text{largest power of}\ 2<\size.
    \label{eq:dpsl}
\end{align}
The base cases is:

\begin{align}
    &\text{\bf if}\ \start+\contrib\cdot(\xi+1)>\size+\en,\ {\bf then}\  \zeta(\cdot,\start,\en,\size)=-\infty\nonumber\\
    &\zeta(0,\cdot,\cdot,\cdot) = 0\nonumber\\
    &\zeta(1,\cdot,\cdot,1) = 1
    \label{eq:base}
\end{align}

\begin{thm}
Let $\zeta^*=\max\limits_{w\in\{0,\ldots,\epochs\}}\zeta(w,0,\xi,T)$ in~\eqref{eq:dpsl}. Above privacy accounting for the \NoRestart instantiation of Algorithm~\ref{Alg:PFTRL} (Algorithm $\aftrl$) ensures $\left(\alpha,\frac{\alpha\zeta^*}{2\sigma^2}\right)$-R\'enyi differential privacy.
\label{thm:accdp}
\end{thm}

\begin{proof}
We argue that $\zeta$ in~\eqref{eq:dpsl} enumerates all the valid configurations of valid cost functions $c$ described above. 

First, suppose $\size$ is a power of two. Then there is a complete binary tree with $\size$ leaves which we think of as being numbered sequentially. We must place $\contrib$ contributions at leaves of this tree, with the constraint that there are at least $\xi$ empty leaves following each leaf with a contribution (and, of course, each leaf can only have one contribution). However, the empty leaves after the last contribution can ``overflow'' by $\en$, i.e., we can imagine $\en$ extra empty spaces are added at the end of the tree to satisfy this constraint.  In addition the first $\start$ leaves must be empty. In other words, $\start$ spaces are subtracted from the beginning and $\en$ spaces are added to the end of the tree.
The value of the tree is the sum over all nodes (leaves and internal) of the square of the number of contributions in the subtree rooted at that node. So a leaf has value $1$ or $0$ depending on whether or not it has a contribution and an internal node has value $c^2$, where $c$ is the number of contributions at leaves under this node. The function $\zeta\left(\contrib,\start,\en,\size\right)$ computes the maximum total value of the tree over an arbitrary placement of contributions satisfying the constraints. 
Second, if $\size$ is not a power of two, then instead of a single complete binary tree, we have multiple complete binary trees of different sizes, which are arranged from largest to smallest.~\eqref{eq:dpsl} exactly encodes this logic.

The base cases are immediate given the above description of the recursive statement. This completes the proof.
\end{proof}

\subsubsection{Analysis for a Simpler Case: \NoRestart with Given Data Order}
\label{sec:privacy given order}
The previous analysis assumes only a gap between two successive participation of any user and does not constrain the data order in any other way.
In some special cases, especially the centralized setting, the server takes full control of the training process and can decide which samples to use in each training step. We can thus take advantage of such knowledge for a simpler (in terms of the privacy accounting algorithm but not necessarily the computation time) privacy computation.

Suppose we use stochastic gradient with mini-batch of size $1$. Then for each node in the tree, we can obtain a list of samples that affect the node, which can be used to compute the sensitivity of the node with respect to each sample. Summing up the (squared) sensitivity of every node, we can then take maximum among all samples and get the final (squared) sensitivity. The algorithm for the computing the squared sensitivity is described in Algorithm~\ref{Alg:privacy given order}.

\begin{algorithm}[ht]
\caption{Privacy accounting for \NoRestart with given data order}
\begin{algorithmic}[1]
\REQUIRE A given data order $[s_1, \dots, s_m]$ where $s_i \in [n]$ with $n$ being the number of samples.
\STATE $\rho_i \leftarrow 0$ for $i \in [n]$ \COMMENT{$\rho_i$ will record the squared sensitivity with respect to sample $i$}
\STATE Construct the leaf layer of the tree $\calL_1 = [(s_1), \dots,(s_m)]$
\STATE $\rho_{s_i} \leftarrow 1$ for $s_i = s_1, \dots s_m$
\FOR{$\level = 2, \dots, \left\lfloor \log_2(m) \right\rfloor + 1$}
\STATE Construct $\calL_{\level}$ by concatenating every two consecutive nodes in $\calL_{\level -1}$, ignoring the rest if $|\calL_{\level-1}|$ is not even
\FOR{every node in $\calL_{\level}$}
\FOR{$i = 1,\dots, n$}
\STATE $\rho_i \leftarrow \rho_i + c^2$, where $c$ is the count of $i$ in the node
\ENDFOR
\ENDFOR
\ENDFOR
\STATE {\bf return} $\max\limits_{i\in [n]}\rho_i$.
\end{algorithmic}
\label{Alg:privacy given order}
\end{algorithm}

Consider a simple example where we have four samples indexed with $1, \dots, 4$, and use them in the order of $[1, 2, 3, 1, 4]$ during training. Representing each node in the tree with a list of samples it contains, we will have five leaf nodes $[(1), (2), (3), (1), (4)]$, two nodes $[(1, 2), (3, 1)]$ in the middle layer, and $(1, 2, 3, 1)$ as the root node. The root node, for example, have sensitivity $2$ with respect to sample $1$, $1$ with respect to $2$ and $3$, and $0$ with respect to $4$. Summing up over all nodes, the whole tree has sensitivity $\sqrt{8}$ with respect to sample $1$, $\sqrt{3}$ with respect to $2$ and $3$, and $1$ with respect to $4$. We take the maximum and conclude that the tree has sensitivity $\sqrt{8}$.

When mini-batch has size larger than $1$, if the batches are formed by the same set of samples across epochs such that one sample affects the same batch, we can simply consider sensitivity with respect to a batch, i.e., $i$ could be used to represent the $i$-th batch instead of the $i$-th sample. If samples might not be grouped in the same way across epochs, then each leaf node would consist of multiple samples that constitute the corresponding batch, and the same computation follows.

Now we consider the time complexity.
We need to construct a tree with nodes being lists of samples / batches. Suppose we build the tree layer by layer from leaf to root. Merging two consecutive nodes to form their parent takes $O(m_i + m_j)$ complexity where $m_i$ and $m_j$ are their sizes. Therefore, forming a layer based on the previous layer takes $O(m)$, where $m$ is the total number of samples / batches across epochs, depending on whether the batches are always formed in the same way through training. The construction of the whole tree thus takes $O(m \log(m))$ time.
The sensitivity computation takes $O(m_i)$ for a node of size $m_i$, and thus enumerating through the tree takes $O(m \log(m))$ as well.
Therefore, the total time complexity is $O(m \log(m))$.
The running time is higher than that of Algorithm~\ref{Alg:accnr1}, but potentially smaller than the dynamic programming version of it.

In our centralized learning experiments,
we shuffle the dataset right after each restarting, then keep the same batches and go through them in the same order across epochs. We use the above privacy computation.

\subsection{Combining \Restart and \NoRestart}
\label{sec:DP-FTRL-SometimesRestart}

In Section~\ref{sec:empEval}, we consider the case where the tree is restarted after every epoch. Each user thus participate only in one leaf of the tree in the privacy analysis. However, restarting may cause instability and hurt model utility. On the other hand, if we keep using the same tree through the training process with the privacy analysis in Appendix~\ref{sec:DP-FTRL-NTR}, each user can affect multiple leaf nodes. The sensitivity is thus high and larger noise is needed under a fixed privacy budget.
Given the trade-off, a natural schedule to consider is to restart the tree every few epochs. For example, if we would like to train for $100$ epochs on CIFAR-10 ($n=50000$) with batch size $500$ and achieve $\epsilon \approx 23.0$, we can either restart every epoch with noise multiplier $\approx 7$, restart every $5$ epochs with noise $\approx 8.5$, restart every $20$ epochs with noise $\approx 12$, or use a single tree (no-restart) with noise $32$.
The privacy computation follows from Appendix~\ref{sec:privacy given order}.

\subsubsection{A Binary Tree Completion Trick and the Privacy Analysis}
\label{sec:tree completion}
Looking at the binary tree of noisy gradients, we can easily see that the amount of noise added to each prefix sum depends on its location in the tree. Specifically, if we consider the original tree-aggregation protocol, to release the prefix sum up to leaf node $i$, the noise will scale with the number of bits that are set in the binary representation of $i$.
A natural trick to consider is thus to complete the tree with ``virtual steps'' such that the noise is the smallest. 
Consider one of the settings in Figure~\ref{fig:acc_privacy} -- running CIFAR-10 ($n=50000$) with batch size $2000$. After one epoch, the tree consists of $25$ nodes. Instead of restarting immediately, we can instead run an additional $7$ ``virtual steps'' so that the tree is complete with $32$ nodes. The virtual steps can be thought of as adding virtual samples of $0$s, i.e., instead of privately releasing the sum of the past $25$ gradients, we privately release the sum of the $25$ gradients and $7$ $0$s. 
This way, the additive noise in the last step can improve by a factor of $3$ in the original tree-aggregation protocol, and by roughly a factor of $4$ ($\approx2.05$ vs. $\approx0.51$) with the trick from~\cite{honaker2015efficient}.
This can be crucial as the future models will be based on the last noise before restarting.

The virtual steps do not come for free as we have added more nodes in the tree, yet the cost of privacy is less than that of actual gradients, because the virtual samples are fixed to $0$ and do not increase the sensitivity. 
For example, if we use the privacy computation in Algorithm~\ref{Alg:privacy given order} where the complete ordering of samples / batches is known and aim to add $A$ virtual steps, we can append $A$ virtual leaf nodes in the end of the leaf layer with, for example, a special symbol $\star$ to indicate it is a virtual sample. The tree construction steps is exactly the same as before. The only difference is in the sensitivity computation, where we simply ignore the virtual sample $\star$.
Taking the example in Appendix~\ref{sec:privacy given order} where the training uses samples $[1,2,3,1,4]$, we might add $3$ virtual steps to form a tree of $8$ leaf nodes $[(1),(2),(3),(1),(4),(\star),(\star),(\star)]$, $4$ nodes $[(1,2),(3,1),(4,\star),(\star,\star)]$ in the 2nd layer, $[(1,2,3,1),(4,\star,\star,\star)]$ in the 3rd layer and $(1,2,3,1,4,\star,\star,\star)$ as the root node.
This way, sample $1$ would have sensitivity $\sqrt{12}$, sample $2$, $3$ and $4$ would have sensitivity $2$.

There is a trade-off between the scale of the last noise before restarting and the additional privacy cost. For example, we can expect that if the size of the tree is far away from the next power of two, then the additional privacy cost might overwhelm the gain in the noise scale; if the tree is almost complete, then we might expect the trick to help.

\section{Omitted Details for Experiment Setup (Section~\ref{sec:exp_setup})}
\label{app:expt_setup}

\subsection{Additional Details on Model Architectures}
\label{app:mod_archs}
Table~\ref{table:mnist_nn} shows the model architecture for MNIST and EMNIST, Table~\ref{table:cifar10_nn} shows that for CIFAR-10, and Table~\ref{table:stackoverflow_nn}
 shows the neural networks adopted from \citep{reddi2020adaptive}.
 
\begin{table*}[!htbp]
\caption{Model architectures for all experiments.}
\begin{subtable}[t]{0.5\textwidth}
\caption{Model architecture for MNIST and EMNIST.}
\label{table:mnist_nn}
\centering
\begin{tabular}{ c c }
\toprule
Layer & Parameters \\
\midrule
Convolution & 16 filters of $8 \times 8$, strides 2 \\
Convolution & 32 filters of $4 \times 4$, strides 2 \\
Fully connected & $32$ units \\
Softmax & - \\
\bottomrule
\end{tabular}
\end{subtable}
~
\begin{subtable}[t]{0.45\textwidth}
\caption{Model architecture for CIFAR-10.}
\label{table:cifar10_nn}
\centering
\begin{tabular}{ c c }
\toprule
Layer & Parameters \\
\midrule
Convolution  $\times 2$ & 32 filters of $3 \times 3$, strides 1\\
Max-Pooling & $2 \times 2$, stride 2 \\
Convolution  $\times 2$ & 64 filters of $3 \times 3$, strides 1 \\
Max-Pooling & $2 \times 2$, stride 2 \\
Convolution  $\times 2$ & 128 filters of $3 \times 3$, strides 1 \\
Max-Pooling & $2 \times 2$, stride 2 \\
Fully connected & $128$ units \\
Softmax & - \\
\bottomrule
\end{tabular}
\end{subtable}
\centering
\begin{subtable}[t]{0.5\textwidth}
\caption{Model architecture for StackOverflow. \citep{reddi2020adaptive}}
\label{table:stackoverflow_nn}
\centering
\begin{tabular}{ c c c }
\toprule
Layer & Output Shape & Parameters \\
\midrule
Input & 20 & 0\\
Embedding & (20, 96) & 960384 \\
LSTM & (20,670) & 2055560 \\
Dense & (20, 96) & 64416 \\
Dense & (20, 10004) & 970388 \\
Softmax & - & -\\
\bottomrule
\end{tabular}
\end{subtable}
\end{table*}

\subsection{Comparison of Optimizers with their Momentum Variants}
\label{app:mom_comp}

\mypar{Centralized Learning}
Figures~\ref{fig:image momentum sgd} and \ref{fig:image momentum ftrl} show a comparison between the original and the momentum versions of DP-SGD (denoted by ``DP-SGD'' and ``DP-SGDM'') and DP-FTRL (denoted as ``DP-FTRL'' and ``DP-FTRLM''), respectively.
For both DP-SGD and DP-FTRL, we consider both small and large number of epochs (presented in the top and bottom row respectively) on the three centralized example-level DP image classification tasks. 
The small-epoch setting follows from that in Section~\ref{sec:empEval}.
The large-epoch setting follows from that in Appendix~\ref{sec:interleaving}. The number of epochs is $20$, $100$ and $50$ for MNIST, CIFAR-10 and EMNIST under the smaller batch size and is 4 times that for larger the batch size. For DP-FTRL(M), we use the tree-completion trick~\ref{sec:tree completion}. On CIFAR-10, we restart every $5$ epochs, and on EMNIST, we restart every epoch -- following the setting that achieves the highest accuracy in Figure~\ref{fig:different restart} in Appendix~\ref{sec:interleaving}.

In Figure~\ref{fig:image momentum sgd}, for CIFAR-10, DP-SGDM outperforms DP-SGD for smaller number of epochs, and DP-SGD is better for larger number of epochs. For the other settings, the two variants are similar.
In Figure~\ref{fig:image momentum ftrl}, we can see that the accuracy of DP-FTRLM is always at least that of DP-FTRL (sometimes even more).

Therefore, we use DP-SGDM in Section~\ref{sec:empEval} (small-epoch setting) and DP-SGD in Appendix~\ref{sec:interleaving} (large-epoch setting). We use DP-FTRLM in both settings.

\begin{figure}[!htbp]
\centering
\captionsetup[subfigure]{justification=centering}
\begin{subfigure}[b]{0.3\textwidth}
\centering
\includegraphics[width=\textwidth]{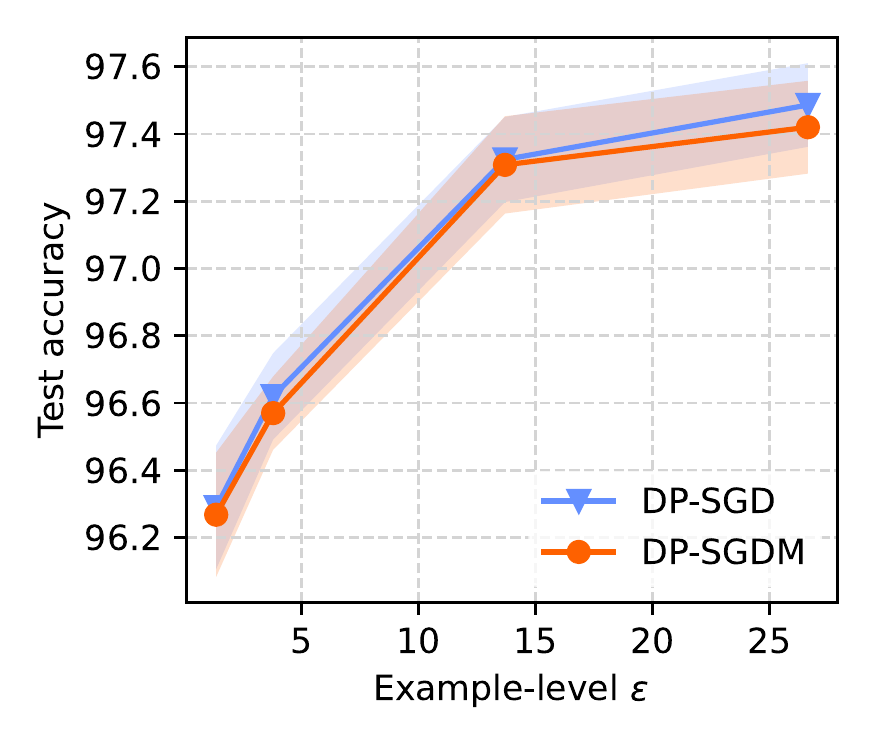}
\caption{MNIST. $5$ epochs.}
\end{subfigure}
\begin{subfigure}[b]{0.3\textwidth}
\centering
\includegraphics[width=\textwidth]{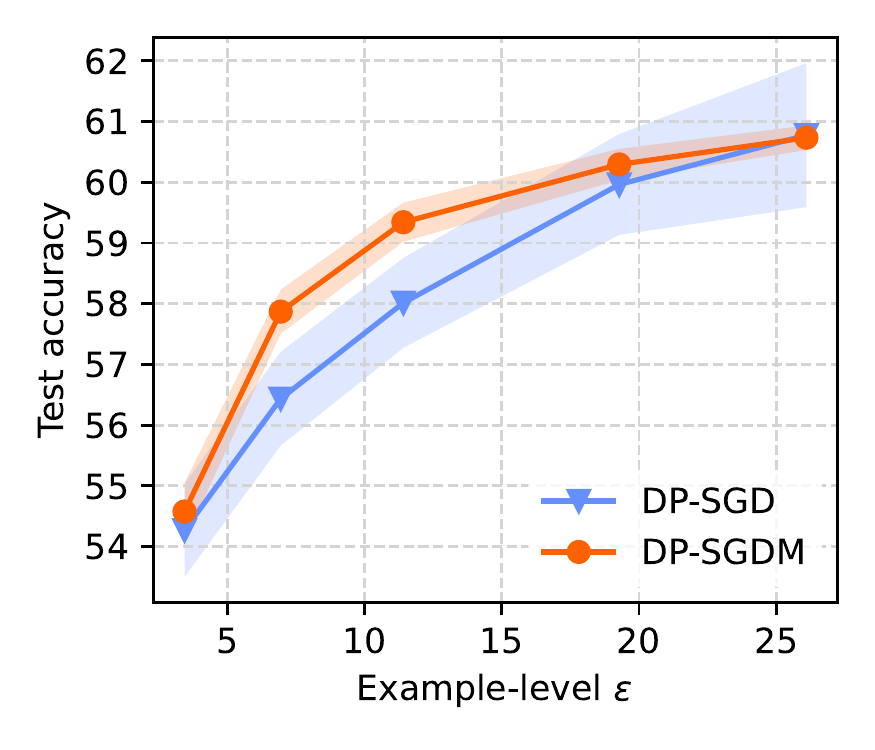}
\caption{CIFAR-10. $5$ epochs.}
\end{subfigure}
\begin{subfigure}[b]{0.3\textwidth}
\centering
\includegraphics[width=\textwidth]{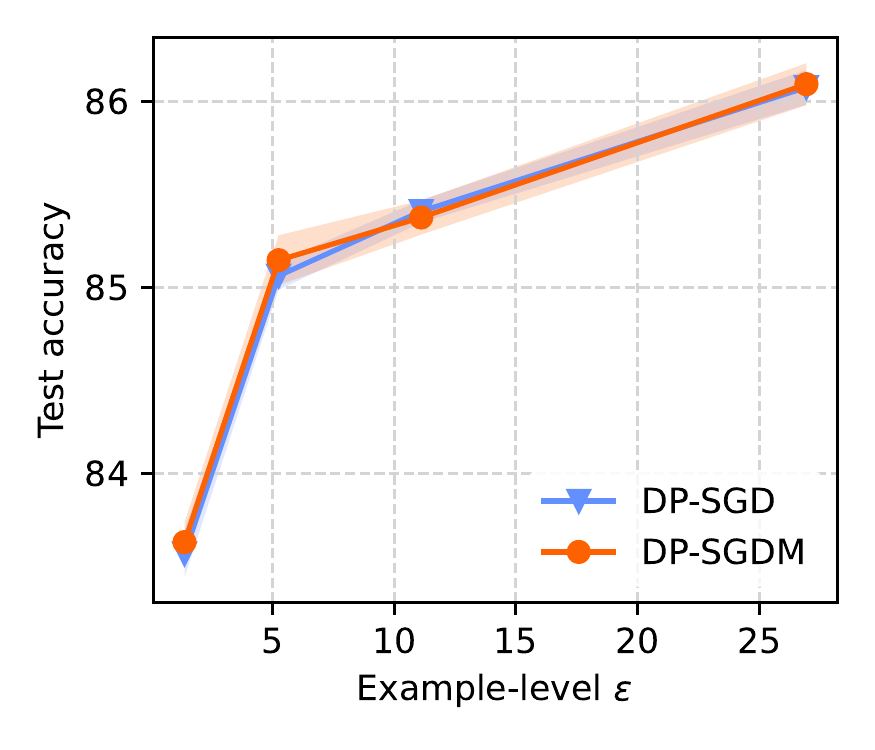}
\caption{EMNIST. $5$ epochs.}
\end{subfigure}
\begin{subfigure}[b]{0.3\textwidth}
\centering
\includegraphics[width=\textwidth]{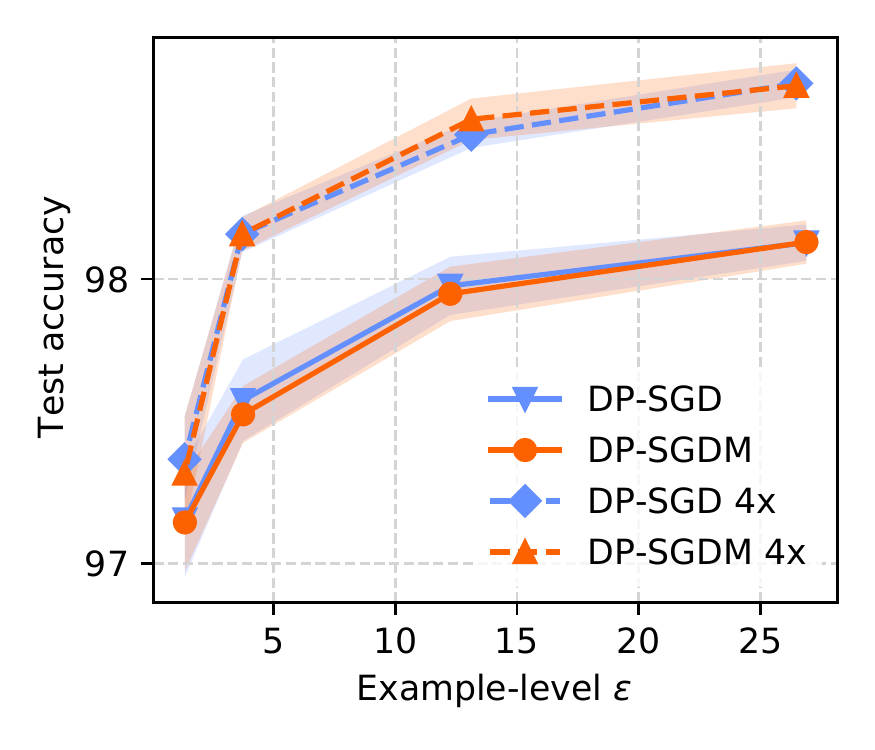}
\caption{MNIST, $20$ / $80$ epochs for smaller / larger batch.}
\end{subfigure}
\begin{subfigure}[b]{0.3\textwidth}
\centering
\includegraphics[width=\textwidth]{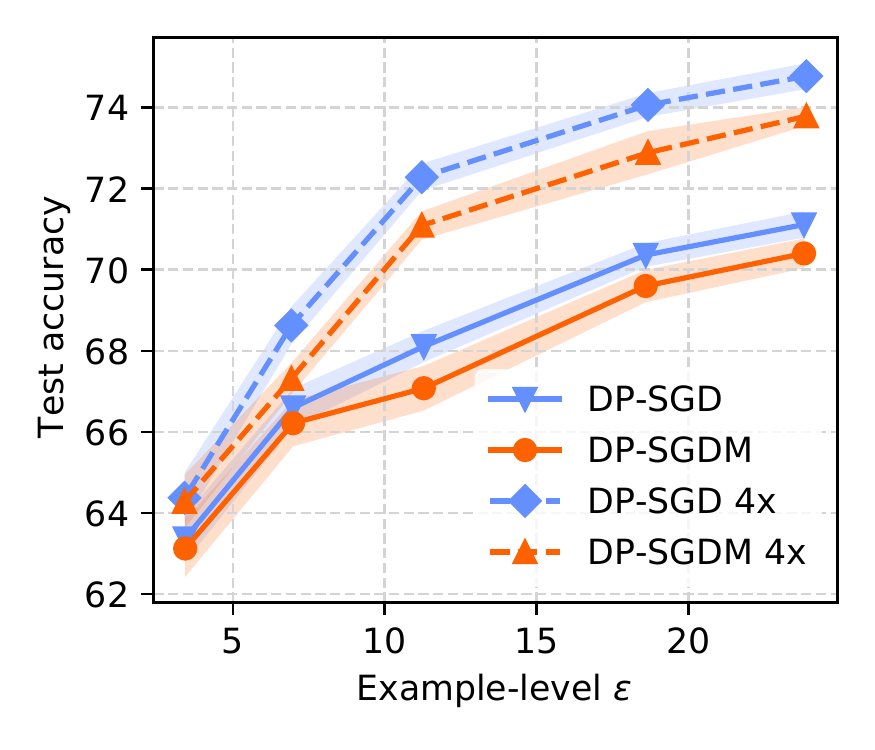}
\caption{CIFAR-10, $100$ / $400$ epochs for smaller / larger batch.}
\end{subfigure}
\begin{subfigure}[b]{0.3\textwidth}
\centering
\includegraphics[width=\textwidth]{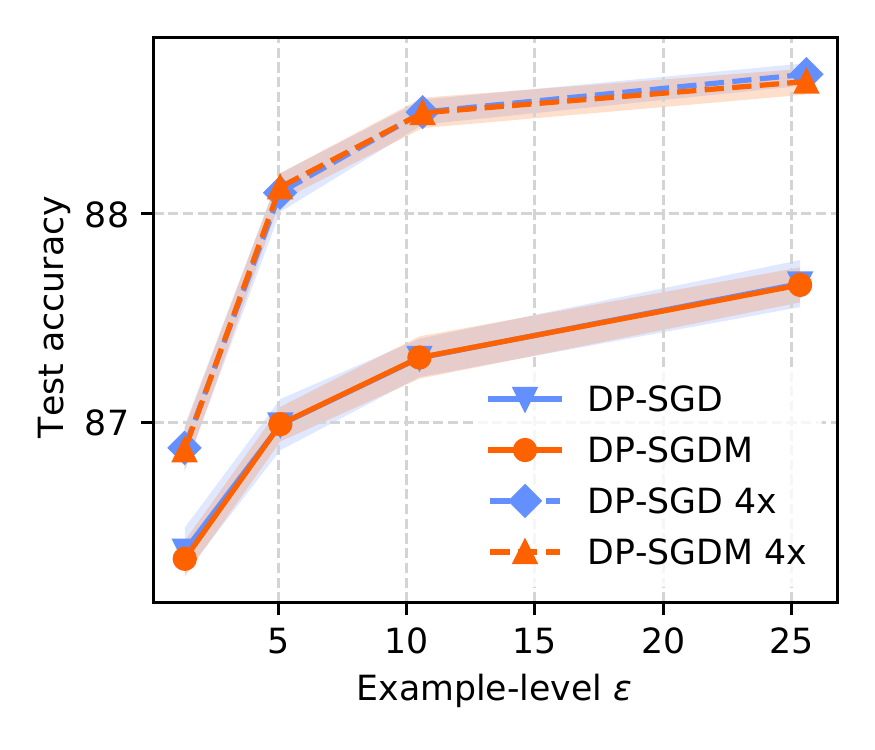}
\caption{EMNIST, $50$ / $200$ epochs for smaller / larger batch.}
\end{subfigure}
\caption{
Effect of momentum in DP-SGD. Final test accuracy vs.\ privacy (example-level $\epsilon$) for various noise multipliers. 
For smaller number of epochs (top row), DP-SGDM is better for CIFAR-10; for larger number of epochs (bottom row), DP-SGD is better for CIFAR-10. In the other cases, the two variants are similar.}
\label{fig:image momentum sgd}
\end{figure}

\begin{figure}[!htbp]
\centering
\captionsetup[subfigure]{justification=centering}
\begin{subfigure}[b]{0.3\textwidth}
\centering
\includegraphics[width=\textwidth]{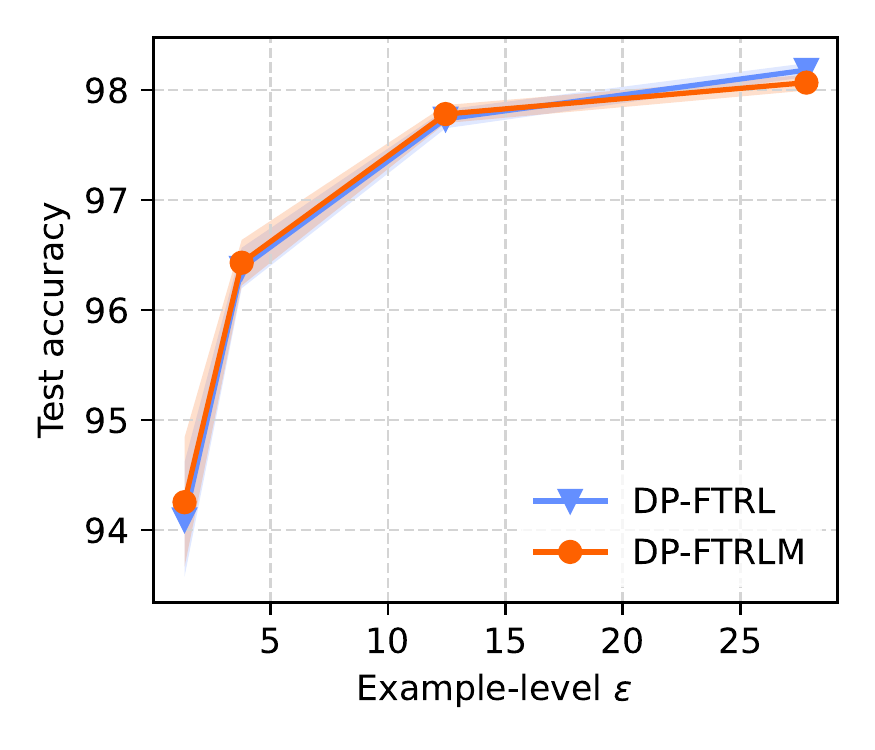}
\caption{MNIST. $5$ epochs.}
\end{subfigure}
\begin{subfigure}[b]{0.3\textwidth}
\centering
\includegraphics[width=\textwidth]{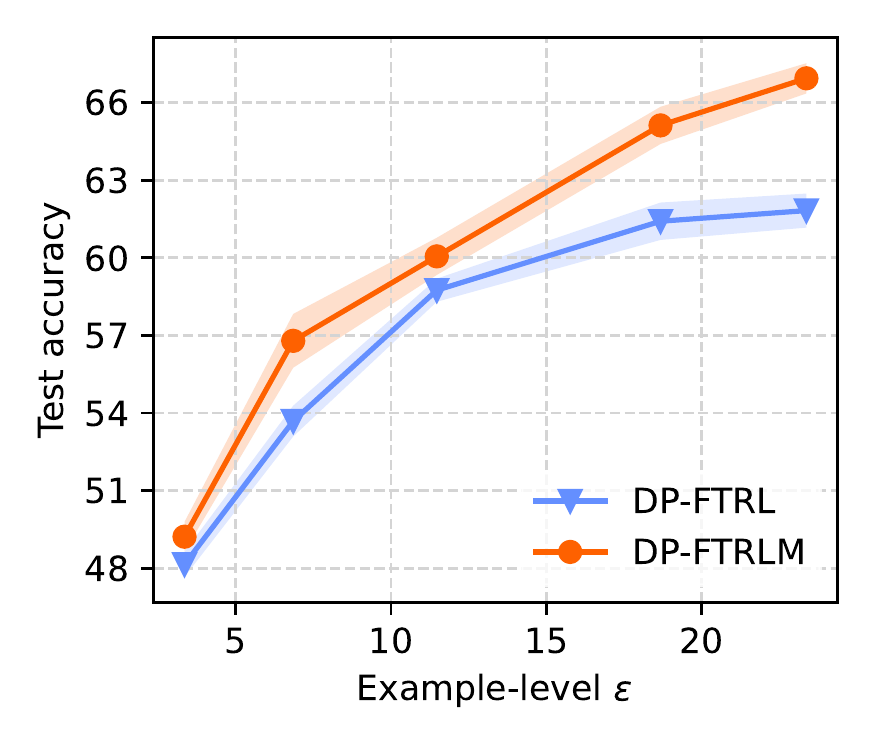}
\caption{CIFAR-10. $5$ epochs.}
\end{subfigure}
\begin{subfigure}[b]{0.3\textwidth}
\centering
\includegraphics[width=\textwidth]{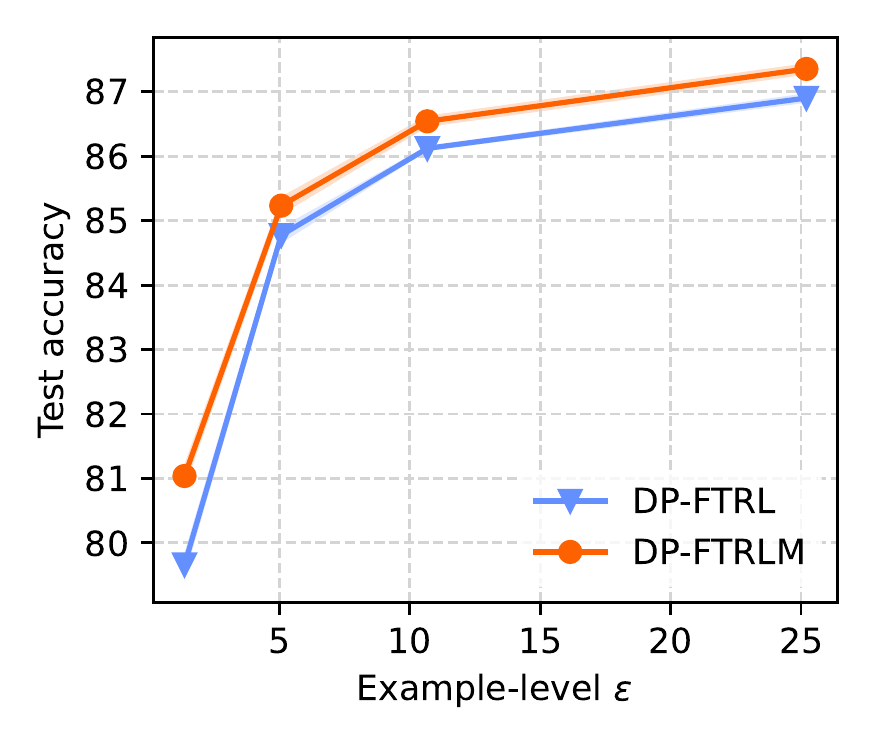}
\caption{EMNIST. $5$ epochs.}
\end{subfigure}
\begin{subfigure}[c]{0.3\textwidth}
\centering
\includegraphics[width=\textwidth]{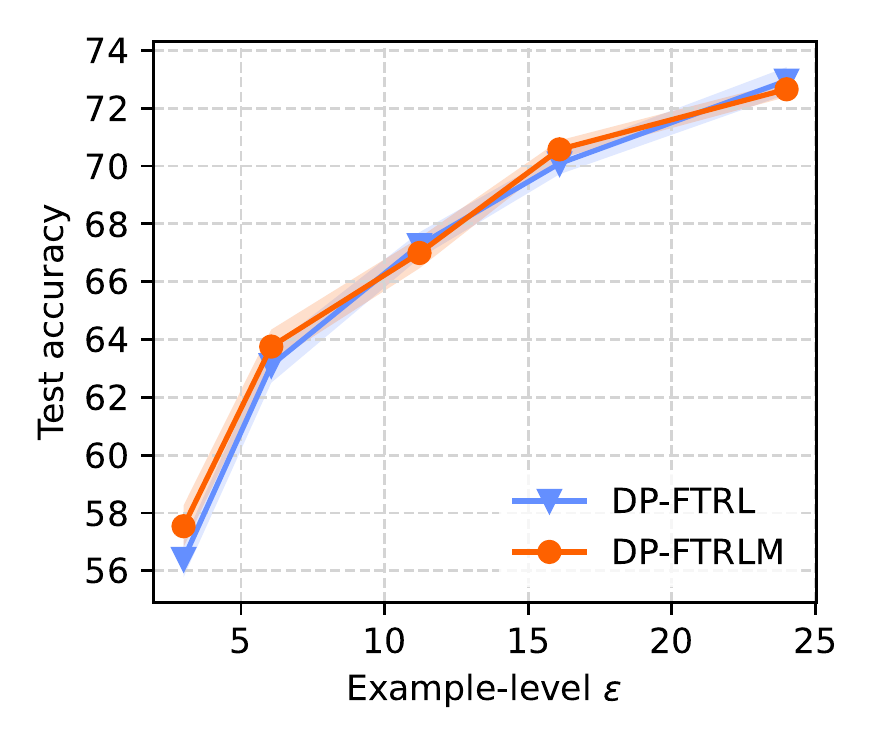}
\caption{CIFAR-10. $100$ epochs.}
\end{subfigure}
\begin{subfigure}[c]{0.3\textwidth}
\centering
\includegraphics[width=\textwidth]{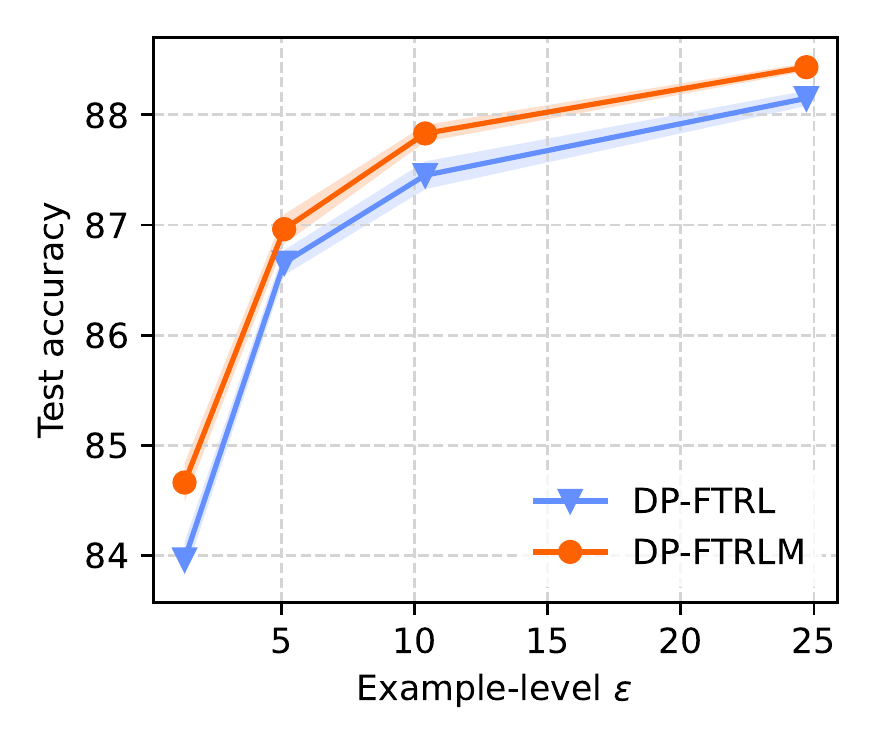}
\caption{EMNIST. $50$ epochs.}
\end{subfigure}
\caption{Effect of momentum on DP-FTRL. Final test accuracy vs.\ privacy (example-level $\epsilon$) for various noise multipliers. DP-FTRLM outperforms DP-FTRL.}
\label{fig:image momentum ftrl}
\end{figure}

\mypar{Federated Learning}
The experiments in \cref{tab:fl-momentum} and \cref{fig:fl-momentum} show the advantages of the momentum variant for the federated StackOverflow task in practice. We compare DP-SGD and its momentum variant DP-SGDM, DP-FTRL and its momentum variant DP-FTRLM under two different privacy epsilons. Privacy epsilon is infinite when noise multiplier is zero; privacy epsilon is 8.53 when noise multiplier is 0.4 for DP-SGD and DP-SGDM; privacy epsilon is 8.5 when noise multiplier is 2.33 for DP-FTRL and DP-FTRLM. We tune and select the hyperparameter with the best validation accuracy \footnote{The accuracy for StackOverflow next word prediction task excludes the end of sequence symbol and the out of vocabulary symbol following \citep{reddi2020adaptive}. The hyperparameters tuning range are described in \cref{app:hyp_tun_privTarget}.}. We then run the experiment with the specific set of hyperparameters for five times to estimate mean and standard deviation of the accuracy. 

\begin{table}[!htbp]
    \centering
    \begin{tabular}{|c|c|c|c|S|S|S|}
    \hline
     \multirow{2}{*}{Server Optimizer} & \multirow{2}{*}{Epsilon} & \multicolumn{2}{c|}{Accuracy} & \multicolumn{3}{c|}{Hyperparameters} \\ 
     \cline{3-7} & & Validation  & Test & \text{ServerLR} & \text{ClientLR} & \text{Clip}  \\ 
    \hline
    DP-SGD & \multirow{4}{*}{$\infty$} & 19.62 $\pm$ .12 & 20.99 $\pm$ .11 & 3 & .5 & 1\\
    DP-SGDM & & 23.87 $\pm$ .22 & 24.89 $\pm$ .27 & 3 & .5 & 1\\
    DP-FTRL & & 19.95 $\pm$ .05 & 21.12 $\pm$ .14 & 3 & .5 & 1\\
    DP-FTRLM & & 23.89 $\pm$ .03 & 25.15 $\pm$ .07 & 3 & .5 & 1\\
    \hline
    DP-SGD & \multirow{2}{*}{7.51} & 16.83 $\pm$ .05 & 18.25 $\pm$ .05 & 3 & .5 & .3 \\
    DP-SGDM & & 16.92 $\pm$ .03 & 18.27 $\pm$ .04 & .1 & .5 & 1\\
    DP-FTRL & \multirow{2}{*}{7.83} & 15.04 $\pm$ .16 & 15.46 $\pm$ .39 & 3 & .5 & .3\\
    DP-FTRLM & & 17.78 $\pm$ .08 & 18.86 $\pm$ .15 & 1 & .5 & .3\\
    \hline
    \end{tabular}
    \caption{Validation and test accuracy for the StackOverflow next word prediction task. Each experiment is run five times to calculate the mean and standard deviation.  Vanilla tree aggregation \citep{Dwork-continual} is used in DP-FTRLM. The momentum variant DP-FTRLM performs better than DP-FTRL.} 
    \label{tab:fl-momentum}
\end{table}

\begin{figure}[!htbp]
\centering
\begin{subfigure}[b]{0.45\textwidth}
\centering
\includegraphics[width=\textwidth]{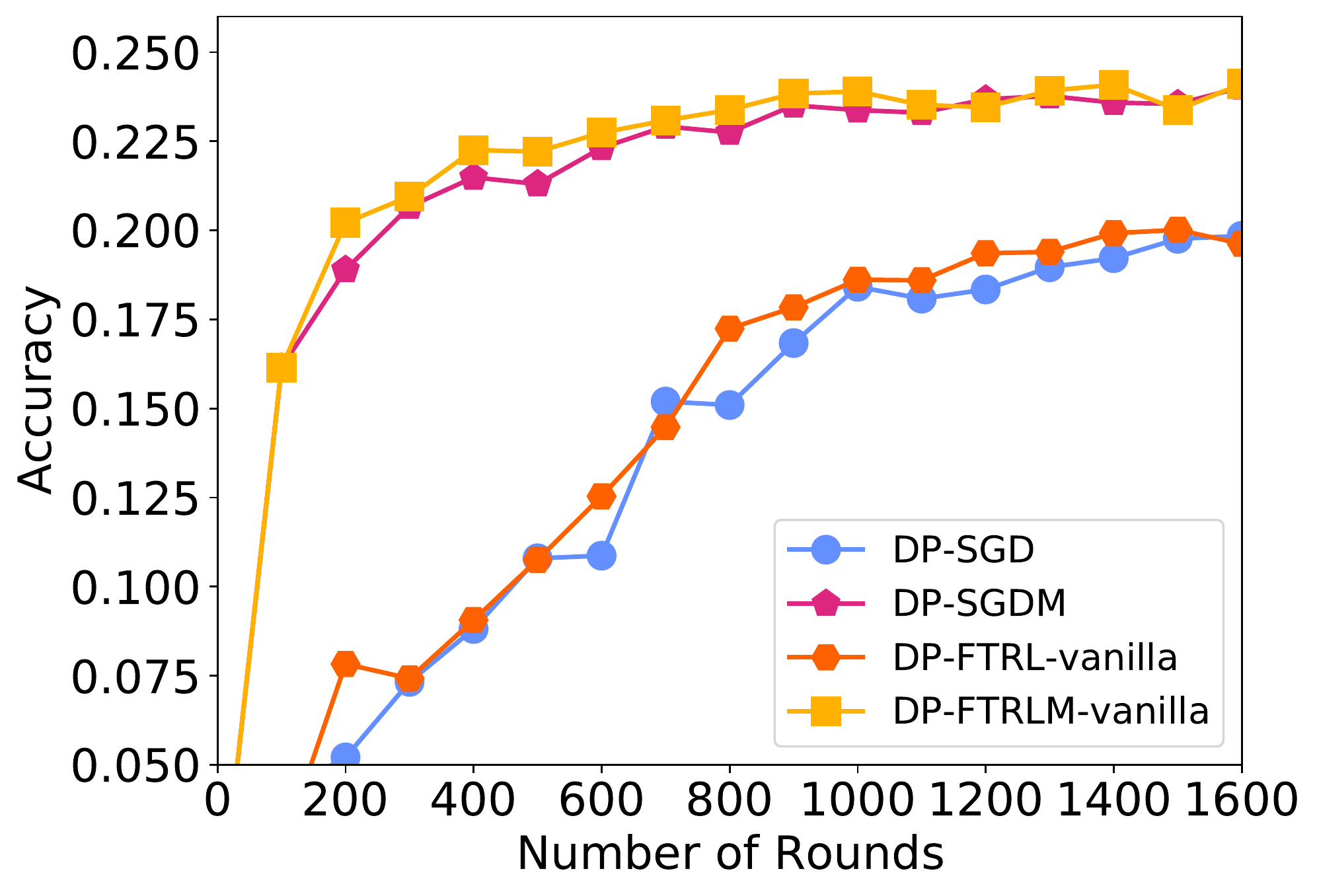}
\caption{Privacy epsilon $\infty$}
\end{subfigure}
\begin{subfigure}[b]{0.45\textwidth}
\centering
\includegraphics[width=\textwidth]{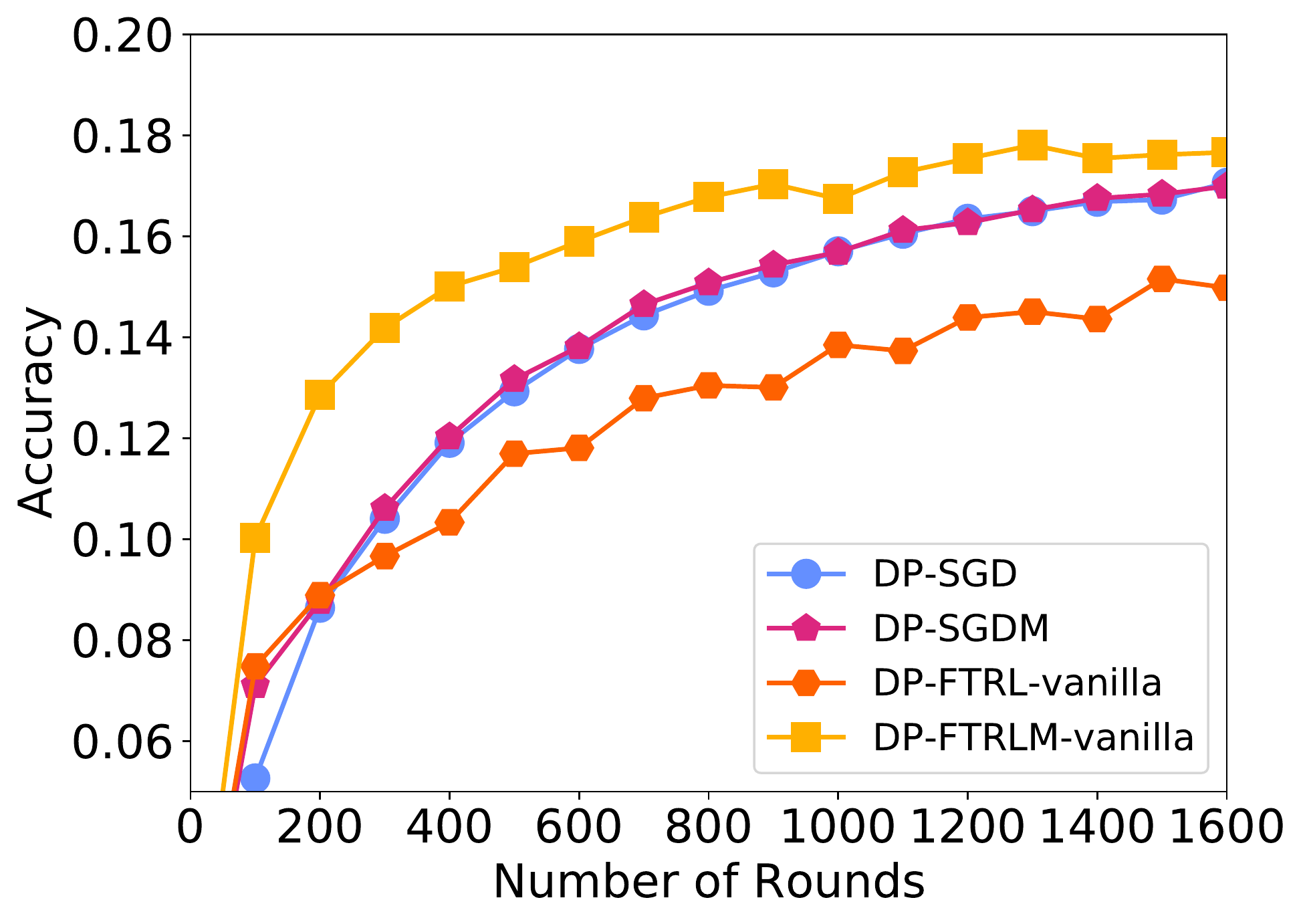}
\caption{Privacy epsilon  $\sim 8.5$}
\end{subfigure}
\caption{Training curves show validation accuracy of StackOverflow. The curve of the best validation accuracy out of the five runs is presented. Vanilla tree aggregation \citep{Dwork-continual} is used in DP-FTRL-vanilla and DP-FTRLM-vanilla. The momentum variant converges faster and performs better. 
}
\label{fig:fl-momentum}
\end{figure}

The momentum variant helps in two ways for StackOverflow: momentum significantly improve the performance of both SGD and FTRL when the noise is relatively small; moreover, momentum stabilizes DP-FTRL when the noise is relatively large. Note that the tree aggregation method in DP-FTRL use different privacy calculation method compared to DP-SGD. A relatively large noise multiplier has to be used to achieve the same privacy $\epsilon$ guarantee. While tree aggregation in DP-FTRL exploits the $O(\log n)$ accumulated noise, it also introduces unstable jump for the noise added in each round, which could be mitigated by the momentum $\gamma$ introduced in DP-FTRLM. In the experiments of StackOverflow, we will always use the momentum variant unless otherwise specified. 

\subsection{Efficient Tree Aggregation}
\mypar{Centralized learning}
Figure~\ref{fig:image efficient} shows a comparison between the efficient (``FTRLM'')~\cite{honaker2015efficient} and the original version (``FTRLM-vanilla'')~\citep{Dwork-continual} of FTRLM for the three centralized example-level DP image classification tasks. We can see clearly that the efficient version always outperforms the vanilla version.
The settings follows from that in Appendix~\ref{app:mom_comp}.

\begin{figure*}[h]
\centering
\begin{subfigure}[b]{0.33\textwidth}
\centering
\includegraphics[width=\textwidth]{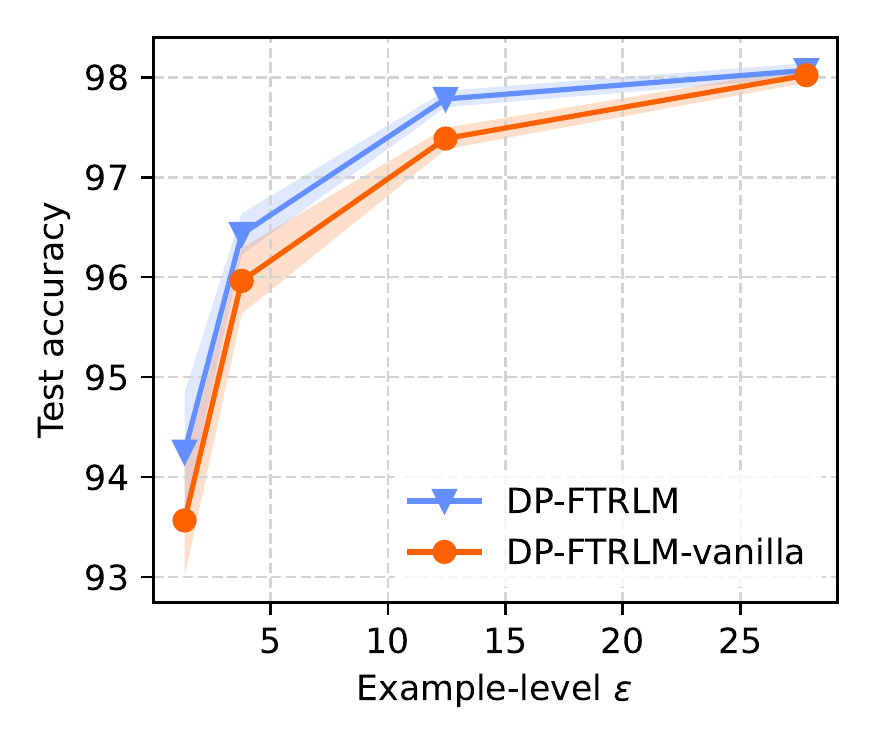}
\end{subfigure}
\begin{subfigure}[b]{0.33\textwidth}
\centering
\includegraphics[width=\textwidth]{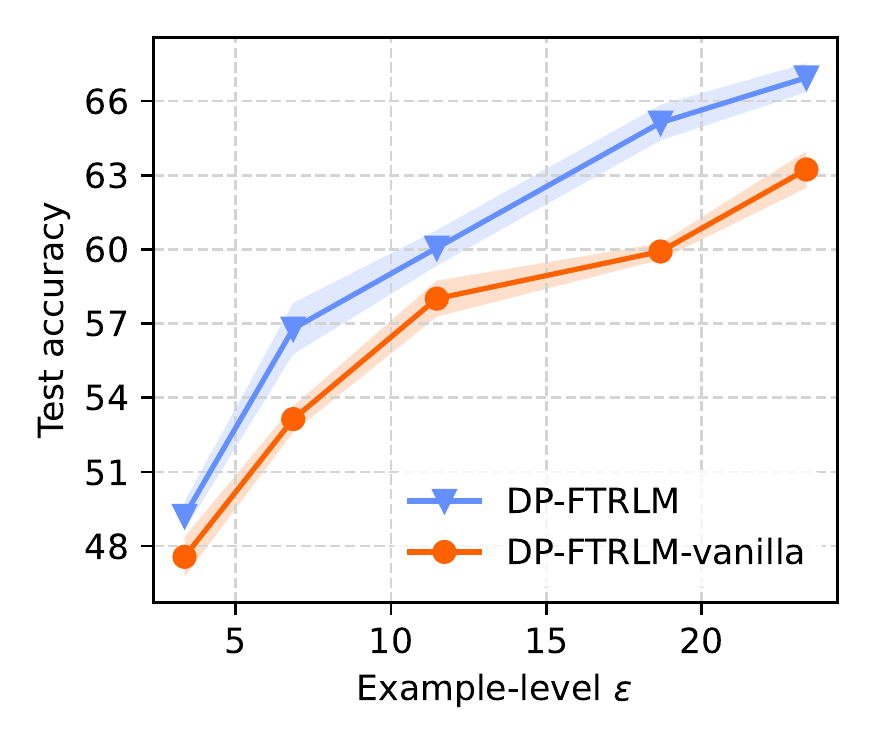}
\end{subfigure}
\begin{subfigure}[b]{0.33\textwidth}
\centering
\includegraphics[width=\textwidth]{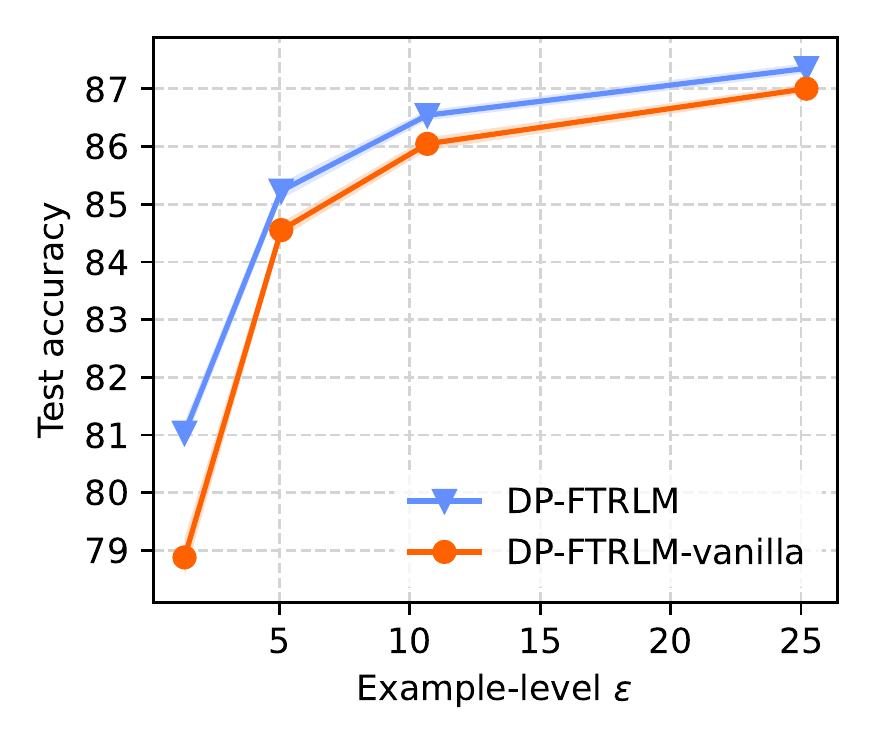}
\end{subfigure}
\caption{Comparison of two variants of DP-FTRL with efficient tree aggregation (``DP-FTRLM'') and vanilla tree aggregation (``DP-FTRLM-vanilla'').
}
\label{fig:image efficient}
\end{figure*}

\mypar{Federated learning}
\cref{fig:fl_results_app} shows the advantage of the efficient tree aggregation algorithm in the StackOverflow simulation for the federated learning setting. In \cref{fig:fl-utility-noise_efficient}, to meet the targeted StackOverflow test accuracies (23\%, 24.5\%), the noise multipliers for DP-FTRLM can increase from (0.268, 0.067) to (0.387, 0.149) after implementing the efficient tree aggregation \citep{honaker2015efficient}. The noise multipliers are used to generate \cref{fig:fl-utility-population-real-efficient}. 

\begin{figure*}[!htbp]
\centering
\begin{subfigure}[b]{0.33\textwidth}
\centering
\includegraphics[width=\textwidth]{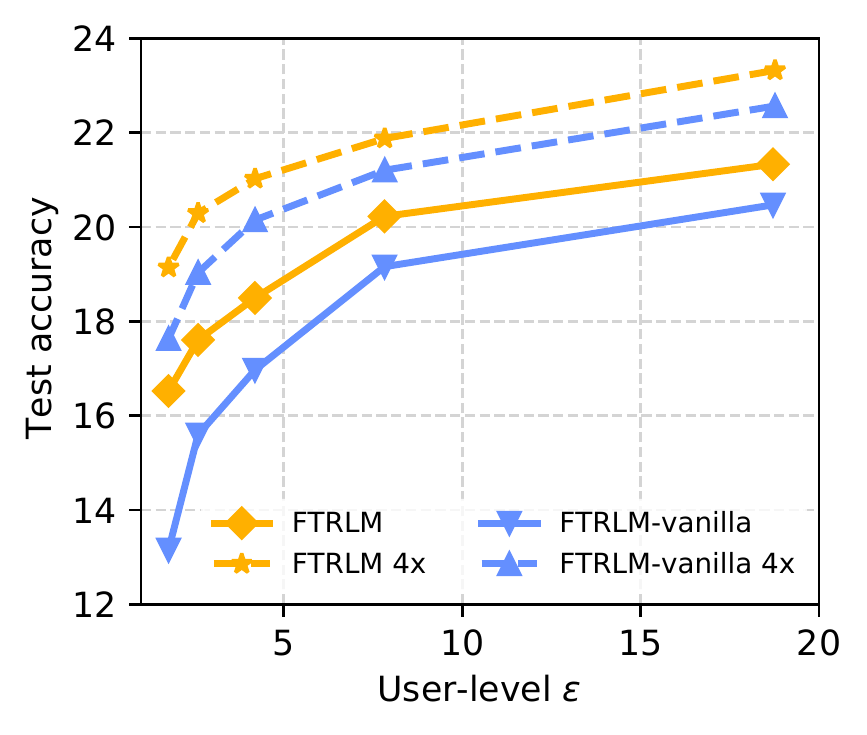}
\caption{}
\label{fig:acc_privacy_stackoverflow_efficient}
\end{subfigure}
\begin{subfigure}[b]{0.33\textwidth}
\centering
\includegraphics[width=\textwidth]{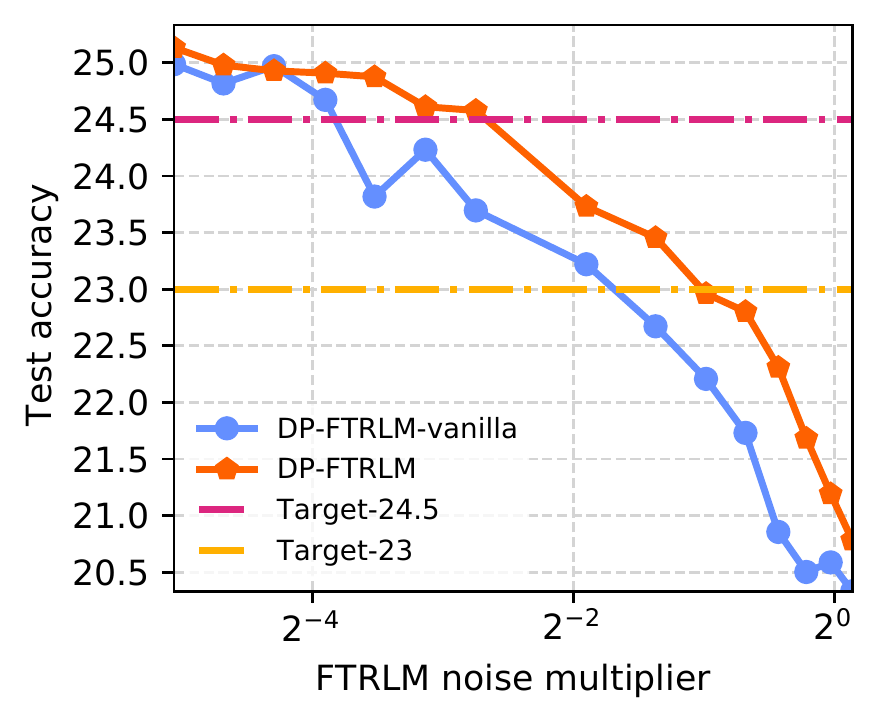}
\caption{}
\label{fig:fl-utility-noise_efficient}
\end{subfigure}
\begin{subfigure}[b]{0.33\textwidth}
\centering
\includegraphics[width=\textwidth]{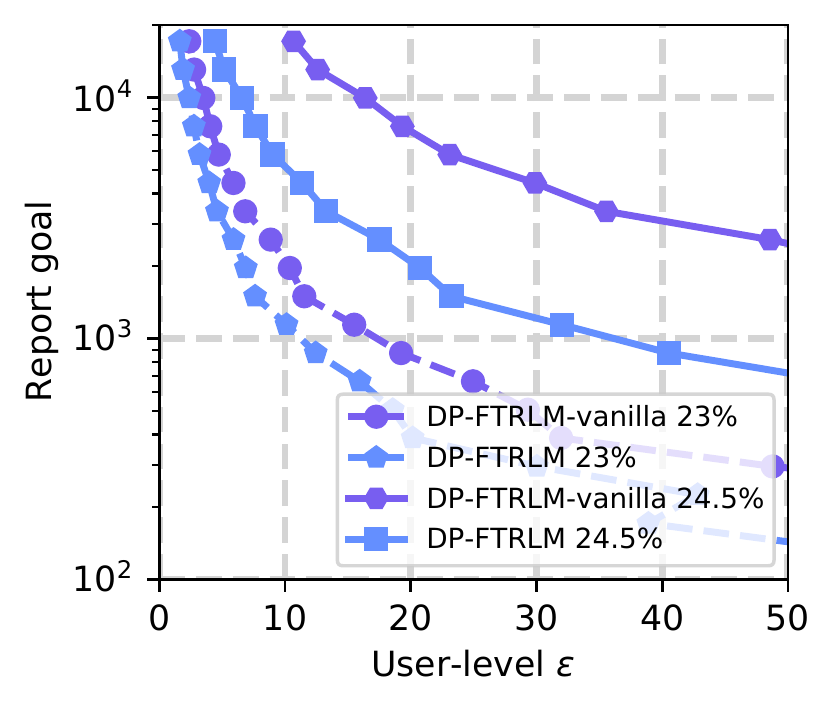}
\caption{}
\label{fig:fl-utility-population-real-efficient}
\end{subfigure}
\caption{Comparison of two variants of DP-FTRL with efficient tree aggregation \citep{honaker2015efficient} and vanilla tree aggregation \citep{Dwork-continual} on StackOverflow for (a) test accuracy under different privacy epsilon;
(b) test accuracy with various noise multipliers;
(c) relationship between user-level privacy $\epsilon$ (when $\delta\approx\nicefrac{1}{\text{population}}$) and computation cost (report goal) for two fixed accuracy targets (see legend).
}
\label{fig:fl_results_app}
\end{figure*}

\subsection{Effect of Tree Completion Trick}
In the centralized learning experiments in Section~\ref{sec:privTarget}, we also make use of the tree completion trick described in Appendix~\ref{sec:tree completion}. Figure~\ref{fig:tree completion trick} plot the comparison between the DP-FTRL result presented in Figure~\ref{fig:acc_privacy} and those without the tree completion trick. We can see that the trick always helps in these settings.

\begin{figure}
\centering
\begin{subfigure}[b]{0.8\textwidth}
\includegraphics[width=\textwidth]{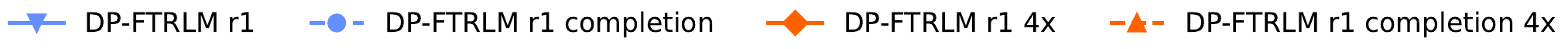}
\end{subfigure}
\begin{subfigure}[b]{0.32\textwidth}
\includegraphics[width=\textwidth]{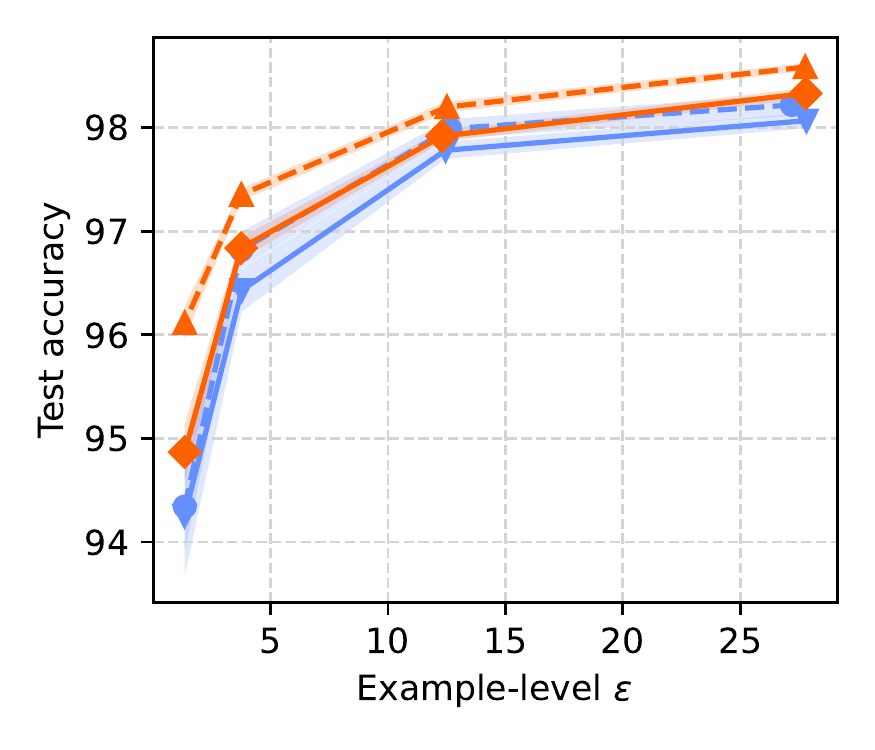}
\caption{MNIST}
\end{subfigure}
\begin{subfigure}[b]{0.32\textwidth}
\includegraphics[width=\textwidth]{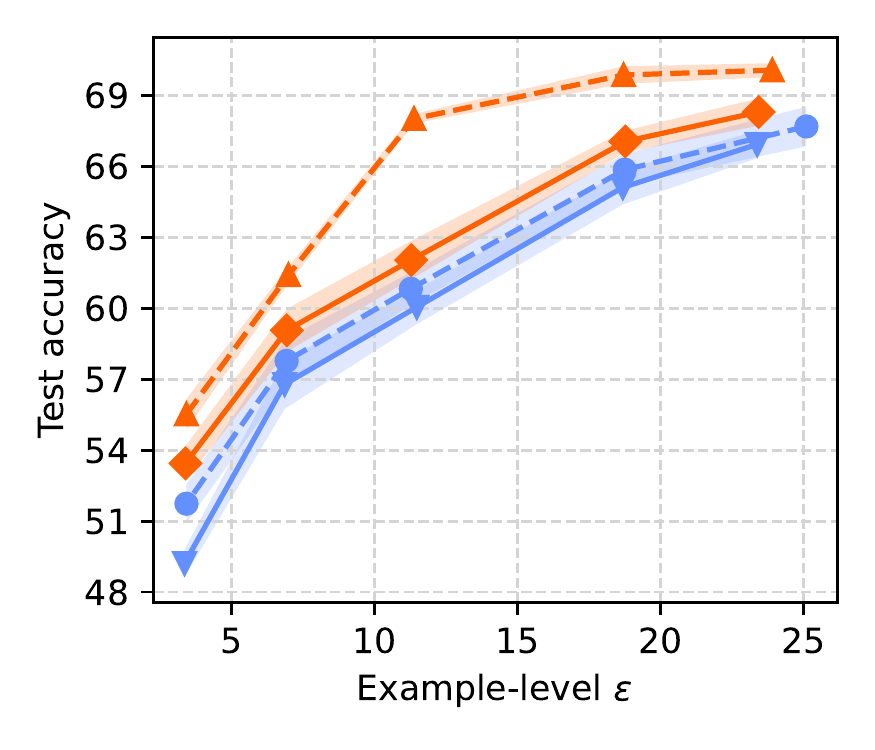}
\caption{CIFAR-10}
\end{subfigure}
\begin{subfigure}[b]{0.32\textwidth}
\includegraphics[width=\textwidth]{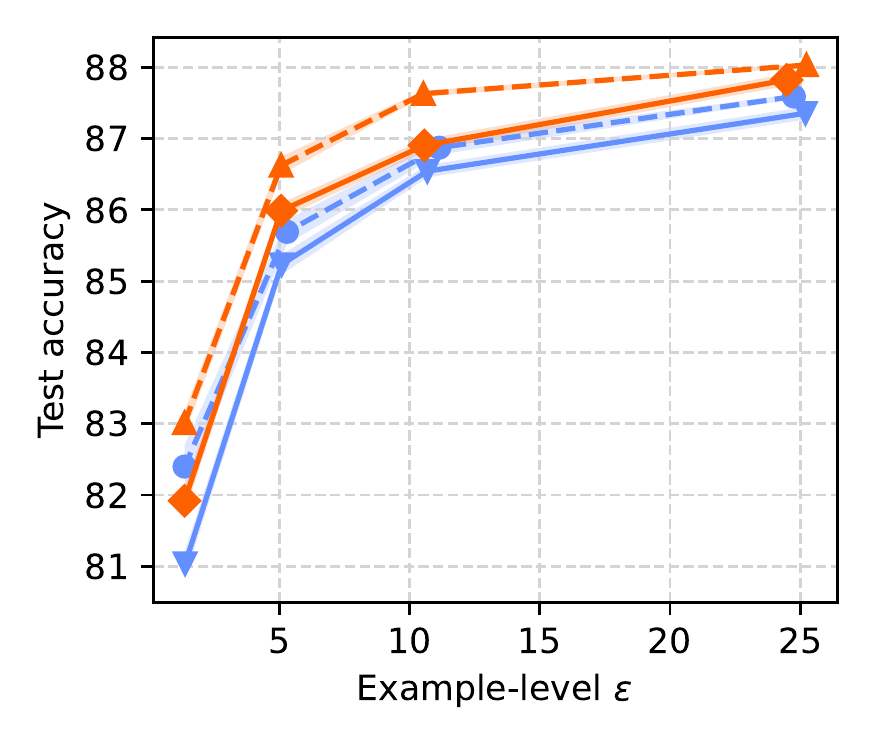}
\caption{EMNIST}
\end{subfigure}
\caption{The effect of the tree completion trick. Settings are the same as those in Figure~\ref{fig:acc_privacy}.}
\label{fig:tree completion trick}
\end{figure}
\section{Omitted Details for Experiments in Section~\ref{sec:privTarget}}
\label{app:privTarget}

\subsection{Details of Hyperparameter Tuning}
\label{app:hyp_tun_privTarget}

\paragraph{Image classification experiments}
For the three image classification experiments, we tune the learning rate ($1/\lambda$ for FTRL) over a grid of the form $\cup_{i\in\{-3,-2,\dots, 3\}} \{10^{i}, 2\times 10^{i}, 5\times 10^{i}\}$, selecting the value that achieves the highest test accuracy averaged over the last 5 epochs while ensuring this chosen value is not an endpoint of the grid.
We use a clipping norm $1.0$ for all the image classification experiments following previous work~\cite{papernot2020making}.

The parameter search for non-private baseline is the same as that for the DP algorithms.  
We use regular SGD (with and without momentum) for the image classification tasks.

\paragraph{StackOverflow experiments}

The StackOverflow benchmark dataset of the next word prediction task has 342,477 users (clients) with training 135,818,730 examples. A validation set of 10,000 examples, and a test set of 16,576,035 examples are constructed following \cite{reddi2020adaptive}. The one layer LSTM described in \citep{reddi2020adaptive} is used. We compare with DP-FedAvg where DP-SGD is used on server.

There are many hyperparameters in federated learning. We fix the number of total rounds to be 1,600 for StackOverflow, and sample 100 clients per round for DP-SGD, and take 100 clients from the shuffled clients for DP-FTRL to make sure the clients are disjoint across rounds. Note that DP-FTRL would run less than one epoch for StackOverflow. On each client, the number of local epochs is fixed to be one and the batch size is sixteen, and we constrained the maximum number of samples on each client to be 256. The momentum for both DP-SGDM and DP-FTRLM is fixed to 0.9.

In most of the experiments,  we will tune server learning rate, client learning rate and clip norm for a certain noise multiplier. 
We tune a relative large grid (client learning rate in $\{0.1, 0.2, 0.5, 1, 2\}$, server learning rate in $\{0.03, 0.1, 0.3, 1, 3\}$, clip norm in $\{0.1, 0.3, 1, 3, 10\}$) when the noise multiplier is zero. And we have several observation: the best accuracy of clip norm 0.3 and 1.0 are slightly better than larger clip norms, which suggests that clip norm could generally help for this language task; increasing server learning rate could complement decreasing clip norm when clip norm is effective; the largest client learning rate that does not diverge often leads to good final accuracy.
As adding noise increases the variance of gradients, we often have to decrease learning rate in practice. Based on this heuristic and the observation from tuning when noise multiplier is zero, we choose client learning rate from $\{0.1, 0.2, 0.5\}$, server learning rate from $\{0.1, 0.3, 1, 3\}$ and clip norm from $\{0.3, 1, 3\}$ unless otherwise specified. 
We use DP-SGD with zero noise for StackOverflow, as gradient clipping can improves accuracy for language tasks.

\subsection{Centralized Training with Large Number of Epochs by Interleaving Restarting and Non-restarting}
\label{sec:interleaving}

Appendix~\ref{sec:DP-FTRL-SometimesRestart} describes the idea of interleaving between restarting and non-restarting. Here, we examine how such schedules might affect the model utility. Additionally, we consider the effect of the tree completion trick (Section~\ref{sec:tree completion}).
We consider CIFAR-10 and EMNIST, which are hard datasets that might require a large number of epochs to learn. 
For CIFAR-10, we fix the batch size to be $500$, number of epochs to be $100$ (thus $10000$ steps in total);
for EMNIST, we fix the batch size to be $500$ and number of epochs to be $50$ (thus $69750$ steps in total). 
On each dataset, similar as in Section~\ref{sec:empEval}, we compare DP-SGD with or without amplification, and DP-FTRL(M) with different restarting schedules.

\begin{figure}
\centering
\begin{subfigure}[b]{0.45\textwidth}
\includegraphics[width=\textwidth]{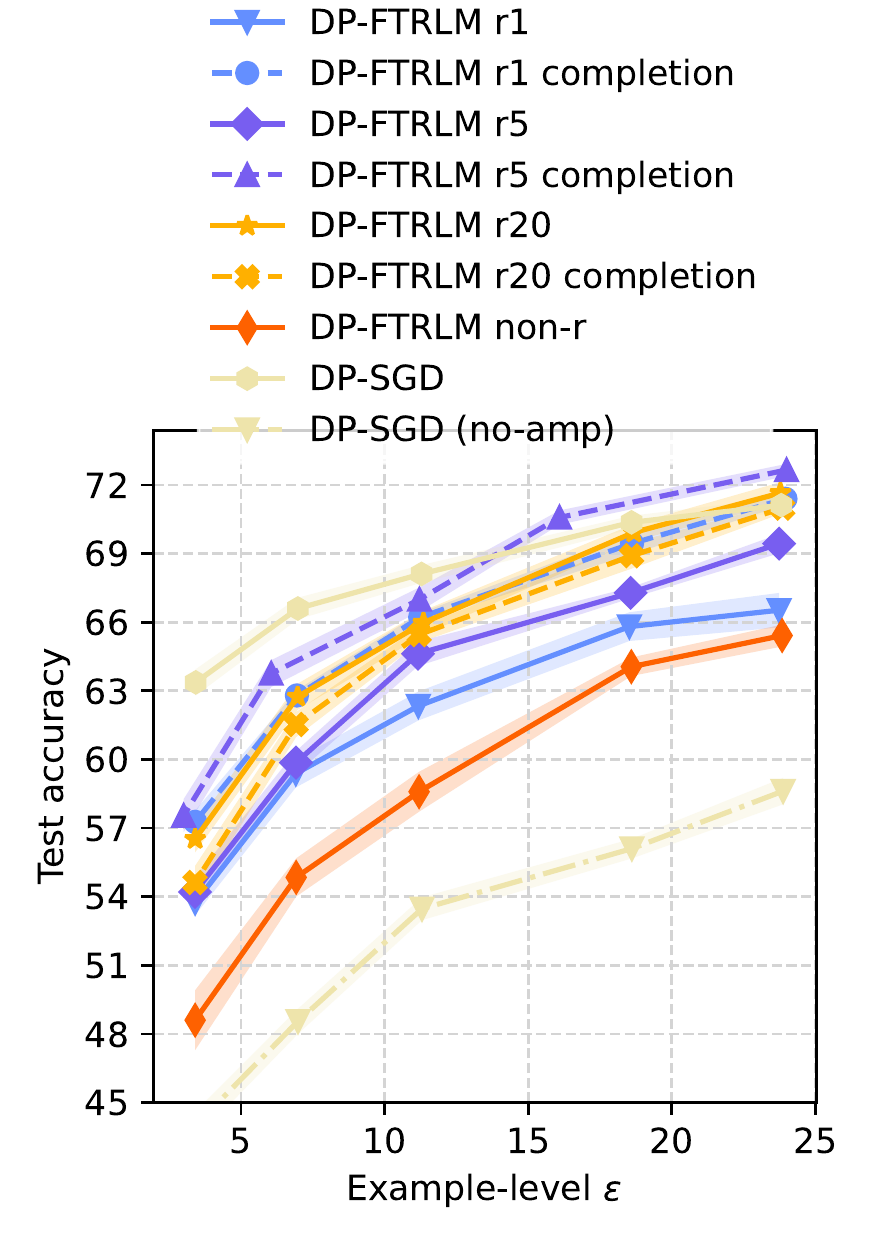}
\caption{CIFAR-10, $100$ epochs.}
\label{fig:cifar10 different restart}
\end{subfigure}
\begin{subfigure}[b]{0.45\textwidth}
\includegraphics[width=\textwidth]{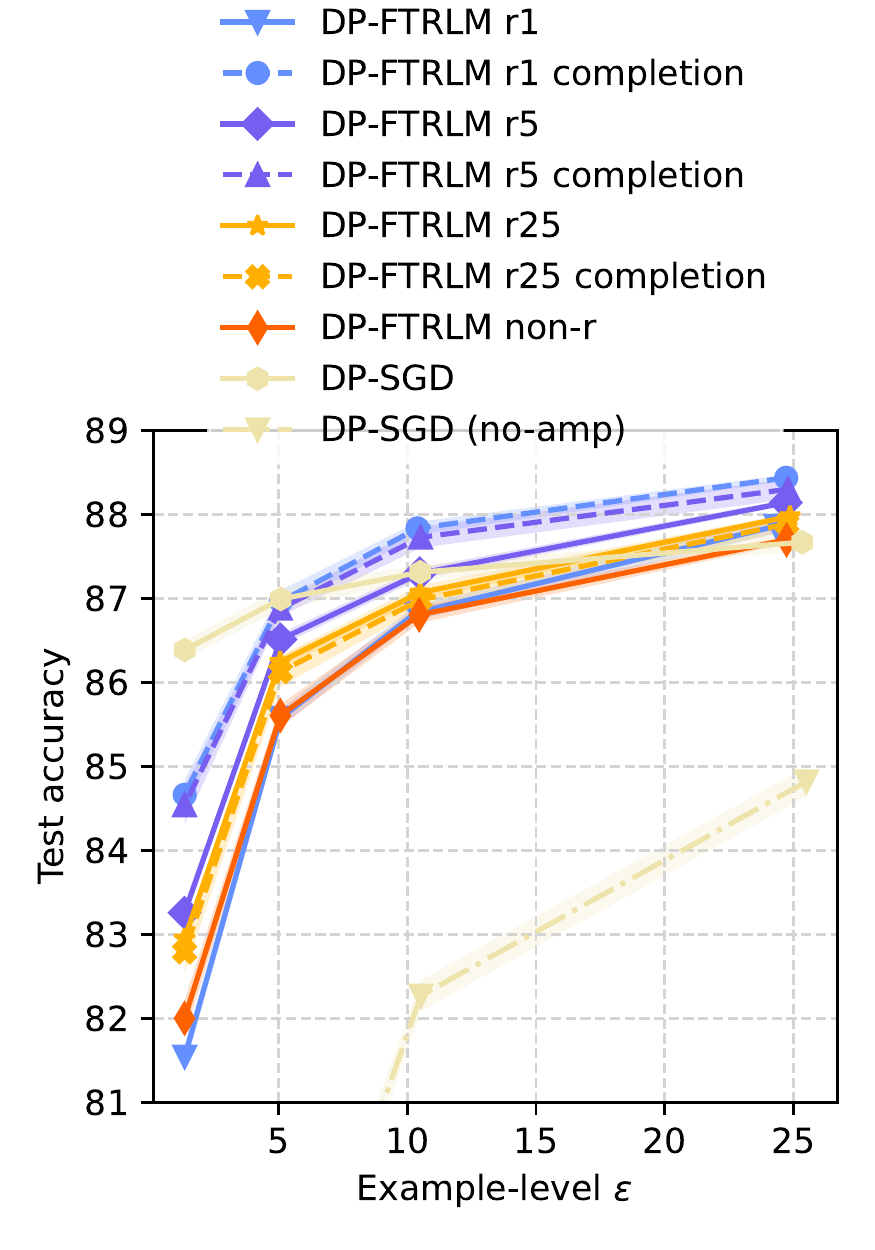}
\caption{EMNIST, $50$ epochs.}
\label{fig:emnist different restart}
\end{subfigure}
\caption{Interleaving between restarting and non-restarting. {\bf r1}, {\bf r5}, {\bf r20}, and {\bf r25} corresponds to restarting every one, five, twenty, and twenty five epochs respectively. {\bf non-r} corresponds to the version of DP-FTRL with no restarting. {\bf completion} refers to the ``completion trick''  to the closest power of two(from Section~\ref{sec:tree completion}).}
\label{fig:different restart}
\end{figure}

In Figure~\ref{fig:cifar10 different restart}, we plot the results for DP-FTRL(M) with non-restarting and restarting every $1$, $5$, $20$ epochs (solid lines), and comparing them with DP-SGD with and without amplification. Additionally, we use the tree completion trick for restarting every $5$ and $20$ epochs (dashed lines).
We can see the following.
\begin{itemize}
\item Neither non-restarting nor restarting every epoch yields accuracy that are comparable to DP-SGD. On the other hand, without the tree completion trick, restarting every $20$ epochs gives much better accuracy, which means that interleaving between restart and non-restart is crucial.
\item The tree completion trick helps for restarting every $5$ epochs and hurts for restarting every $20$ epochs, demonstrating the trade-off we mentioned before.
\item Overall, without the tree completion trick, restarting every $5$ epochs with the tree completion trick gives the best accuracy, and we can see a ``cross-over'' between it and DP-SGD similar as that in Figure~\ref{fig:acc_privacy}, yet at a larger $\epsilon\approx 18$.
With the completion trick, restarting every $20$ epochs gives the best accuracy. A ``cross-over'' happens at $\epsilon \approx 14$.

\item Similar as in Figure~\ref{fig:acc_privacy}, we can see that DP-FTRL is always better than DP-SGD without amplification.
\end{itemize}

In Figure~\ref{fig:emnist different restart}, we plot the results on EMNIST for restarting every $1$, $5$, $25$ epochs and non-restarting. We can observe similar trend as in the CIFAR-10 experiments.
Namely, the tree completion tricks helps in some cases, and the best accuracy is achieved by restarting every $5$ epochs with the tree completion trick, which outperforms DP-SGD with amplification starting from $\eps \approx 5$.

\subsection{Omitted Details for StackOverflow Experiments}
\label{app:expt_so_comp_plot}
\begin{figure}[ht]
\centering
\begin{subfigure}[t]{0.45\textwidth}
\centering
\includegraphics[width=\textwidth]{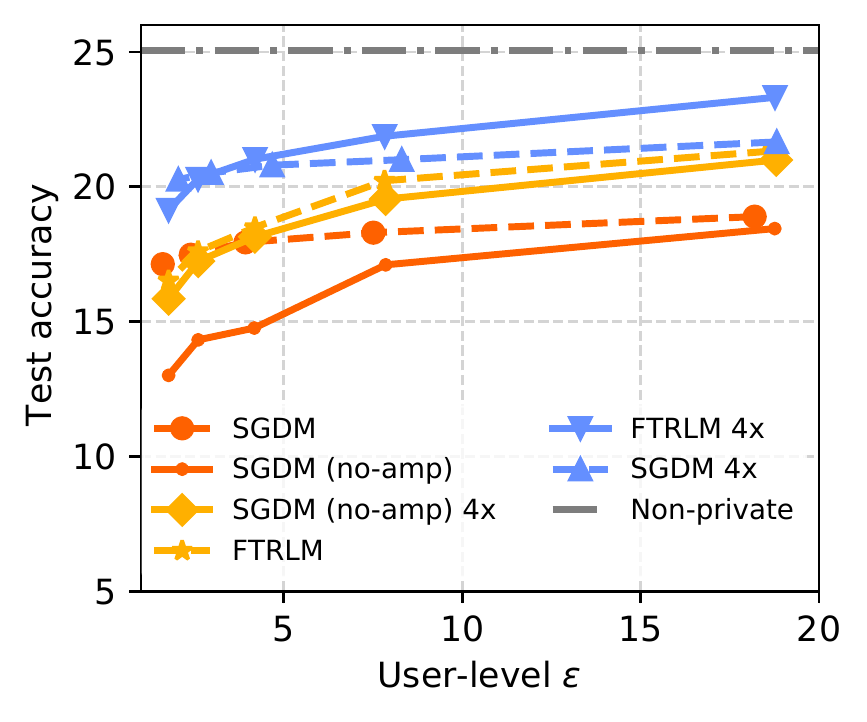}
\label{fig:acc_privacy_stackoverflow_sgd_amp}
\end{subfigure}
\begin{subfigure}[t]{0.45\textwidth}
\centering
\includegraphics[width=\textwidth]{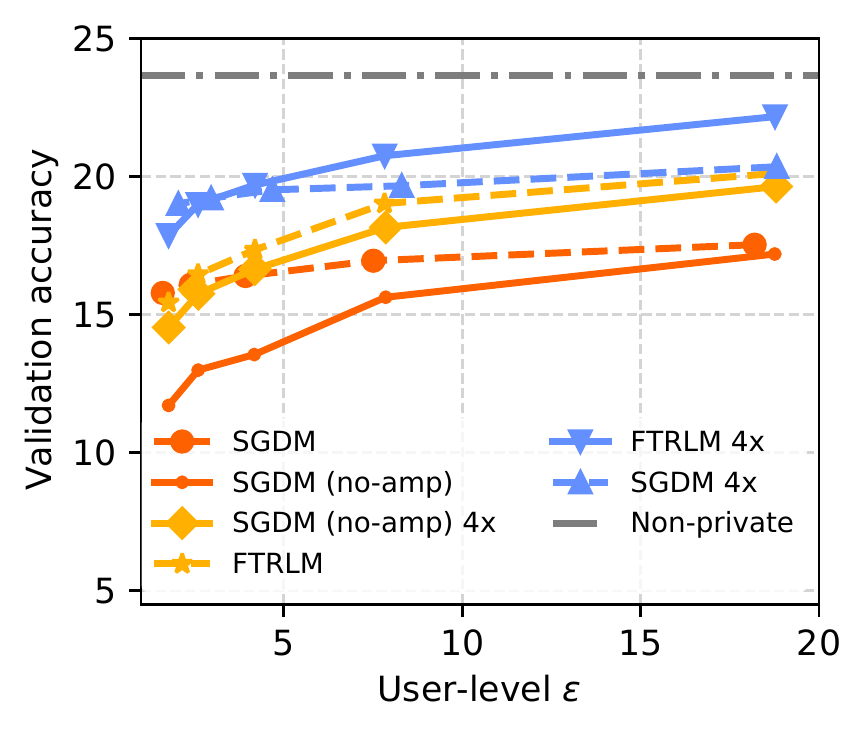}
\end{subfigure}
\caption{Test and Validation accuracy for the StackOverflow next word prediction task under different privacy epsilon by varying noise multipliers.}
\label{fig:fl-privacy-val}
\end{figure}

\begin{table}[ht]
    \centering
    \begin{tabular}{|c|c|c|c|S|S|S|S|}
    \hline
     \multirow{2}{*}{Server Optimizer} & \multirow{2}{*}{Epsilon} & \multicolumn{2}{c|}{Accuracy} & \multicolumn{4}{c|}{Hyperparameters} \\ 
     \cline{3-8} & & Validation  & Test & Noise & \text{ServerLR} & \text{ClientLR} & Clip  \\ 
    \hline
    DP-SGDM & 18.20 & 17.52 & 18.89 & .3 & 1 & .5 & .3 \\
    DP-FTRLM & 18.71 & 20.10 & 21.33 & 1.13 & .3 & .5 & 1\\
    \hline
    DP-SGDM & 7.51 & 16.94 & 18.30 & .4 & .1 & .5 & 1 \\
    DP-FTRLM & 7.83 & 19.01 & 20.22 & 2.33 & 1 & .5 & .3\\
    \hline 
     DP-SGDM & 3.93 & 16.39 & 17.94 & .5 & .3 & .5 & .3 \\
    DP-FTRLM & 4.19 & 17.34 & 18.49 & 4.03 & .1 & .5 & 1\\
    \hline 
     DP-SGDM & 2.40 & 16.08 & 17.48 & .6 & .3 & .5 & .3 \\
    DP-FTRLM & 2.60 & 16.45 & 17.60 & 6.21 & .1 & .5 & 1 \\
    \hline 
    DP-SGDM & 1.61 & 15.78 & 17.13 & .7 & .3 & .5 & .3 \\
    DP-FTRLM & 1.77 & 15.43 & 16.52 & 8.83 & .3 & .5 & .3 \\
    \hline 
    \end{tabular}
    \caption{Validation and test accuracy for the StackOverflow next word prediction task under different privacy epsilon.} 
    \label{tab:fl-privacy}
\end{table}

We compare the accuracy of the momentum variant of DP-FTRL with the momentum variant of DP-SGD as baseline under different privacy epsilon. We tune hyperparameters as described in \cref{app:hyp_tun_privTarget} and select the hyperparameters achieve the best validation accuracy for StackOverflow (see \cref{tab:fl-privacy} and \cref{fig:fl-privacy-val}). DP-FTRLM performs better than DP-SGDM when the epsilon is relatively large, but performs worse when the epsilon is small ($\epsilon < 2.60$ in \cref{tab:fl-privacy}). 
More noise are added to DP-FTRLM to achieve the same privacy epsilon as DP-SGDM. However, DP-FTRLM can result in utility (accuracy) not (much) worse than DP-SGDM without relying on amplification by sampling, which makes it appealing for practical federated learning setting where population and sampling is difficult to estimate \citep{balle2020privacy}. Note that the noise added for both DP-FTRLM and DP-SGDM are considered large for federated learning tasks. The effective noise could be significantly reduced by sampling more clients each round in practice \cite{mcmahan2017learning}, and more discussion on this front is in \cref{app:utilTarget}.
\section{Omitted Details for Experiments in Section~\ref{sec:utilTarget}}
\label{app:utilTarget}

\subsection{Details of Hyperparameter Tuning}
\label{app:hyp_tun_utilTarget}

In \cref{app:expt_so_comp_plot}, a significant amount of noise has to be added in both DP-FTRLM and DP-SGDM to achieve nontrivial privacy epsilons, which leads to undesired accuracy degradation. For example, the test accuracy of DP-FTRLM on StackOverflow dataset decreases from $25.15\%$ when $\epsilon=\infty$ to $20.22\%$ when $\epsilon=8.5$ when the number of clients per round is fixed at 100. In practical federated learning tasks, the total population is very large and many more clients could be sampled every round.  In this section, taking StackOverflow as an example, we study the minimum number of sampled clients per round (report goal in \citep{bonawitz2019towards}) to achieve a target accuracy under certain privacy budget.

\begin{figure}[ht]
\centering
\begin{subfigure}[b]{0.45\textwidth}
\centering
\includegraphics[width=\textwidth]{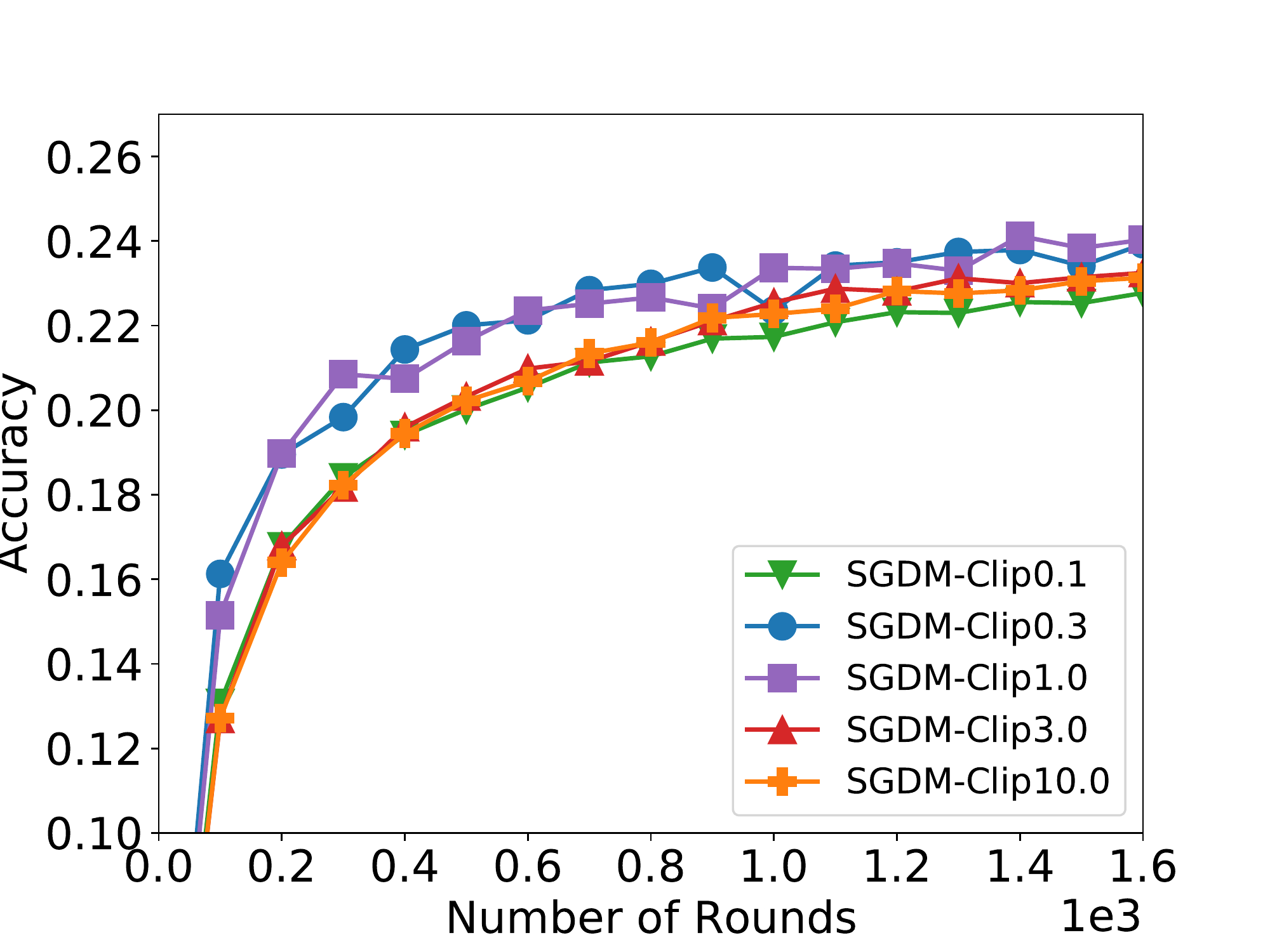}
\caption{SGDM with zero noise}
\end{subfigure}
\begin{subfigure}[b]{0.45\textwidth}
\centering
\includegraphics[width=\textwidth]{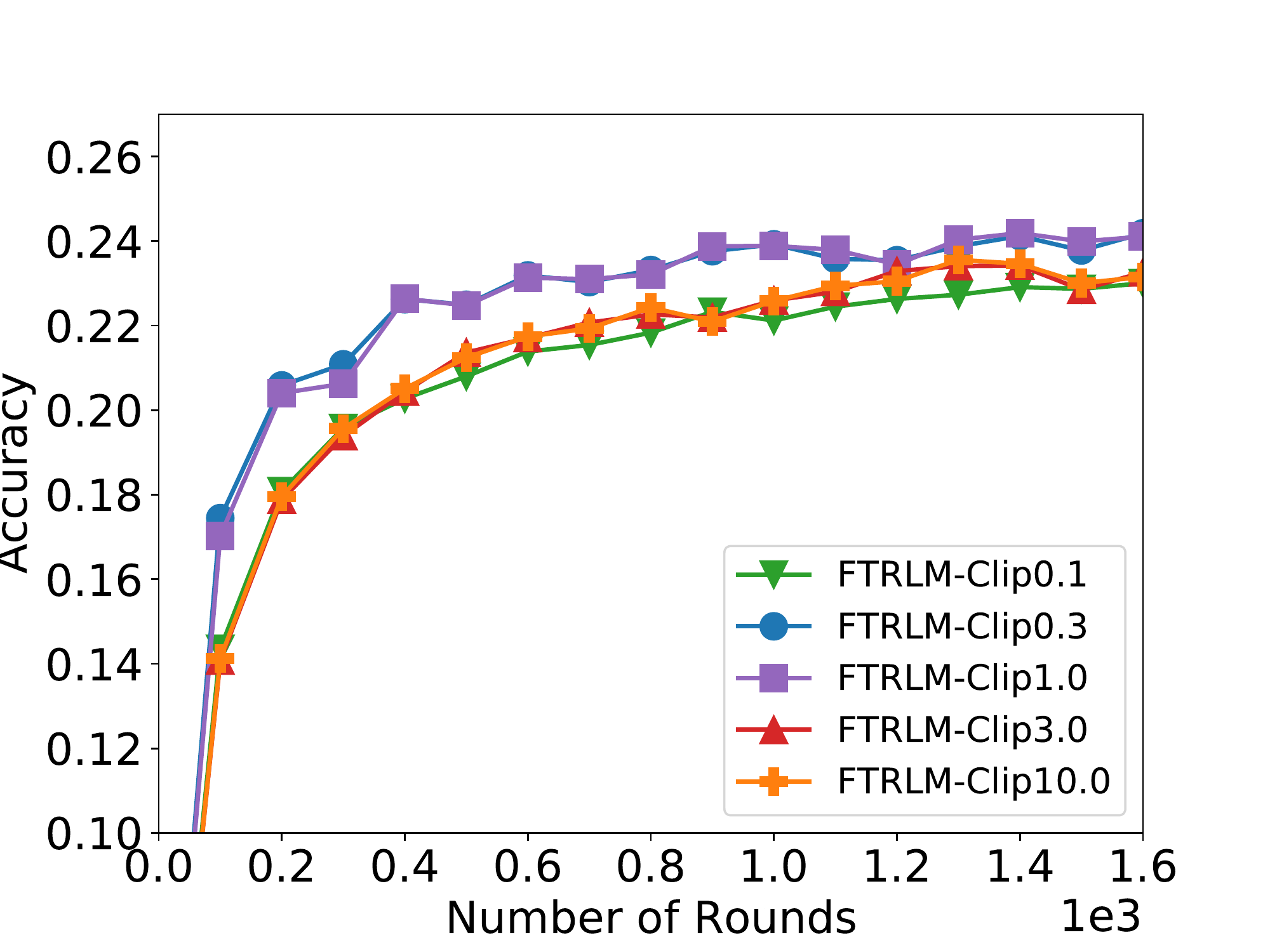}
\caption{FTRLM with zero noise}
\end{subfigure}
\caption{Training curve of the best validation accuracy under various clip norm for StackOverflow.}
\label{fig:fl-utility-clip}
\end{figure}

\paragraph{Fix the clip norm and client learning rate to reduce hyperparameter tuning complexity.} 
We first find the largest noise multiplier that would meet the target accuracy based on selecting 100 clients per round. As an extensive grid search over noise multiplier while simultaneously tuning server learning rate, client learning rate and clip norm is computationally intensive, we fix the clip norm to 1 and the client learning rate to 0.5 based on \cref{fig:fl-utility-clip}. We then tune the server learning rate from $\{0.3, 1, 3 \}$ for each noise multiplier.

\paragraph{Grid search for the largest noise multiplier to meet the target.} We use a grid of ten noise multipliers between $0$ ($\epsilon=\infty$, test accuracy=$24.89$) and $0.3$ ($\epsilon=18.89$, test accuracy=$18.89$) for DP-SGDM, and between $0$ ($\epsilon=\infty$, test accuracy=$25.15$) and $1.13$ ($\epsilon=19.74$, test accuracy=$21.33$) for DP-FTRLM. And we further add five noise multipliers between $0$ and $0.035$ for DP-SGDM, and between $0$ and $0.149$ for DP-FTRLM based on the results of the previous grid search on ten noise multipliers. 
The test accuracy is presented in \cref{fig:fl-utility-noise}. We set the target test accuracy as $24.5\%$ and select noise multiplier $0.007$ (with server learning rate $3$) for DP-SGDM and noise multiplier $0.149$ (with server learning rate $3$) for DP-FTRLM. 

\paragraph{Report goal for the nontrivial privacy epsilon in practice.} The standard deviation of noise added in each round is proportional to the inverse of the number of clients per round (report goal). The practical federated learning tasks often have a very large population and report goal, and we could simultaneously increase the noise multiplier and report goal, so that the utility (accuracy for classification and prediction tasks) will likely not 
degrade \cite{mcmahan2017learning} while the privacy guarantee is improved.  The validation accuracy of simulation performance with two different report goals for StackOverflow is presented in \cref{fig:fl-utility-clients-restarts}. The noise multiplier 0.149 is used for DP-FTRLM and 0.007 is used for DP-SGD when report goal is 100, which is the largest noise multiplier to meet the target test accuracy determined by \cref{fig:fl-utility-noise}. We run each experiment for five times and plot the curves for the median validation accuracy, the corresponding test accuracy are $24.73\%$ for DP-SGDM and $24.51\%$ for DP-FTRLM. We then run the same experiments with report goal of 1000,
and proportionally increase the corresponding noise multiplier to be 1.49 for DP-FTRLM and 0.07 for DP-SGDM. The performance of 1000 report goal is slightly better with test accuracy $25.19\%$ for DP-SGDM and $24.67\%$ for DP-FTRLM. We will assume the utility will not decrease if report goal and noise multiplier are simultaneously and proportionally increased.

As shown in \cref{tab:fl-utility}, both report goals 100 and 1000 would provide trivial privacy guarantee of large epsilon for the target utility. We have to increase the report goal to $2.06e4$ to get a nontrivial privacy epsilon (less than 10) with DP-FTRLM and the StackOverflow population of $3.42e5$ \footnote{The best epsilon DP-SGDM can achieve is $10.16$ by increasing report goal to be as large as the population $3.42e5$}. 
Smaller report goal could achieve similar privacy guarantee if the population becomes larger. In \cref{fig:fl-utility-population-real}, the relationship between privacy guarantee and report goal for DP-FTRLM and DP-SGDM are presented. DP-FTRLM provides better privacy guarantee by smaller report goal when the privacy epsilon is relatively large or very small. The range where DP-FTRLM outperforms DP-SGDM in report goals and privacy guarantees are larger when the population is relatively small or very large.

\begin{figure}[ht]
\centering
\includegraphics[width=0.6\textwidth]{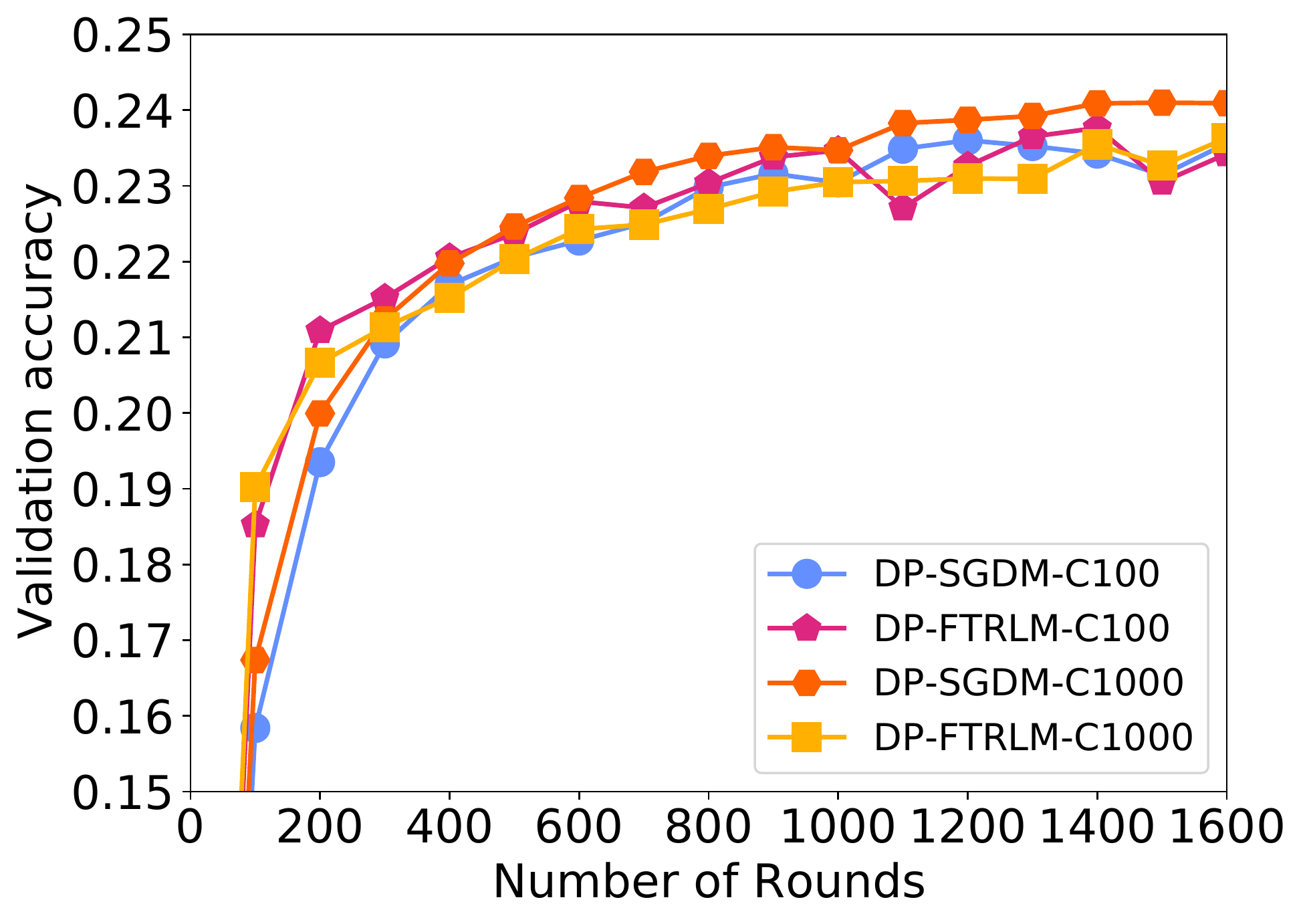}
\caption{Training curves of validation accuracy for DP-SGDM and DP-FTRLM for StackOverflow for report goal 100 and 1000 (suffix C100 and C1000 in the legend). DP-FTRLM with restart (see \cref{app:multipass}) is used when report goal is 1000 (less than five epochs of data). Simultaneously increasing noise multiplier and report goal by 10x could significantly improve the privacy guarantee without sacrificing the utility. The noise multiplier for DP-SGDM-C100, DP-FTRLM-C100, DP-SGDM-C1000, DP-FTRLM-C1000 are 0.007, 0.149, 0.07, and 1.49, respectively. The corresponding test accuracy are 24.73\%, 24.51\%, 25.19\% and 24.67\%. The corresponding privacy $\epsilon$ can be found in \cref{tab:fl-utility}}. \label{fig:fl-utility-clients-restarts}
\end{figure}

\begin{table}[ht]
    \centering
    \begin{tabular}{|c|c|c|c|c|c|c|}
    \hline
     \multirow{2}{*}{Server Optimizer} & \multicolumn{2}{c|}{Privacy} & \multicolumn{3}{c|}{Setting} \\ 
     \cline{2-6} & Epsilon  & Delta & Noise & Report goal & Population  \\ 
    \hline
    DP-SGDM &  1.78e7 & 1e-6 & 0.007 & 100 & 3.42e5\\
    DP-FTRLM & 363.66 & 1e-6 & 0.149 & 100 & 3.42e5\\
    DP-SGDM  & 7.71e4 & 1e-6 & 0.07 & 1000 & 3.42e5\\
    DP-FTRLM & 32.52 & 1e-6 & 1.49 & 1000 & 3.42e5\\
    \hline 
    DP-SGDM  & 9.42 & 1e-6 & 23.97 & 3.42e5 & 3.42e5\\
    DP-FTRLM & 9.76 & 1e-6 & 7.53 & 5.06e3 & 3.42e5\\
    DP-FTRLM & 4.11 & 1e-6 & 24.29 & 1.63e4 & 3.42e5\\
    \hline 
    DP-SGDM & 8.11 & 1e-6 & 0.67 & 9.56e3 & 1e6\\
    DP-FTRLM & 7.57 & 1e-6 & 5.73 & 3.85e3 & 1e6\\
    DP-SGDM & 3.70 & 1e-6 & 1.20 & 1.71e4 & 1e6\\
    DP-FTRLM & 3.56 & 1e-6 & 15.29 & 1.03e4 & 1e6\\
    \hline 
    \end{tabular}
    \caption{The ($\epsilon, \delta$) privacy guarantee for DP-FTRLM and DP-SGDM under realistic and hypothetical report goal and population of StackOverflow that would meet the target test accuracy 24.5\%. Note that the DP-FTRLM privacy accounting is based on the restart strategy in \cref{app:multipass}.} 
    \label{tab:fl-utility}
\end{table}

\fi
\end{document}